\documentclass[sigconf,nonacm]{acmart}
\AtBeginDocument{%
  \providecommand\BibTeX{{%
    \normalfont B\kern-0.5em{\scshape i\kern-0.25em b}\kern-0.8em\TeX}}}

\usepackage{subcaption}

\usepackage{hyperref}
\usepackage{amsmath,mathtools}
\usepackage{mathrsfs}
\usepackage{stmaryrd}
\usepackage{paralist}
\usepackage{xspace}
\usepackage{xcolor}
\usepackage{breqn}
\usepackage{mathtools}
\usepackage{tabularx}
\usepackage{balance}

\usepackage{textcomp}
\usepackage[final]{listings}
\usepackage[shortlabels]{enumitem}

\newcommand\LP[1]{\textcolor{red}{#1}}

\renewcommand{\setminus}{-}

\newcommand{\eat}[1]{}
\newcommand{\OMIT}[1]{}

\newcommand{\LRA}{\Leftrightarrow}
\newcommand{\cS}{\mathcal{S}}

\newcommand{\so}{\mathsf{o}}
\newcommand{\sn}{\mathsf{n}}

\newcommand{\barN}{\bar{N}}
\newcommand{\barM}{\bar{M}}
\newcommand{\bara}{\bar{a}}

\newcommand{\brt}{\bar{t}}

\newcommand{\df}{:=}

\newcommand{\NULL}{\sqlkwblack{NULL}}

\newcommand{\sN}{\mathbf{Name}}

\newcommand{\Null}{\{\NULL\}}
\newcommand{\Num}{\mathbf{Num}}
\newcommand{\Const}{\mathbf{Val}}
\newcommand{\notnull}[1]{#1_{\neq{\scriptscriptstyle \NULL}}}

\renewcommand{\epsilon}{\varepsilon}

\newcommand{\trv}[1]{{\normalfont\textbf{#1}}}
\newcommand{\TV}{\trv{T}}
\newcommand{\vt}{{\normalfont\textbf{t}}}
\newcommand{\vf}{{\normalfont\textbf{f}}}
\newcommand{\vs}{{\normalfont\textbf{s}}}

\newcommand{\vu}{{\normalfont\textbf{u}}}

\newcommand{\Conn}{\Gamma}
\newcommand{\conn}{\gamma}

\newcommand{\true}{\vt}
\newcommand{\false}{\vf}
\newcommand{\unknwn}{\vu}

\newcommand{\tv}{\pmb{\tau}}

\newcommand{\App}{\texttt{Apply}} 
 
\newcommand{\Count}{\texttt{Count}} 
\newcommand{\Group}{\texttt{Group}}

\newcommand{\Name}{\texttt{Name}}

\newcommand{\wfe}{\mathit{WFE}}
\newcommand{\wfc}{\mathit{WFC}}

\newcommand{\Card}{\mathsf{Card}}

\newcommand{\type}{\mathsf{type}}

\newcommand{\isnul}{\mathtt{isnull}}
\newcommand{\isnotnul}{\mathtt{notnull}}
\newcommand{\isempty}{\mathtt{empty}}

\newcommand{\all}{\mathtt{all}}
\newcommand{\any}{\mathtt{any}}

\newcommand{\eq}{\doteq}
\newcommand{\noteq}{\,/\!\!\!\doteq}

\newcommand{\semsql}[1]{{\llbracket{#1}\rrbracket}}

\newcommand{\op}{\, \mathsf{op} \,}

\newcommand{\semvl}[1]{{\llbracket{#1}\rrbracket}}

\newcommand{\semstrvl}[1]{{\llbracket{#1}\rrbracket}}
\newcommand{\utof}{{\scriptscriptstyle \pmb{\textsc{2VL}}}}

\newcommand{\semtwovl}[1]{{\llbracket{#1}\rrbracket}^{\utof}}

\newcommand{\semthreevl}[1]{{\llbracket{#1}\rrbracket}}
\newcommand{\semtwovlgen}[1]{{\llbracket{#1}\rrbracket}^{\star}}
\newcommand{\semmvl}[1]{{\llbracket{#1}\rrbracket}^{\mvl}}
\newcommand{\semjointvl}[1]{{\llbracket{#1}\rrbracket}}

\newcommand{\twovl}{\mathsf{2VL}}
\newcommand{\tinytwovl}{\mathsf{\scriptscriptstyle 2VL}}
\newcommand{\syneq}{\pmb{=}}
\newcommand{\threevl}{\mathsf{3VL}}
\newcommand{\mvl}{\mathsf{MVL}}

\newcommand{\bag}[1]{ \{ {#1} \} }

\newcommand{\ground}{\normalfont\textbf{gr}}
\newcommand{\semtwovlgrnd}[1]{{\llbracket{#1}\rrbracket}^{\ground}}

\newcommand{\semtwovlsyneq}[1]{{\llbracket{#1}\rrbracket}^{\pmb{=}}}

\newcommand{\arty}{\mathsf{arity}}

\newtheorem{theorem}{Theorem}
\newtheorem{proposition}{Proposition}
\newtheorem{lemma}{Lemma}
\newcounter{example}
\newenvironment{example}[1][]{\refstepcounter{example}\par\medskip
	\noindent \textbf{Example~\theexample. #1} \rmfamily}{\medskip}
\newtheorem{corollary}{Corollary}
\newtheorem{definition}{Definition}

\usepackage{textcomp}
\newcommand{\ttpcr}{\renewcommand{\ttdefault}{pcr}\ttfamily}
\newcommand{\sqlkw}[1]{%
	\text{\normalfont\small\ttpcr\bfseries\color{black} #1}}
\newcommand{\sqlkwblack}[1]{%
	\text{\normalfont\small\ttpcr\bfseries\color{black} #1}}

\lstnewenvironment{sql}[1][]{%
  \lstset{%
    language=SQL,%
    basicstyle=\small\ttpcr,%
    keywordstyle=\bfseries\color{black},%
    morekeywords={REFERENCES,IS},%
    deletekeywords={YEAR},%
  }%
  \lstset{#1}%
}{}

\newcommand{\tcond}{\trt{\theta}}
\newcommand{\fcond}{\trf{\theta}}
\newcommand{\tcon}[1]{\left(#1\right)^{\true}}
\newcommand{\fcon}[1]{\left(#1\right)^{\false}}

\newcommand{\sqlra}{\ensuremath{\textsc{RA}_\textsc{sql}}}
\newcommand{\sqlrarec}{\ensuremath{\textsc{RA}^{\textsc{rec}}_\textsc{sql}}}

\newcommand{\nulable}[1]{\mathtt{nullable}(#1)}
\newcommand{\src}[1]{\mathtt{sourceAttr}(#1)}

\newcommand{\name}[1]{\mathtt{name}(#1)}
\newcommand{\isunknownable}[1]{\mathtt{unkable}(#1)}

\newcommand{\trt}[1]{\textsf{tr}^{\true}(#1)}
\newcommand{\trf}[1]{\textsf{tr}^{\false}(#1)}
\newcommand{\trtfrom}[2]{\textsf{tr}_{#1}^{\true}(#2)}
\newcommand{\trffrom}[2]{\textsf{tr}_{#1}^{\false}(#2)}
\newcommand{\trtosql}[1]{{\small\textsf{toSQL}}(#1)}
\newcommand{\trfromsql}[1]{{\small\textsf{fromSQL}}(#1)}
\newcommand{\trtosqlfrom}[2]{{\small\textsf{toSQL}}_{#1}(#2)}

 \newenvironment{repeatresult}[2]
 {\vskip0.5em\par\textsc{#1 #2.}\em}
 {\vskip1em}

\begin{document}

\title{SQL Nulls and Two-Valued Logic}

\author{Leonid Libkin}
\affiliation{%
  \institution{The University of Edinburgh \& RelationalAI}
  \city{Edinburgh \& Paris}
  \country{UK \& France}}
  \orcid{0000-0002-6698-2735}
\email{l@libk.in}

\newcommand{\ligmaff}[1][2.1pt]{%
  {\fontdimen2\font=#1 LIGM, Universit\'e Gustave Eiffel, CNRS}}
  
\author{Liat Peterfreund}
\affiliation{%
  \institution{The Hebrew University of Jerusalem}
\city{Jerusalem}
\country{Israel}}
\orcid{0000-0002-4788-0944}
\email{liatpeter@cs.huji.ac.il}

\renewcommand{\shortauthors}{Leonid Libkin \& Liat Peterfreund}
\begin{abstract}
The design of SQL is based on a three-valued logic (3VL), rather than
the familiar two-valued Boolean logic (2VL). In addition to
\textit{true} and \textit{false}, 
3VL adds \textit{unknown} to handle nulls.  Viewed as indispensable
for SQL expressiveness, it is often criticized for unintuitive
behavior of queries and for being a source of programmer mistakes.

We show that, contrary to the widely held view, SQL could have been
designed based on 2VL, without any loss of expressiveness. Similarly
to SQL's WHERE clause, which only keeps true tuples, we conflate false
and unknown for conditions involving nulls to obtain an equally
expressive 2VL-based version of SQL. This applies to the core of the
1999 SQL Standard.

Queries written under the 2VL semantics can be efficiently translated
into the 3VL SQL and thus executed on any existing RDBMS. We show that
2VL enables additional optimizations. To gauge its applicability, we
establish criteria under which 2VL and 3VL semantics coincide, and
analyze common benchmarks such as TPC-H and TPC-DS to show that most
of their queries are such. For queries that behave differently under
2VL and 3VL, we undertake a user study to show a consistent preference
for the 2VL semantics.

\end{abstract}




\maketitle
\section{Introduction}
To process data with nulls, SQL uses a three-valued logic (3VL), with an additional truth value {\em unknown}. This is one of the most often criticized aspects of the language, and one that is very confusing to programmers \cite{BrassG06}. Database texts are full of damning statements about the treatment of nulls, such as the inability to explain them in a ``comprehensible'' manner \cite{datedarwen-sql}, their tendency to ``ruin everything'' \cite{celko} and outright recommendations to ``avoid nulls'' \cite{date2005}. The latter, however, is often not possible: in large volumes of data, incompleteness is hard to avoid. 

Issues related to null handling stem not just from the use of 3VL, but from multiple and disparate ways of using it. To illustrate:
\begin{itemize}
    \item Conditions, such as those in \sqlkw{WHERE},  are evaluated under 3VL, with any atomic condition involving a \NULL\ resulting in {\em unknown}. In the end, however, only {\em true} tuples are kept; that is, {\em false} and {\em unknown} are conflated.
    \item Constraints, such as \sqlkw{UNIQUE} and foreign keys, are too evaluated under 3VL, but then a constraint holds if it does not evaluate to {\em false}; that is, {\em true} and {\em unknown} are conflated.
    \item  SQL's \NULL\ can also be viewed as a syntactic constant,
    making two \NULL s equal; this is how grouping and set operations
    work.  
\end{itemize}

Not only is the SQL programmer forced to use a logic different from other languages they are familiar with, even that logic is applied in different ways in different scenarios. 

We now look at some examples where 3VL causes confusion even for very simple SQL queries. As a starter, consider the rewriting of \sqlkw{IN} subqueries into \sqlkw{EXISTS} ones. Queries

\begin{enumerate}[label=(Q\arabic*),ref=Q\arabic*]
\item \begin{sql}
SELECT R.A FROM R WHERE R.A NOT IN
  ( SELECT S.A FROM S )
\end{sql}
\label{sql:in}
\end{enumerate}
and 
\begin{enumerate}[label=(Q\arabic*),ref=Q\arabic*,resume]
\item \begin{sql}
SELECT R.A FROM R WHERE NOT EXISTS
  ( SELECT S.A FROM S WHERE S.A=R.A )
\end{sql}
\label{sql:exists}
\end{enumerate}
would regularly be presented as equivalent 
(see, e.g., \cite{sql2alg}). While equivalent if both \sqlkw{NOT}s are removed, these queries differ in SQL:
if $R=\{1,\NULL\}$ and $S=\{\NULL\}$, then \ref{sql:in} returns no tuples, while \ref{sql:exists} returns $\{1,\NULL\}$. Such presumed, but incorrect, equivalence is a trap many SQL programmers are not aware of (see \cite{celko,BrassG06}).

As another example, consider two queries given as an illustration of the HoTTSQL system for proving query equivalences \cite{hottsql}:

\begin{enumerate}[label=(Q\arabic*),ref=Q\arabic*,resume]
\item \begin{sql}
SELECT DISTINCT X.A FROM R X, R Y
  WHERE X.A=Y.A
\end{sql}
\label{sql:hott1}

\item \begin{sql}
SELECT DISTINCT R.A FROM R 
\end{sql}
\label{sql:hott2}
\end{enumerate}

\noindent Queries \ref{sql:hott1} and \ref{sql:hott2} are claimed to be equivalent in \cite{hottsql}, but this is not the case: if $R=\{\NULL\}$, then \ref{sql:hott1} returns 
an empty table while \ref{sql:hott2} returns $\NULL$. In fairness, the reason why they are equivalent in \cite{hottsql} is that HoTTSQL considers only databases {\em without nulls}. Nonetheless, this is illustrative of the subtleties surrounding SQL nulls: what  \cite{hottsql} chose as an ``easy'' example of equivalence involves two non-equivalent queries on a simple database containing $\NULL$. 

Over the years, two main lines of research emerged for dealing with these problems. One is to provide a more complex logic for handling incompleteness 
\cite{date-sigmodr, CGL16, Yue91, DBLP:conf/future/JiaFM92, Gessert90, Grant-IS97}. These proposals did not take off, because the underlying logic is even harder for programmers than 3VL. An alternative is to have a language with no nulls at all, and thus resort 
to the usual two-valued logic. This found more success, for example in
the ``3rd manifesto'' \cite{third-manifesto} and the Tutorial D
language, as well as in the  LogicBlox system \cite{logicblox} and its successor \cite{Relweb}, which
use the sixth normal form to eliminate nulls. But nulls do occur in
most SQL databases and thus must be handled; the world is not yet
ready to dismiss them completely. 

We thus pursue a different approach: a flavor of SQL with nulls, but based 
on Boolean logic. 
This goal is to have a flavor of SQL that can be offered as an
alternative to coexist along with the 3VL standard. To achieve this,
we  need to fulfill the following criteria.

\begin{enumerate}[(1)]
\item {\em  Do not make changes unless necessary}: On databases without nulls queries should be written exactly as before, and return the same results;
\item {\em Do not lose any queries; do not invent new ones}: The new version of SQL should have exactly the same expressiveness as its version based on 3VL;
\item {\em Do not make queries overly complicated}: For each SQL query using 3VL, 
the equivalent query in the two-valued
should not add joins, and be roughly of the same size.
\end{enumerate}

We pursue these goals along two different routes.  First, we provide 
theoretical evidence that our desiderata are fulfilled for the core of
SQL.  Second, we go beyond theoretical results and supplement
them by preliminary evidence of the utility of a version of SQL devoid
of 3VL.  

One may wonder why our goal is even achievable, considering almost 40
years of SQL practice firmly rooted in 3VL.  The reason to pursue this
line of work lies in two recent results that made steps in the right
direction, albeit for simpler
languages. First, \cite{DBLP:journals/pvldb/GuagliardoL17} showed that
in the most basic fragment of SQL capturing relational algebra
(selection-projection-join-union-difference), the truth value {\em
unknown} can be eliminated from conditions
in \sqlkw{WHERE}. Essentially, it rewrote conditions by
adding \sqlkw{IS NULL} or \sqlkw{IS NOT NULL}, in a way that they
could never evaluate to {\em unknown}.  Following
that, \cite{AIJ22} considered many-valued
first-order predicate calculi under set semantics, and showed that no
many-valued logic provides additional expressive power over Boolean
logic.

\paragraph*{Results}
We show that the elimination of {\em unknown} works for SQL, including its core features from the 1992 Standard 
(full relational algebra expanded with 
arithmetic functions and comparisons, 
aggregate functions and \sqlkw{GROUP BY};
comparisons of aggregates in \sqlkw{HAVING}, 
subqueries connected by \sqlkw{IN}, \sqlkw{EXISTS}, \sqlkw{ALL}, \sqlkw{ANY}, and 
set operations \sqlkw{UNION}, \sqlkw{INTERSECT}, \sqlkw{EXCEPT}, optionally with \sqlkw{ALL}), 
as well as \sqlkw{WITH RECURSIVE} added in the 1999 Standard
 \cite{datedarwen-sql,Melton-SQL99}.
 
The {\em unknown} appears when one evaluates a condition such as \lstinline{R.A=S.A}
in which one or more arguments are \NULL. 
Once a condition in \sqlkw{WHERE} is evaluated, 
only {\em true} tuples are kept, 
while {\em unknown} or {\em false} ones are dismissed. 
A minimal change that ensures elimination of 3VL then brings this  conflating {\em unknown} with {\em false} forward, 
before the exit from the  \sqlkw{WHERE} clause. One small variation
lies in treating conditions like $\NULL=\NULL$. One may still evaluate them to
{\em false}, or assume syntactic equality as in \sqlkw{GROUP BY} and
evaluate them to {\em true}. Either way, we obtain 
two-valued versions of SQL satisfying our desiderata. 

Replacing 3VL with 2VL  does not necessitate any changes to the underlying implementation of RDBMSs. A user can write a query under a two-valued Boolean semantics; the query is then translated into an equivalent one under the standard SQL semantics, which any of the existing engines can evaluate. 

We investigate the impact of this result from two different angles. The first concerns optimizations. By changing the semantics we change equivalences among queries.  We show that  the two-valued version of SQL recovers certain optimizations, in particular those often incorrectly assumed by programmers under 3VL. 

Our second question concerns real-life usability of two-valued SQL. To investigate this, we ask two questions:
\begin{enumerate}
    \item\label{item-one} Does it happen often that the choice of logic, 3VL or 2VL, has no impact on the query output?
    \item\label{item-two} If there is a difference between 3VL and 2VL, which one would users prefer?
\end{enumerate}

Regarding (\ref{item-one}), we observe that 
for many queries, there is actually no difference between outputs  while using 2VL or 3VL. We provide sufficient conditions for this to happen, and then analyze queries in commonly used benchmarks, TPC-H \cite{tpch} and TPC-DS \cite{tpcds}, to show that a huge majority of queries fall in that category, giving us the two-valued SQL essentially for free. This is not very surprising since these benchmarks were written by  experienced programmers who know how to avoid semantic pitfalls. 

When there is a difference, the only way to know what users prefer is to ask the users. We thus designed an introductory user survey asking about preference for 2VL or 3VL in both query outputs and query equivalence. As with every user survey, there is a tradeoff between the costs of the running a survey and reliability of its results. This being the first survey of the kind, we wanted to get an initial indication of what users think; starting the project we had no idea whether they would love the idea of 2VL, or reject it outright, or fall somewhere in between. The survey of roughly 80\% practitioners and 20\% academics showed that by -- on average -- the margin of 2-to-1 users prefer 2VL. This should not be viewed as the final word but rather as an initial confirmation of the feasibility of the approach, and an invitation for a more detailed user study before potential proposals for language changes.

Finally, we show an extension of our results: no other reasonable (essentially, avoiding paradoxical behaviour) many-valued
logic in place of 3VL could give a more expressive language than SQL.

\paragraph*{The choice of language}

To prove results formally, we need a language closely resembling SQL
and yet having a formal semantics one can reason about.  Our choice is
an extended relational algebra (RA) similar 
to an algebra into which RDBMS implementations translate SQL. It 
expands the standard textbook RA with bag semantics, duplicate elimination, and several new features. 
Selection conditions, in addition to the standard comparisons such as $=$ and $<$, 
include tests for nulls (as SQL's \sqlkw{IS NULL}) and both \sqlkw{IN} or \sqlkw{EXISTS} {\em subqueries}.  We also add conditions  $\brt \omega  \any(E)$ and $\brt \omega \all(E)$ with the semantics of SQL's  \sqlkw{ANY} and \sqlkw{ALL} (they check whether $\bar t \omega \bar t'$ holds for some, respectively all, $\bar t'$ in the result of $E$, where $\omega$ is one of the standard comparisons). Selection conditions are evaluated according to SQL's 3VL.
The algebra has aggregate functions and a grouping operation. It allows function application to attributes, to mimic expressions in the \sqlkw{SELECT} clause. It also  has an iterator operation whose semantics captures SQL recursion. 

\paragraph*{Related work}

The idea of using Boolean logic for nulls predates SQL;
it actually appeared in QUEL  (the language of Ingres that appeared in 1976  \cite{quel-original}; see details in the latest manual \cite{quel}). Afterwards, however, the main direction was in making the logic of nulls more rather than less complicated, with proposals ranging from three to six values \cite{date-sigmodr,CGL16,Yue91,DBLP:conf/future/JiaFM92,Gessert90} or producing more complex classifications of nulls, e.g., \cite{Grant-IS97,zaniolo-nulls}. Elaborate many-valued logics for handling incomplete and inconsistent data were also considered in the AI literature; see, e.g., \cite{Ginsberg88,fitting91,AAZ11}. Proposals for eliminating nulls have appeared in \cite{third-manifesto,logicblox,Relweb}. 

There is a large body of work on achieving correctness of query results on databases with nulls where correctness assumes the standard notion of certain answers \cite{IL84}. Among such works are \cite{kennedy-glavic-sigmod19,tods16,pods16,italians-approx}. They assumed either SQL's 3VL, or the Boolean logic of marked nulls \cite{IL84}, and showed how query evaluation could be modified to achieve correctness, but they did not question the underlying logic of nulls. Our work is orthogonal to that: we are concerned with finding a logic that makes it more natural for programmers to write queries; once this is achieved, one will need to modify the evaluation schemes to produce subsets of certain answers if one so desires. For connections between real SQL nulls and theoretical models, such as marked or Codd nulls, see \cite{GUAGLIARDO201946}.
 
Some papers looked into handling nulls and incomplete data in
bag-based data models as employed by
SQL \cite{ijcai2017-ConsoleGL,obda-bags,de-bags}, but none focused on
the underlying logic of nulls.

Finally, our user survey can be seen as complementary to the extensive
survey on the use of nulls \cite{vldb22}; that survey asked multiple
questions but not on replacing SQL's logic of nulls. 

\paragraph*{Organization}

Section~\ref{sec:syn} presents the syntax and the semantics of the language. Section~\ref{sec:twoval} shows how to eliminate {\em unknown} to achieve our desiderata. Section~\ref{sec:restore} studies optimizations under 2VL. 
Section~\ref{sec:coincide} discusses conditions under which 2VL and 3VL semantics are equivalent. 
Section~\ref{sec:survey} studies applicability of our results. Extensions are given in 
Section~\ref{sec:mvl}. 

\section{Query Language: \protect\sqlra}
\label{sec:syn}
\newsavebox{\synt}
\sbox{\synt}{%
        \parbox{0.9\textwidth-2\fboxsep-2\fboxrule}{%
\begin{center}
    \begin{tabular}{ll}
       {\bf Terms:}  &$t \df n \,  | \, c \, |\, \NULL\,
|\, N\, |\, f(t_1,\cdots,t_k)\,\,\,\, \,  n \in
\Num,\,\, c\in \Const,\, \, N \in \sN,\,\, f\in \Omega$  \\ & \\
  {\bf Expressions:}  & 
  \begin{tabular}[t]{ rll} 
		$E:=$ & $R$ & (base relation) \\ 
		 & $\pi_{t_1[\shortrightarrow\! N_1'], \cdots,
	t_m[\shortrightarrow\!   N_m']}(E)$ & (generalized projection
	w/ optional renaming) \\ 
		 & $\sigma_{\theta}(E)$ & (selection)  \\ 
 		 & $E\times E$ & (product)  \\ 
 		 & $E\cup E$ & (union)  \\ 
 		 & $E \cap E$ & (intersection)  \\ 
 		 & $E \setminus E$ & (difference)  \\ 
 	 	 & $\epsilon(E)$ & (duplicate elimination)  \\ 
 	 	 & $\Group_{\barN,
 	 	 \langle F_1(N_1)[\shortrightarrow\!   N'_1],\cdots, F_m(N_m)[\shortrightarrow\!   N'_m]
 	 	 \rangle}(E)$ & (grouping/aggregation w/ optional renaming)\\
	\end{tabular}\\ &\\
 {\bf Atomic conditions:} &
$ac \df  \true \, | \, \false \, | \,
 \isnul(t)  \, | \,
  \brt\, \omega \, \brt'\, | \, 
 \brt\in E \, | \,  
 \isempty(E) 
\, | \, 
\brt\, \omega\, \any(E)\, | \, 
\brt\, \omega\, \all(E) $  
\,\,\,\,\,\,\,\,\,\,\,
$ \omega \in \{\eq,\noteq, <, >, \leq, \geq\}$\\&\\
{\bf Conditions:}&
$
\theta \df ac \, |  \, \theta \vee \theta  \, | \, \neg \theta \, |  \, \theta \wedge \theta \,
$\\
   \end{tabular}
   \\
\end{center}
        }%
}%
\newcommand{\syntfig}{%
        \begin{figure*}
                \centering{
                \fbox{\usebox{\synt}}}
                \caption{Syntax of \sqlra}
               \label{fig:syntax}
        \end{figure*}
}%
{\syntfig}

Given the idiosyncrasies of SQL's
syntax, it is not an ideal language -- syntactically -- to reason
about. We know, however, that its queries are all translatable into an
extended RA; indeed, this is what is done inside every
RDBMS, and multiple such translations are described in the  literature
\cite{CG85,NPS91,sql2alg,DBLP:journals/pvldb/GuagliardoL17,benzaken}. 
Our language \sqlra\  is
close to what real-life SQL queries are translated into.

\subsection{Data Model}
The usual presentation of RA assumes a countably
infinite domain of values. To handle languages with
aggregation, we need to distinguish columns of numerical types. As
not to over-complicate the model, we assume two types: a numerical and
non-numerical one (we call it {\em ordinary}).
This is without any loss of generality since
the treatment of nulls as values of all types is the same, except
numerical nulls that behave differently with respect to aggregation\footnote{For example, the value of $1+\NULL$ depends on whether it occurs as an arithmetic expression (in which case it is $\NULL$) or aggregation over a column (when it is $1$).}.

Assume the following pairwise disjoint countable infinite sets:
\begin{itemize}
\item $\sN$ of attribute \emph{names}, and
	\item $\Num$ of \emph{numerical values}, and
	\item $\Const$ of \emph{(ordinary) values}.
\end{itemize}

Each name has a {\em type}:  either $\so$ (ordinary)
 or $\sn$ (numerical). If $N\in\sN$, then $\type(N) \in \{\so,\sn\}$ defines the type of elements in 
 column $N$.
Furthermore, $\type(e)= \sn $ if $e\in \Num$ and 
$\type(e) = \so$ if $e\in\Const$.
 We use a  fresh symbol $\NULL$ to denote the null value.

\emph{Typed} records and relations are defined as follows.
Let $\tau \df \tau_1 \cdots \tau_n$ be a word over the alphabet $\{ \so, \sn \}$. 
A $\tau$-{\em record} $\bara$ with \emph{arity} $n$ is a tuple $(a_1,\cdots, a_n)$ where
$a_i \in \Num \cup \Null$ whenever $\tau_i = \sn$, and
	$a_i \in \Const \cup \Null$ whenever $\tau_i = \so$. 

For an $n$-ary relation symbol $R$ in the schema we write 
$\ell(R) = N_1\cdots N_n \in \sN^n$ for the sequence of its attribute
names.  
The {\em type} of 
$R$ is the sequence $\type(R)=\type(N_1)  \cdots \type(N_n)$. 
A \emph{relation} $\mathbf{R}$ over $R$ is a \textit{bag}
of $\type(R)$-records
(a record may appear
more than once). 
 We  write  $\bara \in_k \mathbf{R}$ if $\bara$ occurs 
 exactly $k>0$ times in $\mathbf{R}$. 

A \emph{relation schema} $\cS$ is a set of relation symbols and their
types, i.e., a set of pairs $(R, \type(R))$. 
A \emph{ database} $D$ over a relation schema $\cS$ associates with each
$(R,\type(R)) \in \cS$ a relation of $\type(R)$-records.

\subsection{Syntax}
A \emph{term} is defined recursively as 
either a numerical value in $\Num$, or an ordinary
value in $\Const$, or $\NULL$, or a name in $\sN$, or an element of
the form $f(t_1,\cdots, t_k)$ where $f:\Num^k\to\Num$ is a  \emph{$k$-ary numerical
function} (e.g., addition or multiplication) and  $t_1,\cdots, t_k$ are
terms
that evaluate to values of numerical type. 

 An \emph{aggregate function} is a function $F$ that maps
 bags of numerical values into a numerical value.
 For example, SQL's
 aggregates \sqlkw{COUNT}, \sqlkw{AVG}, \sqlkw{SUM}, \sqlkw{MIN}, \sqlkw{MAX} are
 such. 

The algebra is parameterized by a collection 
$\Omega$ of numerical and aggregate
functions. 
We assume the standard comparison predicates $\eq,\noteq, <,>, \leq, \geq$ on numerical values are available
\footnote{We focus, without loss of generality, on these predicates; our results apply also for predicates of higher arities, e.g., \sqlkw{BETWEEN} in SQL.}.

Given a schema $\cS$ and such a collection 
$\Omega$, the syntax of $\sqlra$ expressions and conditions 
over $\cS \cup \Omega$ is given in Fig.~\ref{fig:syntax}, where
$R$ ranges over relation symbols in $\cS$, each $t_i$ is a term, each $N_i,N'_i$ is a name, each
$\barN$ is a tuple of names, and each $F_i$ is an aggregate function. 
In the generalized projection and in the grouping/aggregation, the parts in the squared brackets (i.e., $[\shortrightarrow\!   N_i]$ and $[\shortrightarrow\!   N'_i]$)  are optional renamings.

The \emph{size} of an expression is defined 
the
size of its parse tree. We assume that comparisons between tuples are spelled out as Boolean combinations of atomic comparisons: e.g., $(x_1,x_2)=(y_1,y_2)$ is $x_1=y_1 \wedge x_2=y_2$ and $(x_1,x_2) < (y_1,y_2)$ is $x_1 < y_1 \vee (x_1=y_2 \wedge x_2<y_2)$. 

In what follows, we restrict our attention to expressions with
well-defined semantics (e.g., we forbid aggregation over non-numeri\-cal
columns or functions applied to arguments of wrong types). 

\subsection{Semantics}
\label{sec:semp}

To make reading easier, we present the semantics by recourse to SQL, but full and formal definitions exist too and are found in the full version~\cite{2VLarxiv}. 
The semantics function
$$\semvl{E}_{D,\eta}$$
defines the result of the evaluation of expression $E$ on database $D$
under the {\em environment} $\eta$. The {environment} $\eta$ is a partial mapping from the set $\sN$ of names to the union
$\Const \cup \Num \cup \Null$. It provides values
of parameters of the query. This is necessary to give semantics of subqueries that can refer to attributes from the outer query. 

\OMIT{ 
Indeed, to give a semantics of a query
with subqueries, we need to define the semantics of subqueries as
well. Consider for example the  subquery \lstinline{SELECT S.A FROM S}
\lstinline{WHERE S.A=R.A} of query $Q_2$ from the
Introduction. Here \lstinline{R.A} is a parameter, and to compute the
query we need to provide its value. Thus, an 
environment
$\eta$ is a partial mapping from the set $\sN$ of names to the union
$\Const \cup \Num \cup \Null$.  
}

Given an expression $E$ of \sqlra\ and a database $D$, the value of $E$
in $D$ is defined as $\semvl{E}_{D,\emptyset}$  where $\emptyset$ is the
empty mapping (i.e., the top level expression has no parameters). 

Similarly to SQL queries, each \sqlra\ expression $E$ produces tables over a list of attributes; this list will be denoted by $\ell(E)$.

\subsubsection{Base relations and generalized projections}

A base relation $R$ is SQL's $\sqlkw{SELECT * FROM R}$.
Generalized projection captures SQL's $\sqlkw{SELECT}$ clause. 
Each term $t_i$ is evaluated, optionally renamed, and added as
a column to the result. 
 For example
\begin{sql}
	SELECT A, B, A+2 AS C, A*B AS D FROM R
\end{sql}
is written as 
$\pi_{A, B, \textsf{add2}(A)\shortrightarrow\!   C, \textsf{mult}(A,B)\shortrightarrow\!   D}(R)$, 
where $\textsf{add2}(x)\df x+2$ and $\textsf{mult}(x,y)\df x\cdot
y$. 
Projection follows SQL's bag semantics with
terms evaluated along with their multiplicity. 
For instance if $(2,3)\in_2 R$ and $(1,6)\in_3 R$ then the result of 
$\pi_{ \textsf{mult}(A,B)}(R)$ contains the tuple $(6)$ with multiplicity $5$.

\subsubsection{Conditions and selections}
SQL uses three-valued logic and its conditions are evaluated to either true (\vt), or false (\vf), or unknown (\vu).
Logical connectives are used to compose conditions, and truth values are propagated according to Kleene logic
below. 
\begin{center} \hspace*{-1mm}\begin{tabular}{c|ccc}
		$\land$ & \vt & \vf & \vu \\
		\hline %
		\vt & \vt & \vf & \vu \\
		\vf & \vf & \vf & \vf \\
		\vu & \vu & \vf & \vu
	\end{tabular} \hspace*{4mm}
	\begin{tabular}{c|ccc}
		$\lor$ & \vt & \vf & \vu \\
		\hline %
		\vt & \vt & \vt & \vt \\
		\vf & \vt & \vf & \vu \\
		\vu & \vt & \vu & \vu
	\end{tabular}
	\hspace{4mm}
	\begin{tabular}{c|c}
		& $\lnot$\\
		\hline
		\vt  & \vf  \\
		\vf  & \vt  \\
		\vu  & \vu  \\
	\end{tabular}\end{center}

\OMIT{
For each predicate $P\in\Omega$ we assume
its meaning is well defined when its arguments are not \NULL\ (e.g.,
$\leq$ on numbers). Then this is the meaning that is used when all
arguments are not \NULL, and if one is \NULL, then the value is
unknown (\vu). 
}
Atomic condition $\isnul(t)$ is evaluated to $\vt$ if the term $t$ is \NULL, and to $\vf$ otherwise.
Comparisons $t \, \omega \, t'$ are defined naturally when the arguments are not \NULL. If at least one argument is \NULL, then the value is
unknown (\vu). 
Formally: 
\newsavebox{\semsq}
\sbox{\semsq}{%
	\parbox{\columnwidth-2\fboxsep-2\fboxrule}{%
		\begin{align*}
		\semvl{t \,\omega\, t'}_{D,\eta} & \df
		\left
		\{
		\begin{matrix}
		\true  &   \semvl{t}_{\eta},\semvl{t'}_{\eta}\ne \NULL
		\text{ and } \semvl{t}_{\eta} \omega \semvl{t'}_{\eta}   \\ 
		\unknwn & \semvl{t}_{\eta} = \NULL \text{ or } \semvl{t'}_{\eta} = \NULL
		\\
		\false  &\text{otherwise}
		\end{matrix}
		\right.
		\end{align*}
	}%
}%

\newcommand{\figsem}{%
	\centering
	\fbox{{\usebox{\semsq}}}
}%

\medskip
\noindent
{\figsem}
\medskip

The condition $\brt \eq \brt'$, where $\brt=(t_1,\ldots,t_m)$ and $\brt'= (t_1',\ldots,t_m')$ that compares 
 tuples of terms
is the abbreviation of the 
conjunction $ \bigwedge_{i=1}^m t_i \eq t_i'$, and the comparison $\brt\noteq\brt'$ abbreviates 
$ \bigvee_{i=1}^m t_i \noteq t_i'$. Comparisons $<,\leq,\geq,>$ of tuples are defined lexicographically.

The
condition $\bar t \in E$, not typically included in RA, tests whether
a tuple belongs to the result of a query, and corresponds to
SQL's \sqlkw{IN} subqueries. 
If $E$ evaluates to the bag containing $\brt_1,\cdots, \brt_n$ then $\bar t \in E$ stands for the disjunction $\vee_{i=1}^n \brt \eq \brt_i$.
Other predicates not typically included in RA
presentation, though included here for direct correspondence with SQL,
are \sqlkw{ALL} and \sqlkw{ANY} comparisons. 
The condition $\brt \, \omega\, \any(E)$ checks whether 
there exists a tuple
$\brt'$ in $E$ so that $\brt\, \omega\, \brt'$ holds, where $\omega$ is one of the allowed comparisons. 
Likewise, 
$\brt \, \omega\, \all(E)$ checks whether 
$\brt\, \omega\, \brt'$ holds for every tuple
$\brt'$ in $E$ (in particular, if $E$ returns no tuples, this condition
is true). 
\OMIT{ 
That is, $\brt \, \omega\, \any(E)$ stands for the disjunction $\bigvee_{i=1}^n \brt \eq \brt_i$, and  $\brt \, \omega\, \all(E)$ stands for the conjunction $\bigwedge_{i=1}^n \brt \eq \brt_i$
where $\brt_1,\cdots, \brt_n$ are the tuples in the result of evaluating $E$.
}
If $\omega$ is $\eq$ or \mbox{$\noteq$}, conditions with $\any$ and
$\all$ are applicable at either ordinary or numerical type; if
$\omega$ is one of $<,\leq, >,\geq$, then all the components of $\brt$ and all attributes of $E$ 
are of numerical type.

The condition $\isempty(E)$ checks if the
result of $E$ is empty, and corresponds to SQL's \sqlkw{EXISTS}
subqueries. 
\OMIT{
One might be tempted to say that these are expressible via joins in
traditional relational algebra. There is a good reason to include them
directly. First, we want to stay as close to SQL as possible. Even
more importantly, these conditions behave {\em differently} in the
presence of nulls, and their expressibility via joins would require
complex conditions checking which attributes values are \NULL.
Indeed,
as we shall see, they behave differently with respect to treatment of
nulls;
in fact \sqlkw{EXISTS} subqueries follow the two-valued logic,
whereas \sqlkw{IN} subqueries are based on the three-valued logic. 
}
Note that \sqlkw{EXISTS} subqueries can be evaluated to $\vt$ or $\vf$,
whereas \sqlkw{IN} subqueries can also be evaluated to $\vu$.
The semantics of composite conditions is defined by the 3VL truth-tables. 

Selection evaluates the condition $\theta$ for each tuple,
and keeps tuples for which $\theta$  is \vt\ (i.e., not \vf\ nor \vu). 
Operations of generalized projection and
selection correspond to sequential scans in query plans (with
filtering in the case of selection).

\subsubsection{Bag operations and grouping/aggregations}

The operation $\epsilon$ is SQL's \sqlkw{DISTINCT}: it eliminates
duplicates and keeps one copy of each record.
Operations union, intersection, and difference, 
have the standard meaning under
the bag semantics, and correspond, respectively, to SQL's \sqlkw{UNION ALL}, \sqlkw{INTERSECT ALL}, and \sqlkw{EXCEPT ALL}. 
Dropping the keyword \sqlkw{ALL} amounts to using set semantics for both arguments and operations. 

Cartesian product has the standard meaning and corresponds to listing relations in the \sqlkw{FROM} clause.


We use SQL's semantics of functions: if one of its arguments
is \NULL, then the result is null (e.g., $3+2$ is 5, but
$\NULL+2$ is $\NULL$).


Finally, we describe the operator
$\Group_{\barN,\langle \barM \rangle} (E)$.
The tuple $\barN$ lists
attributes in \sqlkw{GROUP BY}, and the $i$'th coordinate  of $\barM$ is of the form  $F_i(N_i)$ where $F_i$ is an aggregate over the numerical columns $N_i$
optionally renamed $N'_i$ if $[\shortrightarrow\! N'_i]$ is present. 
For example 
\begin{sql} 
SELECT A, COUNT(B) AS C, SUM(B) FROM R GROUP BY A 
\end{sql} 
will be expressed by 
$\Group_{A,\langle F_{\texttt{count}}(B)[\shortrightarrow\!   C], F_{\texttt{sum}}(B)\rangle} (R)$
where \linebreak $F_{\texttt{count}}(\{a_1,\ldots,a_n\})\df n$ and
$F_{\texttt{sum}}(\{a_1,\ldots,a_n\})\df a_1+ \cdots +a_n$.  Note that
$\barN$ could be empty; this corresponds to computing aggregates
over the entire table, without grouping, for example, as in \sqlkw{SELECT
COUNT} \lstinline{(B)}, \sqlkw{SUM} \lstinline{(B)} \sqlkw{FROM} 
\lstinline{R}. 


\begin{example} 
\label{queries-ex}
We start by showing how queries $Q_1$--$Q_4$ from the
introduction are expressible in \sqlra:
$$\begin{array}{rcl}
Q_1 & = & \sigma_{\neg(R.A \in S)}(R)
\\
Q_2 & = & \sigma_{\isempty(\sigma_{R.A=S.A}(S))}(R)
\\
Q_3 & = & \epsilon\Big(\pi_{X.A}\big(\sigma_{X.A=Y.A}\big(\rho_{R.A\to X.A}(R) \times 
\rho_{R.A\to Y.A}(R)\big)\big)\Big)
\\
Q_4 & = & \epsilon\big(\pi_{R.A } (R)\big)
\end{array}
$$

\noindent
A more complex example is a query $Q_5$ below; it is a slightly
simplified (to fit in one column) query 22 from TPC-H 
\cite{tpch}: 

\begin{sql}
SELECT c_nationkey, COUNT(c_custkey)
FROM customer 
WHERE c_acctbal > 
 (SELECT avg(c_acctbal)
  FROM customer WHERE c_acctbal > 0.0 AND
  c_custkey NOT IN (SELECT o_custkey FROM orders) )
GROUP BY c_nationkey
\end{sql}

Below we use abbreviations $C$ for \lstinline{customer} and
$O$ for \lstinline{orders}, and abbreviations for attributes like $c\_n$
for \lstinline{c_nationkey} etc. 
The \sqlkw{NOT IN} condition in the
subquery is then translated as  $\neg (c\_c \in \pi_{o\_c}(O))$, the whole condition is translated as
$\theta \df 
(c\_a > 0) \wedge 
\neg (c\_c \in \pi_{o\_c}(O))$
and the aggregate subquery becomes 
$$Q_{agg}\df\Group_{\emptyset,\langle F_{\texttt{avg}}(c\_a)\rangle} \Big(\pi_{c\_a}(\sigma_{\theta}(C))\Big)\,.$$
Notice that there is no grouping for this aggregate, hence the empty
set of grouping attributes. 
Then the condition in the \sqlkw{WHERE} clause of the query is 
$\theta' \df c\_a > {\any}(Q_{agg})$
which is then applied to $C$, i.e., 
$\sigma_{ c\_a > {\any}(Q_{agg})}(C)$, and finally 
grouping by $c\_n$ and counting of $c\_a$ are performed over it,
giving us
$$\Group_{c\_n,\langle F_{\texttt{count}}(c\_c)\rangle }\big( \sigma_{ c\_a > {\any}(Q_{agg})}(C)  \big)\,.$$
Putting everything together, we have the final \sqlra\ expression:
$$\Group_{c\_n,\langle F_{\texttt{count}}(c\_c)\rangle }\Big( \sigma_{ c\_a > {\any}\Big(
\Group_{\emptyset,\langle F_{\texttt{avg}}(c\_a)\rangle} \big(\pi_{c\_a}(\sigma_{\theta}(C))\big)
\Big)}(C)  \Big)\,.$$
\end{example}

\vspace*{-7mm}

\subsubsection*{{\bf ADDING RECURSION}}
We now incorporate recursive queries, a
feature added in 
the SQL 1999 standard with its \sqlkw{WITH RECURSIVE} construct. 
While extensions of relational algebra with various
kinds of recursion 
exist (e.g., transitive closure \cite{agrawal88} or fixed-point
operator \cite{JGGL20}), we stay
closer to SQL as it is. Specifically, SQL uses a special type of iteration
-- in fact two
kinds depending on the syntactic shape of the query \cite{postgres}. 


\subsubsection*{Syntax of \sqlrarec}
\newcommand{\ecup}{\sqcup}

Recall that $\cup$ stands for bag union, i.e., multiplicities of tuples
are added up, as in SQL's \sqlkw{UNION ALL}. We also need the
operation $B_1 \ecup B_2$ defined as $\epsilon(B_1\cup B_2)$, i.e.,
union in which a single copy of each tuple is kept. This corresponds
to SQL's \sqlkw{UNION}. 

	An \sqlrarec expression is defined with the grammar of \sqlra\
	in Fig.~\ref{fig:syntax}) with the addition of the 
	constructor 
	$
	\mu R.E
	$  
	where $R$ is a fresh relation symbol (i.e., not in the
	schema) and $E$ is an expression  of the form  
	$E_1 \cup E_2$ or  $E_1\ecup  E_2$ where both $E_1$
	and $E_2$ are \sqlrarec expressions and $E_2$ may contain a
	reference to 
	$R$.

In SQL, various restrictions are imposed on query
$E_2$, such as the linearity of recursion (at most one
reference to $R$ within $E_2$), restrictions on the use of recursively
defined relations in subqueries, on the use of aggregation, etc. These
eliminate many of the common cases of non-terminating queries.  Here
we shall not impose these restrictions, as our result is {\em more
general}: passing from 3VL to two-valued logic is possible even if
such restrictions were not in place. 
Note that different RDBMSs use
different restrictions on recursive queries (and sometimes even
different syntax); hence showing this more general result will ensure
that it applies to all of them.
%

	\subsubsection*{Semantics of \sqlrarec}
Similarly to the syntactic definition, we distinguish between the two cases. 

\OMIT{
\begin{enumerate}
\item $i=0$
\item $RES_i, WT_i = \semsql{E_1}_D$
\item If $WT_i = \emptyset$ return $RES_i$, otherwise
\item set $i \df i+1$
\item  $WT_i \df \semsql{E_2}_{D \cup WT_{i-1}}$
\item $RES_{i} \df RES_{i-1} \cup WT_i$
\end{enumerate}
}

For $\mu R. E_1 \cup E_2$, the semantics $\semsql{\mu R. E_1 \cup
E_2}_{D,\eta}$ is defined  
by the following iterative process:
\begin{enumerate}
\item $RES_0, R_0 \df \semsql{E_1}_{D,\eta}$
\item  $R_{i+1} \df \semsql{E_2}_{D \cup R_{i},\eta}$, \ \ 
$RES_{i+1} \df RES_{i} \cup R_{i+1}$
\end{enumerate}
with the condition that if $R_i = \emptyset$, then the iteration  stops
 and  $RES_i$ is returned.

For $\mu R. E_1 \ecup E_2$, the semantics is defined 
by a different iteration 
\begin{enumerate}
\item $RES_0, R_0 \df \semsql{\epsilon(E_1)}_{D,\eta}$
\item  $R_{i+1}\! \df\! \semsql{\epsilon(E_2)}_{D \cup R_{i},\eta}\! -\! RES_i$, \   
$RES_{i+1}\! \df\! RES_{i} \cup R_{i+1}$
\end{enumerate}
with the same stopping condition as before. 

Note that while for queries not involving recursion only the environment changes during the computation, for recursion the relation that is iterated over ($R$ above) changes as well, and each new iteration is evaluated on a modified database. 

Note that since $R_{i+1}$
does not contain any tuples from $RES_i$, we have $RES_i\cup
R_{i+1}=RES_i\ecup R_i$ above. That is, either $\cup$ or $\ecup$ could
be used in rule (2) of this iterative process. 

\section{Eliminating unknown} 
\label{sec:twoval}
To replace 3VL with Boolean logic, we need to
eliminate the unknown truth value. 
In SQL, \vu\ arises in \sqlkw{WHERE}
which corresponds to conditions in \sqlra. It 
appears as the 
result of evaluation of comparison predicates  such as
$=, \leq, \noteq$ etc. 
Consequently, it also arises in \sqlkw{IN}, \sqlkw{ANY}
and \sqlkw{ALL} conditions for subqueries.   
Indeed, to 
check $\bar
t \ \sqlkw{IN} \ \{\bar t_1,\cdots,\bar t_n\}$ one computes the
disjunction of  $\bar t \eq \bar t_i$, for $i\leq n$, each being the
conjunction of $a_j \eq b_j$ for $j\leq m$, if $\bar
t=(a_1,\cdots,a_m)$ and $\bar t_i=(b_1,\cdots, b_m)$, and similarly for \sqlkw{ANY}
and \sqlkw{ALL}. 

In comparisons, \vu\ appears due to the
rule that {\em if one parameter is \NULL, then the value of the predicate
is \vu}. 
Thus, we need to change this rule, and to say what to do when one of
the parameters is \NULL. 
In doing so, we are guided by SQL's existing semantics of conditions in 
\sqlkw{WHERE}. While those  can evaluate to
\vt, \vf, or \vu, in the end only the {\em true}
values are kept: that is, 
{\em \vu\ and \vf\ are conflated}. SQL does it at the end of evaluating
a condition; thus a natural approach to a two-valued version of SQL is
to use the same rule
{\em throughout} the evaluation.

This is a natural proposal, and in fact we shall that this results in
a version of SQL satisfying our desiderata.
It may have a potential
drawback with respect to optimizations. Namely, both
$\NULL \eq \NULL$ and $\NULL \noteq \NULL$ evaluate to \vf, and thus 
$\NULL \noteq \NULL$ cannot be equivalent to $\neg(\NULL \eq \NULL)$.
This however can easily be resolved
by treating conditions consistent with syntactic equality differently.

\subsection{The two-valued semantics $\semtwovl{\,}$ and $\semtwovlsyneq{\,}$}

As explained above, in the new semantics $\semtwovl{\,}$ we only need to modify the
rule for comparisons of terms,  $t\,\omega\, t'$. In the simplest case
(\vf\ instead of \vu) this is done by

\newsavebox{\semutof}
\sbox{\semutof}{%
	\parbox{\columnwidth-2\fboxsep-2\fboxrule}{%
		\begin{align*}
		\semtwovl{t \,\omega\, t'}_{D,\eta} & \df
		\left
		\{
		\begin{array}{ll}
		\true  &   \semvl{t}_{\eta},\semvl{t'}_{\eta}\ne \NULL
	, \ \text{and}\  \semvl{t}_{\eta} \,\omega\, \semvl{t'}_{\eta}   \\ 
		\false  &\text{otherwise}
		\end{array}
		\right.
		\end{align*}
	}%
}%

\newcommand{\figsemutof}{%
	\centering
	\fbox{{\usebox{\semutof}}}
}%

\medskip
\noindent
{\figsemutof}

\medskip
The rest of the semantics is exactly the same as before.  Note that in
conditions like 
$\brt \eq \brt'$,
or $\brt \in E$, 
or $\brt \, \omega \, \all(E)$, and 
$\brt \, \omega \, \any(E)$, 
the conjunctions and disjunctions will be interpreted as the 
standard  Boolean ones, since \vu\ no longer arises. 

A more elaborate version $\semtwovlsyneq{\,}$ takes into account
syntactic equality. It is the same as the $\semtwovl{\,}$ semantics
except for three comparisons compatible with equality: $\eq, \leq$ and
$\geq$. For them, it is as follows 

\newsavebox{\semsyneq}
\sbox{\semsyneq}{%
	\parbox{\columnwidth-2\fboxsep-2\fboxrule}{%
		\begin{align*}
		\semtwovlsyneq{t\ \omega\ t'}_{D,\eta} & \df
		\left
		\{
		\begin{array}{ll}
		\vt  &    \semtwovl{t\ \omega\ t'}_{D,\eta}
		= \vt \text{ or }\semvl{t}_{\eta} = \semvl{t'}_{\eta} = \NULL  \\ 
		\vf  &\text{otherwise}
		\end{array}
		\right.
		\end{align*}
	}%
}%

\newcommand{\figsemsyneq}{%
	\centering
	\fbox{\usebox{\semsyneq}}
}%

\medskip
\noindent
{\figsemsyneq}
\medskip

\noindent and keeping the rest as in the definition of
$\semtwovl{\,}$. The only difference is that now conditions
$\NULL\eq\NULL$, $\NULL\leq\NULL$, and $\NULL\geq\NULL$ evaluate to
true. 

\subsection{Capturing SQL with $\semtwovl{\,}$ and $\semtwovlsyneq{\,}$}
We now show that the two semantics presented above
fulfill our 
desiderata for a two-valued version of SQL. Recall that it postulated
three requirements: (1) that no expressiveness be gained or lost
compared to the standard SQL; (2) that over databases without nulls no
changes be required; and (3) that when changes are required in
the presence of nulls, they ought to be small and not affect significantly 
the size of the query.
These conditions are formalized in the definition
below.

\begin{definition}
	\label{capture-def}
	A semantics $\semsql{\,}'$ of  queries 
        {\em captures} the
	semantics $\semsql{\,}$ of $\sqlra$ if the following 
	are satisfied: 
	\begin{enumerate}\itemsep=0pt
		\item
		\label{capture-def1}
		for every expression $E$ of $\sqlra$ there exists an
		expression $G$ of $\sqlra$ such that, for each
	database $D$, 
		$$\semsql{E}_D'=\semsql{G}_D\,;$$
		\item
		\label{capture-def2} 
		for every expression $E$ of $\sqlra$ there exists an
		expression $F$ of $\sqlra$ such that, for every
		database $D$, 
		$$\semsql{E}_D=\semsql{F}_D'\,;$$
		\item
		\label{capture-def3}
		for every expression $E$ of 
		$\sqlra$, and every database $D$ without nulls, 
		 $\semsql{E}_D=\semsql{E}'_D$.
	\end{enumerate}
        When in place of $\sqlra$ above we use $\sqlrarec$, then we
		speak of caputring the semantics of $\sqlrarec$.	

If the size of expressions $F$ and $G$ in
	items (\ref{capture-def1}) and (\ref{capture-def2}) 
	is at most linear in the size of  $E$,  we
	say that the semantics is captured {\em efficiently}.
	\qed
\end{definition}

Our main result is that the two-valued semantics of SQL capture its
standard semantics efficiently. 

\newcommand{\thmmain}{The $\semtwovl{\,}$ 
and $\semtwovlsyneq{\,}$ 
semantics of
	\sqlrarec\ expressions, and of \sqlra\ expressions, capture
	their  SQL semantics
$\semsql{\,}$	efficiently.}
\begin{theorem} \label{thm:main}
\thmmain
\end{theorem}	

Note that the capture statement for \sqlra\ is {\em not} a corollary of the
statement of \sqlrarec.

\OMIT{  
In the rest of the section, we do the following:
\begin{itemize}
\item In Section \ref{sec:tosql}, we provide details of one directions
of the translation, from $\semtwovl{\,}$ to $\semthreevl{\,}$. It
shows how a query written in the new two-valued style cane be run
correctly using the existing RDBMS technology and current
implementations of SQL.
\item In Section \ref{sec:restore}, we show how the new semantics
helps otimizations, by restoring some often assumed equaivalences that
nonetheless do not hold in SQL.

\item In Section \ref{sec:coincide}, we present sufficient conditions
for queries to produce the same result under $\semtwovl{\,}$
and $\semthreevl{\,}$. For such queries (that we shall later see are
very common) one could use the two-valued semantics without any need
to rewrite queries.

\end{itemize}
}

\subsection*{Running 2VL on existing RDBMSs}

We sketch one direction of the proof of Theorem \ref{thm:main}, namely
from
$\semtwovl{\,}$ and $\semtwovlsyneq{\,}$ to $\semsql{\,}$. 

From $\semtwovl{\cdot}$ to $\semsql{\,}$, we define the translation $\trtosqlfrom{\tinytwovl}{\,}$ that specifies how to take a query $E$
written under the 2VL semantics that conflates \vu\ with \vf\ and translate it into a query
$\trtosqlfrom{\tinytwovl}{E}$ that gives the same result when evaluated under the usual
SQL semantics: $\semtwovl{E}_D = \semthreevl{\trtosqlfrom{\tinytwovl}{E}}_D$ for every
database $D$. Thus, $\trtosqlfrom{\tinytwovl}E$ can execute a 2VL query in any
existing implementation of SQL.

To do so, we define translations of conditions and queries by
mutual induction. 
 Translations $\trtfrom{\tinytwovl}{\cdot},\, \trffrom{\tinytwovl}{\cdot}$ on conditions $\theta$
	ensure 
	\begin{align*}
	\semtwovl{\theta}_{D,\eta} = \true &\text{ if and only if } \semthreevl{\trtfrom{\tinytwovl}{\theta}}_{D,\eta} = \true \\
	\semtwovl{\theta}_{D,\eta} = \vf &\text{ if and only if } \semthreevl{\trffrom{\tinytwovl}{\theta}}_{D,\eta} = \true
	\end{align*}
	(note that $\semtwovl{\theta}_{D,\eta}$ produces
	only \vt\ and \vf). 
		Then 
	we go from  $E$ to $\trtosqlfrom{\tinytwovl}{E}$ by inductively replacing each condition $\theta$ with $\trtfrom{\tinytwovl}{\theta}$.

The full details of the translations are in
Figure~\ref{fig:trtf}. 

\newsavebox{\twotothree}
\sbox{\twotothree}{%
	\parbox{\textwidth-10pt}
    {%
		\begin{center}
			\textbf{Basic conditions}
			\begin{align*}
			\trt{\true} &\df \true &  
			\trf{\true} &\df \false \\
			\trt{\false} &\df \false&
			\trf{\false} &\df \true \\
			\trt{\isnul(t)} &\df \isnul(t) &
			\trf{\isnul(t)} &\df \neg \isnul(t)\\
			\trt{\brt\eq\brt'} &\df \brt\eq\brt'&
			\trf{\brt\eq\brt'}& \df \bigvee_{i=1}^n  \Big(\neg({t_i \eq t'_i}) \vee \isnul(t_i)  \vee \isnul(t'_i) \Big)  \\
			\trt{P(\brt)}&\df P(\brt) & 
			\trf{P(\brt)}&\df \neg P(\brt) \vee \bigvee_{i=1}^n  \isnul(t_i) \\
			\trt{\brt \in E} &\df \brt \in G &
			 	{\trf{\brt \in E} }&  \df \neg( \brt \in G) \vee  
\bigvee_{i=1}^n \isnul(t_i) 
			\\&& \text{where } 
			 \brt \df& (t_1,\cdots, t_n)
		 \\
			\trt{\isempty(E)}& \df \isempty(G)&\trf{\isempty(E)}& \df \neg \isempty(G)\\
			\trt{t\, \omega\, \any(E)}
			&\df 
			t\, \omega\, \any(G) 
			&  
			\trf{t\, \omega\, \any(E)}
			&\df
			\isempty (\sigma_{  \neg \theta
			}(G))
			\\ 
			\trt{t\, \omega\, \all(E)}&\df  t\, \omega\, \all(G)&  
			\trf{t\, \omega\, \all(E)}&\df  
			\neg \isempty( \sigma_{ \theta }(G) )\\
			&&\text{where } \ell(G) &\df N \text{ and }  \theta\df \isnul(t) \vee \isnul(N) \vee (\neg \isnul(t) \wedge \neg \isnul (N) \wedge \neg t\,\omega\,N)
			\end{align*}
			\textbf{Composite conditions}
			\begin{multline*}
			\trt{\theta_1 \vee \theta_2}
			\df \trt{\theta_1} \vee \trt{\theta_2} \ \ 
			\trf{\theta_1 \vee \theta_2}
			\df \trf{\theta_1} \wedge \trf{\theta_2} \ \
			\trt{\theta_1 \wedge \theta_2}\df \trt{\theta_1} \wedge \trt{\theta_2} \ \ \\
			\trf{\theta_1 \wedge \theta_2}\df \trf{\theta_1} \vee \trf{\theta_2} \ \
			\trt{\neg \theta} \df \trf{\theta}  \ \ \trf{\neg \theta} \df  \trt{\theta}
			\end{multline*} 
		\end{center}
	}%
}%

\newcommand{\twotothreefig}{%
	\begin{figure*}[t]
		\centering
		\fbox{\usebox{\twotothree}}
		\caption{\mbox{$\utof$} semantics to SQL semantics}
		\label{fig:twotothree}
	\end{figure*}
}%

\newsavebox{\trtdef}
\sbox{\trtdef}{%
	\parbox{\columnwidth-7pt}{%
	    \begin{tabular}{ll}
	        \textbf{Basic:} &$ 
		\trtfrom{\tinytwovl}{\theta} \df \theta 
		 \text{ for } \theta  \df
\true \, | \, \false \, | \, \isnul(t) \, | \, 
t \, \omega \, t'$
\\
	&$\trtfrom{\tinytwovl}{\isempty(E)} \df \isempty(\trtosqlfrom{\tinytwovl}{E})$ \\
	&$\trtfrom{\tinytwovl}{\brt \, \omega \, \any(E)} \df \brt \, \omega \, \any(\trtosqlfrom{\tinytwovl}{E}) $\\
	&$\trtfrom{\tinytwovl}{\brt \, \omega \, \all(E)} \df \brt \, \omega \, \all(\trtosqlfrom{\tinytwovl}{E}) $\\
	&\\
	        \textbf{Comp.:} &$
	        	\trtfrom{\tinytwovl}{\theta_1 \vee \theta_2} \df	\trtfrom{\tinytwovl}{\theta_1} \vee \trtfrom{\tinytwovl}{\theta_2}$\\&
		$\trtfrom{\tinytwovl}{\theta_1 \wedge \theta_2} \df	\trtfrom{\tinytwovl}{\theta_1} \wedge \trtfrom{\tinytwovl}{\theta_2}$\\
		&$\trtfrom{\tinytwovl}{\neg \theta}\df \trffrom{\tinytwovl}{\theta}$\\
	    \end{tabular}
	}%
}%

\newsavebox{\trfdef}
\sbox{\trfdef}{%
\parbox{\columnwidth-7pt}{%
		\begin{tabular}{ll}	
	\textbf{Basic:} &
$\trffrom{\tinytwovl}{\theta}\df \neg \theta \text{ for } \theta \df \true \,|\, \false \,|\, \isnul(t)$\\
& $\trffrom{\tinytwovl}{t\, \omega \, t' } \df
\isnul(t) \vee \isnul(t') \vee \neg t\, \omega \, t'  $ \\
&$\trffrom{\tinytwovl}{\isempty(E)} \df \neg \isempty(\trtosqlfrom{\tinytwovl}{E})$\\
&		$\trffrom{\tinytwovl}{\brt\, \omega\, \any(E)}
		\df
		\isempty (\sigma_{  \neg \theta
		}(\trtosqlfrom{\tinytwovl}{E}))
		$\\ 	&$\trffrom{\tinytwovl}{\brt\, \omega\, \all(E)}\df  
		\neg \isempty( \sigma_{ \theta }(\trtosqlfrom{\tinytwovl}{E}) )$\\
  &\hfill where $		
\theta \df \trffrom{\tinytwovl}{ \brt \,\omega\, \ell(E)}$
\\&\\
		\textbf{Comp.:}&$
		\trffrom{\tinytwovl}{\theta_1 \vee \theta_2}
		\df \trffrom{\tinytwovl}{\theta_1} \wedge \trffrom{\tinytwovl}{\theta_2}$\\
		&$ \trffrom{\tinytwovl}{\theta_1 \wedge \theta_2}\df \trffrom{\tinytwovl}{\theta_1} \vee \trffrom{\tinytwovl}{\theta_2}$
		\\&$
		\trffrom{\tinytwovl}{\neg \theta}\df  \trtfrom{\tinytwovl}{\theta}$
		\end{tabular}
	}%
}%

\newcommand{\figtrfdef}{%
	\begin{figure}[h]
		\centering
		\fbox{{\usebox{\trtdef}}}
		\fbox{{\usebox{\trfdef}}}
		\caption{$\trtfrom{\tinytwovl}\cdot$ and   $\trffrom{\tinytwovl}\cdot$ of basic and composite conditions}
		\label{fig:trtf}
	\end{figure}
}%

\figtrfdef

\begin{example}
	\label{uftosql-ex}
	We now look at translations of queries $Q_1$--$Q_5$ of Example~\ref{queries-ex}. That is, suppose these queries have been
	written assuming the two-valued $\twovl$ semantics; we show how they
	would then look in conventional SQL. To start with, queries $Q_2,
	Q_3$, and $Q_4$ remain unchanged by the translation.

The query 	$\trtosqlfrom{\tinytwovl}{Q_1}$ is $\sigma_{\isnul(R.A)\vee \neg (R.A \in \sigma_{\neg \isnul(S.A)}S)}(R)$. 
	In SQL, this is equivalent to
\begin{sql}
SELECT R.A FROM R 
WHERE R.A IS NULL OR R.A NOT IN 
   (SELECT S.A FROM S WHERE S.A IS NOT NULL)
\end{sql}

In  $\trtosqlfrom{\tinytwovl}{Q_5}$,
	the  condition $(c\_a > 0) \wedge 
	\neg (c\_c \in \pi_{o\_c}(O))$ in the
	subquery is translated by   
	$\trtfrom{\tinytwovl}{\cdot}$ as $$
	(c\_a > 0) \wedge 
	\Big(\isnul(c\_c)\vee \neg \big(c\_c \in \sigma_{ \neg \isnul(o\_c)
	}(\pi_{o\_c}O)\big)\Big)$$ 
	which is then used in the aggregate subquery $Q_{agg}$ (see
	details in Example \ref{queries-ex} at the end of Section \ref{sec:semp}); the rest of
	the query does not change.
	In SQL,  these are translated into
	additional \sqlkw{IS NULL} and \sqlkw{IS NOT NULL}
	conditions in the \sqlkw{WHERE} of the aggregate query:


\begin{sql}
WHERE c_acctbal > 0.0 AND (c_custkey IS NULL OR 
  c_custkey NOT IN (SELECT o_custkey FROM orders 
                    WHERE o_custkey IS NOT NULL))
\end{sql}

	This translation of $Q_5$ makes no extra assumptions
	about the schema. Having additional information (e.g., 
	that
	\lstinline{c_custkey} is the key of \lstinline{customer}) simplifies
	translation even further;
	see Section \ref{sec:coincide}.
\end{example}

From $\semtwovlsyneq{\,}$ to $\semsql{\,}$, we define the translation $\trtosqlfrom{\syneq}{\,}$ that specifies how to take a query $E$
written under the syntactic equality semantics and translate it into a query
$\trtosqlfrom{\syneq}{E}$ where $\semtwovlsyneq{E}_D = \semthreevl{\trtosqlfrom{\syneq}{E}}_D$ for every
 $D$. Similarly to before, we define translations of conditions and queries by
mutual induction such that
 translations $\trtfrom{\syneq}{\cdot},\, \trffrom{\syneq}{\cdot}$ on conditions $\theta$
	ensure 
	\begin{align*}
	\semtwovlsyneq{\theta}_{D,\eta} = \true &\text{ if and only if } \semthreevl{\trtfrom{\syneq}{\theta}}_{D,\eta} = \true \\
	\semtwovlsyneq{\theta}_{D,\eta} = \vf &\text{ if and only if } \semthreevl{\trffrom{\syneq}{\theta}}_{D,\eta} = \true
	\end{align*}
and
	we go from  $E$ to $\trtosqlfrom{\syneq}{E}$ by inductively replacing each condition $\theta$ with $\trtfrom{\syneq}{\theta}$.
The full details of the translations are in
Figure~\ref{fig:trtfsyneq}.


\newsavebox{\trtsyneqdef}
\sbox{\trtsyneqdef}{
\parbox{\columnwidth-14pt}
{
\begin{tabular}{ll}
  \textbf{Basic:} &$ 
		\trtfrom{\syneq}{\theta} \df \theta 
		 \text{ for } \theta  \df
\true \, | \, \false \, | \, \isnul(t) $
\\
& $\trtfrom{\syneq}{t\,\omega\,t'} \df \neg\isnul(t) \wedge
        \neg \isnul(t')\wedge t\,\omega \, t'$\\
        \multicolumn{2}{r}{\text{for $\omega \in \{<,>,\neq\}$}}
\\&
$\trtfrom{\syneq}{t\,\omega \, t'} \df \left( \isnul(t) \wedge
         \isnul(t') \right) \vee$\\&\hfill$ 
        \left(\neg\isnul(t) \wedge
        \neg \isnul(t')\wedge t\,\omega \, t'\right)
    $
\\
\multicolumn{2}{r}{\text{for $\omega \in \{\le,\ge,\eq\}$}}\\
	&$\trtfrom{\syneq}{\isempty(E)} \df \isempty(\trtosqlfrom{\syneq}{E})$ \\
	&$\trtfrom{\syneq}{\brt \, \omega \, \any(E)} \df  \neg \isempty\left( \sigma_{\theta} \left( \trtosqlfrom{\syneq}{E} \right)\right)$\\
	&$\trtfrom{\syneq}{\brt \, \omega \, \all(E)} \df 
 \isempty\left( \sigma_{\neg \theta} \left( \trtosqlfrom{\syneq}{E} \right)\right)$\\
 &\hfill where $\theta \df \trtfrom{\syneq}{\brt \,\omega\,\ell({E})}$
 \\
	&\\
	        \textbf{Comp.:} &$
	        	\trtfrom{\syneq}{\theta_1 \vee \theta_2} \df	\trtfrom{\syneq}{\theta_1} \vee \trtfrom{\syneq}{\theta_2}$\\&
		$\trtfrom{\syneq}{\theta_1 \wedge \theta_2} \df	\trtfrom{\syneq}{\theta_1} \wedge \trtfrom{\syneq}{\theta_2}$\\
		&$\trtfrom{\syneq}{\neg \theta}\df \trffrom{\syneq}{\theta}$\\
\end{tabular}
}
}

\newsavebox{\trfsyneqdef}
\sbox{\trfsyneqdef}{
\parbox{\columnwidth-14pt}
{
		\begin{tabular}{ll}	
	\textbf{Basic:} &
$\trffrom{\syneq}{\theta}\df \neg \theta \text{ for } \theta \df \true \,|\, \false \,|\, \isnul(t)$\\
& $\trffrom{\syneq}{t\, \omega \, t' } \df
 \isnul(t) \vee \isnul(t') \vee $\\&\hfill$\left(\neg\isnul(t) \wedge\neg \isnul(t')\wedge \neg t\,\omega \, t'\right)
 $ \\
     \multicolumn{2}{r}{\text{for $\omega \in \{<,>,\neq\}$}}
\\
& $\trffrom{\syneq}{t\, \omega \, t' } \df(\isnul(t)\wedge  \neg \isnul(t'))\vee $\\&\hfill$
    \vee \left(\neg\isnul(t) \wedge\neg \isnul(t')\wedge \neg t\,\omega \, t'\right)$
\\&  \hfill$\vee
   ( \isnul(t')  \wedge \neg \isnul(t))$\\
     \multicolumn{2}{r}{\text{for $\omega \in \{\le,\ge,\eq\}$}}
\\
&$\trffrom{\syneq}{\isempty(E)} \df \neg \isempty(\trtosqlfrom{\syneq}{E})$\\
&	$\trffrom{\syneq}{\brt\, \omega\, \any(E)}
	\df
		\isempty (\sigma_{   \theta
		}(\trtosqlfrom{\syneq}{E}))
		$\\ 
		&$\trffrom{\syneq}{\brt\, \omega\, \all(E)}\df  
		\neg \isempty( \sigma_{ \neg \theta }(\trtosqlfrom{\syneq}{E}) )$\\
&\hfill where $		
\theta \df \trtfrom{\syneq}{ \brt \,\omega\, \ell(E)}$
\\&\\
		\textbf{Comp.:}&$
		\trffrom{\syneq}{\theta_1 \vee \theta_2}
		\df \trffrom{\syneq}{\theta_1} \wedge \trffrom{\syneq}{\theta_2}$\\
		&$ \trffrom{\syneq}{\theta_1 \wedge \theta_2}\df \trffrom{\syneq}{\theta_1} \vee \trffrom{\syneq}{\theta_2}$
		\\&$
		\trffrom{\syneq}{\neg \theta}\df  \trtfrom{\syneq}{\theta}$
		\end{tabular}
}
}
\newcommand{\figtrtfsyneqdef}{%
	\begin{figure}[h]
		\centering
  \fbox{{\usebox{\trtsyneqdef}}}
  \fbox{{\usebox{\trfsyneqdef}}}
		\caption{$\trtfrom{\syneq}{\cdot}$ and   $\trffrom{\syneq}{\cdot}$ of basic and composite conditions}
		\label{fig:trtfsyneq}
	\end{figure}
}%
\figtrtfsyneqdef

\begin{example}
	\label{uftosql-ex}
 \sloppy{	Following the previous example, we look at the translation $\trtosqlfrom{\syneq}{\cdot}$ of queries $Q_1$--$Q_5$.
In these translations we often see the condition of the form 
\begin{multline*}
\theta[t,t']=\left( \isnul(t) \wedge \isnul(t')\right) \vee\\ 
\left( \neg\isnul(t) \wedge \neg \isnul(t') \wedge t\eq t'\right).
\end{multline*}
 Query $Q_1$, which is equivalent to $\sigma_{\neg \left(R.A \eq \any(S) \right)}(R)$, is  translated as  
 $ \sigma_{\isempty \left( \sigma_{\theta[R.A,S.A]}(S) \right)}(R)$.
Query $Q_2$ is translated into the same expression.
Query $Q_3$ is translated as $$
\epsilon\Big(\pi_{X.A}\big(\sigma_{\theta[X.A,Y.A]}\big(\rho_{R.A\to X.A}(R) \times 
\rho_{R.A\to Y.A}(R)\big)\big)\Big)\,.
$$ 
While $Q_4$ remains unchanged, in  the
	subquery of ${Q_5}$, the  condition $(c\_a > 0) \wedge 
	\neg (c\_c \in \pi_{o\_c}(O))$ 
is translated to 
$$
    \left(\neg \isnul(c\_a ) \wedge c\_a > 0 \right)\wedge
    \left(\isempty (\sigma_{\theta[o\_c,c\_c]}(O))
    \right)\,.
$$
}
\end{example}
\noindent
Notice that the size of expressions $\trtosqlfrom{\tinytwovl}{E}$ and $\trtosqlfrom{\syneq}{E}$  is indeed linear in $E$.

\section{Restoring expected optimizations}
\label{sec:restore}
Recall queries $Q_1$ and $Q_2$ from the introduction. Intuitively, one
expects them to be equivalent: indeed, if we remove the \sqlkw{NOT}
from both of them, then they are equivalent. And it seems
that if conditions $\theta_1$ and $\theta_2$ are equivalent, then so must be
$\neg\theta_1$ and $\neg\theta_2$. So what is going on there?

Recall that the effect of the \sqlkw{WHERE} clause is to keep tuples
for 
which the condition is evaluated to \vt. So equivalence of conditions
$\theta_1$ and $\theta_2$, from SQL's point of view, means
$\semsql{\theta_1}_{D,\eta}=\vt  \LRA\   \semsql{\theta_2}_{D,\eta}=\vt$ 
for all $D$ and $\eta$. Of course in 2VL this is the same as stating that
$\semsql{\theta_1}_{D,\eta}=\semsql{\theta_2}_{D,\eta}$ due to the 
 fact that there are only two mutually exclusive truth values. In 
 3VL this is not the case however: we can have non-equivalent conditions that evaluate to \vt\ at the same time. 
 \OMIT{ 
 For example $\theta_1\df (A\eq 1)$ and $\theta_2\df (A\eq 1) \wedge \neg\isnul(A)$ are such; in particular $\neg\theta_1$ and $\neg\theta_2$ do not evaluate to \vt\ at the same time. 
 }

With the two-valued semantics eliminating this problem, 
we restore many query equivalences. 
It is natural to 
assume them for granted even though they are not true under 3VL, 
perhaps accounting for some typical
programmer mistakes in SQL \cite{celko,BrassG06}. In terms of \sqlra\
expressions, these equivalences are as follows. 
\newcommand{\proprestequiv}{The following equivalences hold, $\star\in \{\twovl, \syneq \}$:
\begin{enumerate}[
labelindent=0pt]
\item $\semtwovlgen{\sigma_{\theta}(E)}_{D,\eta} =
\semtwovlgen {E \setminus
\sigma_{\neg \theta}(E)}_{D,\eta} $
\item
$\semtwovlgen{\brt\in {E}}_{D,\eta} = \false $ if and only if $\semtwovlgen{\sigma_{\brt \eq \ell(E) }(E)}_{D,\eta} = \emptyset$
\item
$\semtwovlgen{\brt  \,\omega \, \any(E)}_{D,\eta}  =\false$ if and only if $\semtwovlgen{\sigma_{\brt \, \omega\,\ell(E) }(E)}_{D,\eta} = \emptyset$ 
\item
$\semtwovlgen{\brt \, \omega \, \all(E)}_{D,\eta} = \true$ if and only if $\semtwovlgen{\sigma_{ \neg(\brt \, \omega\,\ell(E)) }(E)}_{D,\eta} = \emptyset$ 
\end{enumerate}
	for 
 every
 \sqlrarec expression $E$, tuple
	$\brt$ of terms, 
	condition $\theta$, database $D$, and environment $\eta$.

Neither of those is true in general under SQL's 3VL semantics.
}
\begin{proposition}\label{prop:restequiv}
\proprestequiv
\end{proposition}

\section{Two-valued semantics for free}
\label{sec:coincide}
Theorem \ref{thm:main} shows that every query written under 2VL semantics can be translated into a query that runs on existing RDBMSs and produces the same result. But ideally we want \emph{the same} query to produce the right result, without any modifications. We now show that this happens very often, for a very large class of queries, including majority of benchmark queries used to evaluate RDBMSs. A key to this is the fact some attributes cannot have \NULL\ in them, in particular those in primary keys and those declared as \sqlkw{NOT NULL}. 

We provide a sufficient condition that a query produces the same result under the 2VL and 3VL semantics, for a given list of attributes that cannot be \NULL. It is an easy observation that this equivalence in general is undecidable; hence we look for a sufficient condition. It is defined in two steps. The first tracks attributes in outputs of \sqlrarec\ queries that are {\em nullable}, i.e., can have \NULL\ in them. The second step restricts nullable attributes in queries. 

\OMIT{ 
The translation schemes described in Section~\ref{sec:tosql} allow the programmer to write queries under two-valued semantics, and translate them into an equivalent runable version on existing DBMS. 
These translations, however, as seen in Example~\ref{uftosql-ex}, can be trivial. 
In this section we analyze when it is the case. 
We note that finding a condition which is both sufficient and necessary can be shown to be undecidable.
Thus, we focus on providing a sufficient condition. We do that in two steps. In the first, we track attributes that may contain nulls; and in the second, we restrict the usage of these attributes under negation (within conditions).  
}

\subsubsection*{Tracking nullable attributes}
We define recursively the sequence $\nulable{E}$ of attributes of an 
expression $E$; those may have a \NULL\ in them; others are guaranteed not to have any. We assume that $\nulable{R}$ for a base relation $R$ is defined in the schema: it is the subsequence of attributes of $R$ that are not part of $R$'s primary key nor are declared with \sqlkw{NOT NULL}. Others are as follows:
\begin{itemize}
	\item
	$	\nulable {\epsilon(E)} = \nulable{\sigma_{\theta}E} \df \nulable{E}$;
	\item
	$\nulable{E_1 \times E_2} \df
	\nulable{E_1} \cdot \nulable{E_2}$; 
	\item For $\cup$ and $\cap$, assume that $\ell(E_1)=A_1\cdots A_n$ and $\ell(E_2)=B_1\cdots B_n$. Then $\nulable{E_1\op E_2} = A_{i_1}\cdots A_{i_k}$ where $i_j$ is on the list if: 
	\OMIT{
	(for $\cup$) at least one of  $A_{i_j}\in\nulable{E_1}$ and or $B_{i_j}\in\nulable{E_2}$ is true, or (for $\cap$)  both $A_{i_j}\in\nulable{E_1}$ and or $B_{i_j}\in\nulable{E_2}$ are true;}
	for $\cup$, either  $A_{i_j}\in\nulable{E_1}$ or $B_{i_j}\in\nulable{E_2}$, and for $\cap$  both $A_{i_j}\in\nulable{E_1}$ and $B_{i_j}\in\nulable{E_2}$;
	\item
	$\nulable{E_1\setminus E_2} \df \nulable{E_1}$;
	\item
	$\nulable{\mu R. E_1 \op E_2} = \nulable{ E_1 \cup E_2} $; 
	\item 
	$\nulable {\pi_{t_1,\cdots,t_m}(E)} \df t_{i_1} \cdots t_{i_k}$ where $t_{i_j}$'s are those terms 
that mention names in $\nulable{E}$.
\sloppy
	\item $
	\nulable{\Group_{\bar M,\langle F_1(N_1),\cdots, F_m(N_m)\rangle}(E)} \df \bar{M'} \bar{F'}$
	where $ \bar{M'}$ is the sequence obtained from $\bar M$ by keeping names that are in $\nulable{E}$, and $\bar{F'}$ is obtained from $F_1(N_1),\cdots, F_m(N_m)$ by keeping $F_i(N_i)$ whenever $N_i$ is in $\nulable{E}$.
	In the last two rules, if renamings are specified, names are changed accordingly.  
\end{itemize}
\subsubsection*{Restricting the nullable attributes}
Their use is restricted under negation in selection conditions. 
We say that $\sigma_{\theta}(E)$  is \emph{null-free} if for every sub-condition of $\theta$ of the form $\neg \theta'$ the following hold:
\begin{itemize}
	\item the constant $\NULL$ does not appear in $\theta'$;
	\item for every atomic condition  $\brt\,\omega\,\brt'$  in $\theta'$, no name in $\brt,\brt'$  is in $\nulable{E}$; 
\item for every atomic condition $\brt\in F, \, \brt\,\omega\, \any(F)$ or 	$\brt\,\omega\, \all(F)$ in $\theta'$, the set $\nulable{F}$ is empty and no name in $\brt$ is in $\nulable{E}$.
	\OMIT{\item
	 for every atomic condition $\brt\eq \brt'$ in $\theta'$, no name in $\brt,\brt'$ is in $\nulable{E}$.}
\end{itemize}
\newcommand{\thmtrivtrans}{	Let $E$ be an \sqlrarec\ expression.
	If every subexpression of $E$ of the form $\sigma_{\theta}(F)$ is null-free, then 
	\[
	\semtwovl{E}_{D,\eta} \ = \ \semtwovlsyneq{E}_{D,\eta} \ = \ \semsql{E}_{D,\eta}\,.
	\]}
\begin{theorem}\label{thm:trivtrans}
\thmtrivtrans
\end{theorem}
We clarify here that by a subexpression we mean expressions that are given by subtrees of the parse-tree of an expression. 
In particular this includes subexpressions of queries used in conditions. For example, in an expression $E=\sigma_{A \in E_1}(E_2)$, subexpressions of both $E_1$ and $E_2$ are also subexpressions of $E$.

\begin{example}
	Consider the queries from our running example. 
	Theorem~\ref{thm:trivtrans} applies to $Q_1,Q_4$ if $R.A$ is a key, and to $Q_2,Q_3$ if both $R.A$ and $S.A$ are keys in $R$ and $S$ respectively.
	In query $Q_5$ from our running example
	we have the condition 
	$(c\_a > 0) \wedge 
	\neg (c\_c \in \pi_{o\_c}(O))$.
	Note that if $c\_c$ is a key it is not in $\nulable{Q_5}$.
	If, in addition,  $o\_c$ is specified as \sqlkw{NOT NULL} in table $O$,   then Theorem~\ref{thm:trivtrans} says that for $Q_5$ its SQL and each of the two-valued semantics $\semtwovl{\cdot}$ and $\semtwovlsyneq{\cdot}$ coincide. 
\end{example}

\section{Applicability of 2VL semantics} 
\label{sec:survey}
To gauge the level of applicability of our results, we answer two questions here: (a) how often do the 3VL and 2VL semantics cooincide, so the user can safely forget the {\em unknown}? and (b) when 2VL and 3VL semantics differ, which one is preferred by users?

\subsection*{How often do the semantics coincide?}

To answer this, we look at popular relational performance benchmarks: TPC-H \cite{tpch} containing 22 queries, and TCP-DS \cite{tpcds} containing 99 queries. 
A meticulous analysis of queries in those benchmarks shows that the following satisfy conditions of Theorem \ref{thm:trivtrans}:
\begin{itemize}
    \item {\em All} of TPC-DS queries;
    \item 21 out of 22 (i.e., 95\%)  TPC-H queries.
\end{itemize}
Thus, out of 121 benchmark queries, only one (Q16 of TPC-H) failed the conditions. It means that 120 of those 121 queries produce the same results under 2VL and 3VL semantics. These benchmarks were constructed to represent typical workloads of RDBMSs, meaning that many queries will not be affected by a switch to 2VL.





\subsection*{Which one is preferred by users?}

For some queries, as we have seen, 3VL and 2VL do differ. The lack of those in benchmarks might be partly explained by the fact that those queries are written by experienced programmers who tend to avoid \NULL\ pitfalls. When such queries do occur, is it more natural to expect SQL programmers to follow 3VL or 2VL?

To provide a preliminary answer to this question, we designed a short 10-question user survey. It should be noted that this approach is very common in social sciences, but in our field socio-technological aspects perhaps do not get the attention they deserve, at least for forming research agenda (with a few exceptions though such as \cite{vldb2017tamer,vldb22}). This survey is intended to be a preliminary one, to gauge the level of potential applicability.

The survey started with queries where 3VL vs 2VL makes no difference and asked if users agree with SQL's output. It then showed three queries with different 3VL and 2VL results 
and asked users which one they preferred. It then showed three pairs of queries equivalent under 2VL but not 3VL and asked users whether they want these queries to be equivalent. Finally, it showed a foreign key constraint involving nulls, and asked whether it should hold.

Of 57 received responses, 81\% came from database practitioners and 19\% from academics. The results are shown in Figure \ref{survey-fig}. The first column is for queries where results coincide (i.e., the 2VL column here is the same as the 3VL column). The next three columns are about outputs of queries, the following three are about query equivalences, and the last one about a foreign key constraint.

\begin{figure}[h]
\begin{center} 
\includegraphics[scale=0.20]{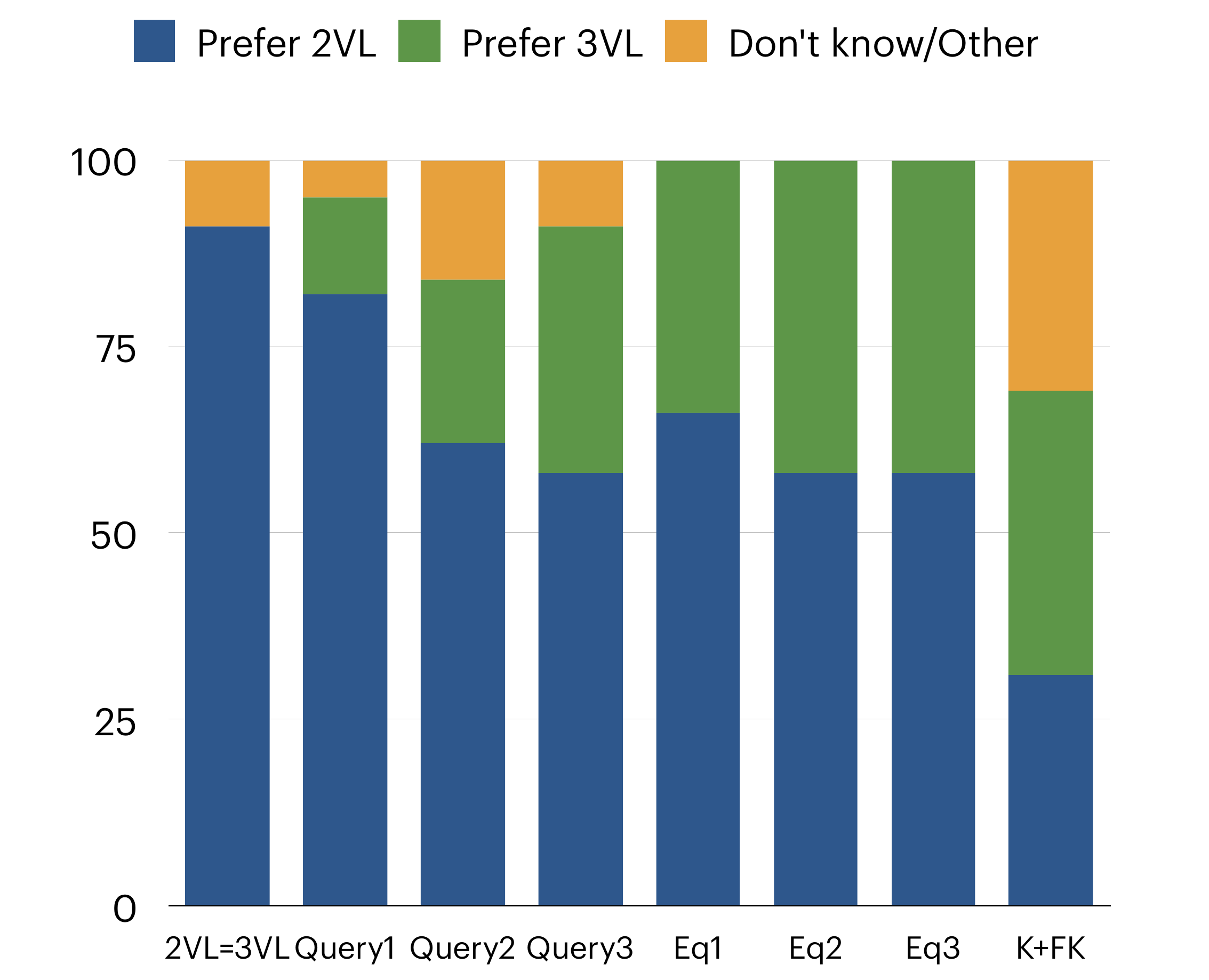}
\vspace*{-2mm}
\caption{Results of the user survey}
\label{survey-fig}
\end{center}
\end{figure}

To summarize:
\begin{itemize}
    \item When 2VL and 3VL coincide, by a 10-to-1 margin users agree with SQL's behavior.
    \item When 2VL and 3VL do not coincide, by a 3-to-1 margin (on average) users prefer query outputs under 2VL.
    \item For query equivalence, users still prefer 2VL but slighly less convincingly, by 60\% to 40\% on average.
    \item Foreign keys, when SQL switches from {\em true} to not {\em false}, are truly confusing to users who are almost evenly split  between 2VL, 3VL, and "do not know". 
\end{itemize}

We reiterate that these results should not be interpreted as the definitive answer on the right choice of the semantics, but rather as a strong indication of 2VL's feasibility and the need of more extensive user surveys to justify alternatives to (rather than outright replacement of) SQL's 3VL semantics.

\section{Robustness: other many-valued logics}
\label{sec:mvl}
\newcommand{\bI}{\mathbf{I}}
\newcommand{\thmgr}{The $\semtwovlsyneq{\,}$ semantics of
	\sqlrarec\ expressions captures the SQL semantics.}

We now show the robustness of the equivalence result, by proving that 
no other many-valued logic could
have been used in place of SQL's 3VL in a way that would have altered the
expressiveness of the language.
In fact, SQL's 3VL, known well before SQL as 
{\em Kleene's logic} \cite{mvl-book}, is not the only many-valued proposed to handle
nulls; there were others with 3,4,5, and even 6
values \cite{CGL16,Yue91,DBLP:conf/future/JiaFM92,Gessert90}. It is thus natural to ask if using one of those
would give us a more expressive language? We now give the negative answer, extending a partial result from 
\cite{AIJ22} that was proved for first-order logic under several restrictions on the connectives of the logic. We extend it to the full language \sqlrarec, and eliminate previously imposed restrictions. 

A many-valued propositional logic $\mvl$ is given 
by a finite collection $\TV$ of \emph{truth values}
with $\vt,\vf\in\TV$, and a finite set $\Conn$ of \emph{logical
connectives} $\conn: \TV^{\arty(\conn)} \rightarrow \TV$. 
We assume that $\mvl$
includes at least
the usual connectives $\neg, \wedge, \vee$ 
whose restriction to $\{\true,\false \}$
follows the rules of Boolean logic (so that queries on databases without nulls would not produce results that differ from their normal behavior). 

The only condition we impose on $\mvl$ is that $\vee$ and $\wedge$ be
associative and commutative; otherwise we cannot write 
conditions $\theta_1$ \sqlkw{OR} $\cdots$   \sqlkw{OR}
$\theta_k$ and  $\theta_1$ \sqlkw{AND} $\cdots$   \sqlkw{AND}
$\theta_k$ 
without worrying about the order of conditions.
Not having commutativity and associativity
is also problematic for optimizing conditions in \sqlkw{WHERE}, as such optimizations assume Boolean algebra identities. 


A semantics 
$\semmvl{\,}$  of \sqlra\ 
conditions is determined by the semantics of comparisons $t \omega t'$; it then follows the 
connectives of $\mvl$ to express the semantics of complex condition, and the expressions of \sqlra\ and \sqlrarec\ follow the semantics of SQL. 
Such a semantics 
$\semmvl{\,}$ is {\em SQL-expressible for atomic predicates} if:
\begin{enumerate}\itemsep=0pt
\item without nulls, it coincides with 
SQL's semantics $\semsql{\,}$;
\item for each truth value $\tv\in\TV$ and each comparison
$\omega$ there is a condition $\theta_{\omega,\tv}(t,t')$ that evaluates to \vt\ in SQL if and only if 
$t\, \omega\, t'$ evaluates to $\tv$ in $\semmvl{\,}$.
\end{enumerate}

These conditions simply exclude 
pathological situations when conditions like $1 \leq
2$ evaluate to truth values other than \vt, \vf, 
or when conditions
like ``$\NULL\eq n$ evaluates to \vt'' are not expressible in SQL
(say, $\NULL\eq n$ is \vt\ iff the $n$th Turing machine in some enumeration halts on the empty input). 
Anything reasonable
is permitted by being expressible.

\OMIT{
\begin{example}
	Consider the 4-valued logic that contains the truth values $\true, \false$ along with the values $\trv{l},\trv{ul}$ that stands for $\trv{l}$ikely and $\trv{u}$n$\trv{l}$ikely, respectively. 
	We can then define the semantics $\semmvl{t\eq t'}_{\eta}$ as follows
\[
\semmvl{t\eq t'}_{\eta}\df
	\left
	\{
	\begin{matrix}
	\true  &  \semvl{t}_{\eta}, \semvl{t'}_{\eta}\ne \NULL,  \semvl{t}_{\eta}= \semvl{t'}_{\eta}  \\
	\trv{l} &  \text{exactly one of }\semvl{t}_{\eta}, \semvl{t'}_{\eta}\text{ is }\NULL\\
	\trv{ul} &  \text{both }\semvl{t}_{\eta}, \semvl{t'}_{\eta}\text{ are }\NULL\\
	\false &   \semvl{t}_{\eta}, \semvl{t'}_{\eta}\ne \NULL,  \semvl{t}_{\eta}\ne \semvl{t'}_{\eta}
	\end{matrix}
	\right.
\]	
as it is, intuitively, less likely that two missing values are the same compared to the likelihood of a missing value and a non-missing value being the same. 
\LP{Not sure that's a good example but I think it is crucial to have an example here since we do not present a formal definition}
\end{example}
}

\begin{example}
\label{4vl-ex}
Consider a 4-valued logic from \cite{CGL16}. It has truth
values \vt, \vf, \vu\ just as SQL's 3VL, and also a new
value \vs. This value means ``sometimes'': under some
interpretation of nulls the condition is true, but under some 
it is false. The semantics is then defined in the same
way as SQL's semantics except \vu\ is replaced by \vs.

In this logic the unknown \vu\ appears when one uses complex
conditions. Suppose conditions $\theta_1$ and $\theta_2$
both evaluate to \vs. Then there are interpretations of nulls
where each one of them is true, and when each one of them is false,
but we cannot conclude that there is an interpretation where both are
true. Hence $\theta_1\wedge\theta_2$ evaluates to \vu\ rather
than \vs. For full truth tables of this 4VL, see \cite{CGL16}. 
\end{example}

\newcommand{\thmmvl}{For a many-valued logic $\mvl$ in which $\wedge$ and
	$\vee$ are associative and commutative, let 
        $\semmvl{\,}$  be a semantics of 
        \sqlra\ or \sqlrarec\ expressions based on $\mvl$. Assume that
	this semantics is SQL-expressible for atomic predicates. 
	Then
	it captures the  SQL semantics.}
\begin{theorem} \label{thm:mvl}
	 \thmmvl
\end{theorem}

Different many-valued semantics are not pure theoretical 
inventions; for example, in MS SQL Server one can switch off the {\tt ansi\_nulls} option to obtain a different MVL of nulls that will be covered by Theorem \ref{thm:mvl}.

\section{Conclusions}
\label{sec:concl}
We showed that one of the most criticized aspects of SQL
and one that is the source of confusion for numerous SQL programmers
-- the use of the three-valued logic -- was not really necessary, and
perfectly reasonable two-valued semantics exist that achieve exactly
the same expressiveness as the original three-valued design. Of course
with all the legacy SQL code based on 3VL, the ultimate goal is not to replace it but rather propose alternatives. Such alternatives can apply not only to SQL but also to newly designed query languages such as GQL for graph data \cite{gpml-gql,icdt23}. 

As for future lines of research,  one is to sharpen the definition of the language to get even closer to everyday SQL.
Another direction is to adapt works like \cite{kennedy-glavic-sigmod19,pods16}
to return results with certainty
guarantees, but under 2VL as opposed to SQL's semantics. 
And most importantly we shall explore avenues of having some of these proposals adapted in relational DBMSs.

\section*{Acknowledgments}
This work was supported by a Leverhulme Trust Research Fellowship; EPSRC grants N023056 and S003800;
and Agence Nationale de la Recherche project ANR-21-CE48-0015 (Verigraph). 
We are grateful to Molham Aref and Paolo Guagliardo for helpful discussions, and to survey respondents for their willingness to participate. 
Part of this work was done while both authors were affiliated with ENS, PSL University, and while second author was affiliated with CNRS, LIGM, Universit\'{e} Gustave Eiffel. 

\balance

\bibliographystyle{abbrv}
\bibliography{main}

\onecolumn

\appendix
\parskip=0.27cm
\parindent=0cm
\begin{center}
	{\Huge Appendix: Proofs and Formal Settings}
\end{center}

\section{Appendix for Section~\ref{sec:syn}}

\subsection*{Section \ref{sec:semp} - Formal Semantics of $\sqlra$}\label{sec:sem}
\newsavebox{\semt}
\sbox{\semt}{%
	\parbox{\textwidth}{%
		\begin{align*}
		\semjointvl{t}_{\eta}& \df \left\{\begin{matrix}
		\eta(t) & t \in \sN \\ 
		t&  t \in \Const \cup \Num \cup \Null 
		\end{matrix}\right.\\
		\semvl{(t_1,\cdots ,t_m )}_{\eta } &\df (\semvl{t_1}_{\eta}, \cdots, \semvl{t_m}_{\eta} )\\
		\semvl{f(t_1,\cdots ,t_m )}_{\eta } &\df
		\left\{\begin{matrix}
		\NULL & \exists i: \semvl{t_i}_{\eta} = \NULL \\ 
	 f(\semvl{t_1}_{\eta}, \cdots, \semvl{t_m}_{\eta})	& \text{otherwise}
		\end{matrix}\right.
		\end{align*}
	}%
}%
\newcommand{\semfigt}{%
	\begin{figure*}
		\centering{
		\fbox{\usebox{\semt}}}
		\caption{Semantics of Terms}
	\end{figure*}
}%


\newsavebox{\semeqsql}
\sbox{\semeqsql}{%
	\parbox{\textwidth}{%
		\begin{align*}
			\semthreevl{t\,\omega\,t'}_{D,\eta} & \df
			\left
			\{
			\begin{matrix}
				\true  &  \semvl{
		t}_{\eta} , \semvl{
		t'}_{\eta} \ne \NULL,  \semvl{t}_{\eta}\,\omega\, \semvl{t'}_{\eta} \\ 
				\false &  \semvl{
		t}_{\eta}, \semvl{
		t'}_{\eta} \ne \NULL,  \semvl{t}_{\eta}\, \not \omega\, \semvl{t'}_{\eta}\\ 
				\unknwn &  \semvl{t}_{\eta} = \NULL \text{ or }\semvl{t'}_{\eta} = \NULL
			\end{matrix}
			\right.\\
			\text{where }\omega&\in\{ \eq,\noteq, < ,> ,\le ,\ge \}
\\
\semvl{P(\brt)}_{D,\eta}& \df
		\left
		\{
		\begin{matrix}
		\true  & \forall\,i:\,\semvl{
		t_i}_{\eta} \ne \NULL \text{ and } \mathbf{P}(\brt)  \\ 
				\false &  \forall\,i:\,\semvl{
		t_i}_{\eta} \ne \NULL \text{ and } \neg \mathbf{P}(\brt)\\
		\unknwn & \exists\,i:\,\semvl{
		t_i}_{\eta} = \NULL\\
			\end{matrix}
			\right.\\
		\text{where } \brt &
			 \df (t_1,\cdots, t_n)
		\end{align*}
	}%
}%
\newcommand{\semfigeqsql}{%
	\begin{figure*}
		\centering
		\fbox{\usebox{\semeqsql}}
		\caption{SQL's Semantics of Predicates and Equality}
	\end{figure*}
}%

\newsavebox{\semn}
\sbox{\semn}{%
	\parbox{\textwidth}{%
		\begin{align*}
\ell\left(\pi_{t_1 [\shortrightarrow  \! N_1],\cdots, t_m [\shortrightarrow   \!N_m] }(E)\right) & \df \tilde N_1 \cdots \tilde N_m\\
\text{where } \tilde N_i& \df \left\{\begin{matrix}
N_i & \text{ if }[\shortrightarrow \!  N_i]  \\ 
\Name(t_i)  & \text{otherwise}
\end{matrix}\right.\\
\ell\left(\sigma_{\theta}(E)\right) & \df \ell(E)\\
\ell\left( E_1 \times E_2 \right) &\df \ell(E_1)\cdot \ell(E_2)\\
\ell\left( E_1 \op E_2 \right) &\df \ell(E_1) \text{ for } \op\in \{\cup, \cap, \setminus  \} \\
\ell\left(\epsilon(E)\right) & \df \ell(E)\\
\ell ( \Group_{\barN,\langle F_1(N_1)[\shortrightarrow \!  N'_1],\cdots, F_m(N_m)[\shortrightarrow  \! N'_m] \rangle} (E) ) 
&\df
\begin{multlined}[t]
 \barN\cdot\tilde N_1 \cdots \tilde N_m
\end{multlined}
 \\
\text{where } \tilde N_i& \df \left\{\begin{matrix}
N_i & \text{ if }[\shortrightarrow \!  N_i]  \\ 
\Name(F_i(N_i))  & \text{otherwise}
\end{matrix}\right.
		\end{align*}
	}%
}%

\newcommand{\semfign}{%
	\begin{figure*}
		\centering
		\fbox{\usebox{\semn}}
		\caption{Names Assigned to Expressions}
	\end{figure*}
}%

\newsavebox{\semc}
\sbox{\semc}{%
	\parbox{\textwidth}{%
	\textbf{Basic conditions}
	        \begin{align*}
		\semstrvl{\true}_{D,\eta}& \df \true 
                \ \ \ \ 
		\semstrvl{\false}_{D,\eta} \df \false\\
		\semstrvl{\isnul(t)}_{D,\eta}&\df	\left\{\begin{matrix}
		\true  &  \semjointvl{t}_{\eta} = \NULL \\ 
		\false &  \text{otherwise}
		\end{matrix}\right.\\
		\semstrvl{\brt \eq \brt'}_{D,\eta} & \df \bigwedge_{i=1}^{n} \semstrvl{t_i \eq t'_i}_{D,\eta}\\
		\semstrvl{\brt \noteq \brt'}_{D,\eta} & \df
		\neg \semstrvl{\brt \eq \brt'}_{D,\eta} \\
		\text{for } \omega\in \{\le,<,>,\ge \}\,\, 
			\semstrvl{\brt \,\omega\, \brt'}_{D,\eta}&\df
	\bigvee_{0\le i\le n-1} \big(
	\bigwedge_{1\le j \le i} 
	\semvl{t_j  \eq t'_j}_{D,\eta}
	\wedge
	\semvl{t_{i+1}  \, \omega \, t'_{i+1}}_{D,\eta} 
	\big)
	 \\
	 			\text{where } \brt& \df (t_1,\ldots, t_n)\\
			\brt'& \df (t'_1,\ldots, t'_n)\\
	\semstrvl {\brt\, \omega\, \any(E)}_{D,\eta}& \df
		\bigvee_{\brt' \in \semstrvl {E}_{D,\eta}
	}  \semstrvl{\brt\, \omega\, \brt'}_{D,\eta}\\
			\semstrvl{\brt\in E}_{D,\eta}& \df
		\semstrvl {\brt\eq \any(E)}_{D,\eta}
		\\
	\semstrvl {\brt\, \omega\, \all(E)}_{D,\eta}& \df
		\bigwedge_{\brt' \in \semstrvl {E}_{D,\eta}
	}  \semstrvl{\brt\, \omega\, \brt'}_{D,\eta}\\
		\semstrvl {\isempty(E)}_{D,\eta} &\df 
		\left\{\begin{matrix}
		\true  & \semstrvl{E}_{D,\eta} = \emptyset	 \\ 
		\false &   \text{otherwise}
		\end{matrix}\right.\\
		\end{align*}
\textbf{Composite conditions}
\begin{align*}
	\semstrvl{\theta_1 \vee \theta_2}_{D,\eta} &\df \semstrvl{\theta_1 }_{D,\eta}\vee \semstrvl{\theta_2}_{D,\eta}\\
	\semstrvl{\theta_1 \wedge \theta_2}_{D,\eta}& \df \semstrvl{\theta_1 }_{D,\eta}\wedge \semstrvl{\theta_2}_{D,\eta}\\
	\semstrvl{\neg \theta}_{D,\eta} &\df \neg \semstrvl{\theta}_{D,\eta}\\
\end{align*}
\OMIT{
\textbf{Three-valued logic rules}
\begin{center} \hspace*{-1mm}\begin{tabular}{c|ccc}
    $\land$ & \vt & \vf & \vu \\
    \hline %
    \vt & \vt & \vf & \vu \\
    \vf & \vf & \vf & \vf \\
    \vu & \vu & \vf & \vu
  \end{tabular} \hspace*{4mm}
  \begin{tabular}{c|ccc}
    $\lor$ & \vt & \vf & \vu \\
    \hline %
    \vt & \vt & \vt & \vt \\
    \vf & \vt & \vf & \vu \\
    \vu & \vt & \vu & \vu
  \end{tabular}
  \hspace{4mm}
  \begin{tabular}{c|c}
    & $\lnot$\\
    \hline
    \vt  & \vf  \\
    \vf  & \vt  \\
    \vu  & \vu  \\
  \end{tabular}\end{center}}
}
}%

\newcommand{\semfigc}{%
	\begin{figure}[h]
		\centering
		\fbox{\usebox{\semc}}
		\caption{Semantics of Conditions}
	\end{figure}
}%

\newsavebox{\semk}
\sbox{\semk}{%
	\parbox{\textwidth}{
\begin{center}
 \begin{tabular}{c|ccc}
    $\land$ & \vt & \vf & \vu \\
    \hline %
    \vt & \vt & \vf & \vu \\
    \vf & \vf & \vf & \vf \\
    \vu & \vu & \vf & \vu
  \end{tabular}
  \hspace{2cm}
  \begin{tabular}{c|ccc}
    $\lor$ & \vt & \vf & \vu \\
    \hline %
    \vt & \vt & \vt & \vt \\
    \vf & \vt & \vf & \vu \\
    \vu & \vt & \vu & \vu
  \end{tabular}
  \hspace{2cm}
  \begin{tabular}{c|c}
    & $\lnot$\\
    \hline
    \vt  & \vf  \\
    \vf  & \vt  \\
    \vu  & \vu  \\
  \end{tabular}%
\end{center}
}}

\newcommand{\semfigk}{%
	\begin{figure*}
		\centering
		\fbox{\usebox{\semk}}
		\caption{SQL 3-Valued Logic Rules}
		\label{fig:semk}
	\end{figure*}
}%

\newsavebox{\seme}
\sbox{\seme}{%
	\parbox{\textwidth}{%
	\begin{align*}
	\bara \in_k \semstrvl{R}_{D,\eta}& \text{ if } \bara \in_k R^D\\
	\OMIT{
(a_{1},\ldots,a_{m}) \in_k \semstrvl{\pi_{t_1 ,\cdots,t_m }(E)}_{D,\eta}
& \text{ if } 
\text{for every $1\le i \le m$: } 
a_i = \semstrvl{t_i}_{\eta;\eta_{\ell(E)}^{\bar c_j}}
\text{ where }
\bar c_j \in_{k_j}\semstrvl{E}_{D,\eta}\text{ and } 
k = \sum_{j=1}^n k_j
\\}
(a_{1},\ldots,a_{m}) \in_k \semstrvl{\pi_{t_1 ,\cdots,t_m }(E)}_{D,\eta}
& \text{ if } 
\semvl{E}_{D,\eta} = \bag{ \underbrace{\vec{c}_1,\ldots ,\vec{c}_1 }_{k_1\text{ times }},\ldots,  \underbrace{\vec{c}_n,\ldots,\vec{c}_n}_{k_n\text{ times }} }
    \\&k = \Sigma_{j\in I} k_j \text{ where }
    I \df \{j\,\mid\, 1\le j \le n,\, a_1 = \semvl{t_1}_{D,\eta;\eta_{\ell(E)}^{c_{j}}}, \ldots, a_m = \semvl{t_m}_{D,\eta;\eta_{\ell(E)}^{c_{j}}}  \},\, 
    \\
(\bar a_1, \bar a_2)\in_k\semstrvl{ E_1 \times E_2}_{D,\eta} & \text{ if } \bar a_1 \in_{k_1}\semstrvl{ E_1}_{D,\eta},\, \bar a_2\in_{k_2} \semstrvl{ E_2}_{D,\eta},\ k = k_1\cdot k_2
	\\
\bar a \in_k \semstrvl{ E_1 \op E_2}_{D,\eta} & \text{ if } \bar a\in \semstrvl{ E_1}_{D,\eta}, \, \bar a \in  \semstrvl{ E_2}_{D,\eta},\, \op \in \{\cup, \cap ,\setminus \}, \, k= \left\{ \begin{matrix}
	k_1+k_2 & \op \text{ is } \cup \\
	\min\{k_1,k_2\} & \op \text{ is } \cap \\
	\max\{0,k_1-k_2\}&\op \text{ is } \setminus
	\end{matrix}\right.
	\\
\OMIT{		\semstrvl{ \rho_{ \barN\shortrightarrow   \barN'}E}_{D,\eta}
	& \df \semstrvl{E}_{D,\eta}  \\}
		\bara \in_k	\semstrvl{\sigma_{ \theta}(E)}_{D,\eta} &\text{ if }
		\bara \in_k \semstrvl{E}_{D,\eta} \ \text{and}\
		\semstrvl{\theta}_{D,\eta;\eta_{\ell(E)}^{\bara}} = \true
		\\
		\bara \in_1 \semstrvl{ \epsilon (E) }_{D,\eta}
	& \text{ if  } \bara \in_k \semstrvl{ E }_{D,\eta} \text{ and
	} k > 0 
		\\
	( \bara, \bara' )\in_1	\semstrvl{\Group_{\bar M,\langle F_1(N_1),\cdots, F_m(N_m)\rangle} (E)}_{D,\eta}& \text{ if }
	\bara \in_k \semstrvl{\pi_{\barM}(E)}_{D,\eta},\\&
\bara' = ( \langle F_1(N_1)(E'_1),\cdots,
F_m(N_m)(E'_m))
\\
	&\text{ where }
	 E'_i = \semvl{\pi_{N_i}( \sigma_{ \barM \eq \bara} (E))}_{D,\eta}
\text{ and }\\
F(N)(\mathbf E) &\df 
\left\{ 
\begin{matrix}
\Card(\mathbf E) & F = \Count, \ N=\textasteriskcentered \\
F\left(\notnull{\mathbf E}\right)&\text{otherwise}
\end{matrix}\right.
\OMIT{\\
 (\bara,n) \in_k 
 \langle \Count(\textasteriskcentered)
 \rangle
 (\mathbf E)& \text{ if } 
	\bara\in_k \mathbf{E},
	n= \Card(\mathbf{E}) \\
	(\bara,n)\in_k	
	\langle F(N)\rangle(\mathbf E) &\text{ if } 
	 \bara\in_k \mathbf{E}, 
	 n= F\Big(
		 \notnull{(\semstrvl{ \pi_{N}(\mathbf E)}_{D,\eta})}\Big) 
}	\end{align*}
	}%
}%

\newcommand{\semfige}{%
	\begin{figure*}[t]
		\centering
		\fbox{\usebox{\seme}}
		\caption{Semantics of Well-formed Expressions}
		\label{fig:seme}
	\end{figure*}
}%

\begin{figure*}
\begin{minipage}[c]{1\textwidth}
\begin{subfigure}{1\textwidth}
  \centering
   	\fbox{\usebox{\semn}}
  \caption{Names}
  \label{fig:semn}
\end{subfigure}\\
\begin{subfigure}{1\textwidth}
  \centering
   	\fbox{\usebox{\semt}}
  \caption{Terms}
  \label{fig:semt}
\end{subfigure}
\\
\begin{subfigure}{1\textwidth}
  \centering
   	\fbox{\usebox{\semeqsql}}
  \caption{Comparisons and Predicates}
  \label{fig:semeqsql}
\end{subfigure}
\\
\begin{subfigure}{1\textwidth}
	\centering
	\fbox{\usebox{\seme}}
	\caption{Expressions}
	\label{fig:seme}
\end{subfigure}
\end{minipage}
\caption{Names of Expressions and Semantics of Terms, Predicates and Expressions}
\label{fig:semantics}
\end{figure*}

\begin{figure*}
	\begin{minipage}[c]{1\textwidth}
		\begin{subfigure}{1\textwidth}
			\centering
			\fbox{\usebox{\semc}}
		\end{subfigure}
	\end{minipage}
	\caption{
		 Semantics of 
	Conditions}
	\label{fig:semc}
\end{figure*}


%

We define the formal semantics of \sqlra\ expressions in
the spirit of \cite{DBLP:journals/pvldb/GuagliardoL17,hottsql,benzaken}.

We define the semantic function
$\semvl{E}_{D,\eta}$
which is, as said before, the result of evaluation of expression $E$ on database $D$
under the {\em environment} $\eta$ where 
$\eta$
provides values
of parameters of the query. 

Recall that every relation $R$ is associated with a sequence $\ell(R)$
of attribute names. 
Just as SQL queries do, every \sqlra\ expression $E$ produces a table whose
attributes 
have names. We start by defining those in 
Fig.~\ref{fig:semn}.  We make the assumption that names do not
repeat; this is easy to enforce with renaming. This differs from SQL
where names in query results can repeat, and this point was rather
extensively discussed in \cite{DBLP:journals/pvldb/GuagliardoL17}.
However, the treatment of repeated names in the definition of the
semantics of SQL queries is completely orthogonal to the treatment of
nulls, and thus we can make this assumption without loss of generality
so as not to clutter the description of our translations with the
complexities coming from treating repeated attributes.

Next, we define the semantics of terms: it is
given by the environment, see Fig.~\ref{fig:semt}.


After that we give the semantics of 
comparison operators as well as predicates of higher arities in Fig.~\ref{fig:semeqsql}. We follow SQL's three valued
logic with true value {\em true} ($\true$), {\em false} ($\false$) and
{\em unknown} ($\unknwn$). The usual SQL's rule is: evaluate a
predicate normally if no attributes are nulls; otherwise return
$\vu$. 
That is, for each predicate $P$, 
we have its
interpretation $\mathbf{P}$ over $\Const \cup \Num$. 

To provide the formal semantics of \sqlra\ expressions, we need some extra
notation.  
We assume that there is a one-to-one function $\Name$ that maps terms into (unique) names (i.e., elements in $\sN$).
Given
$\alpha\in\sN^*$ and $\barN, \barN' \in\sN^*$,  the
sequence  $ \alpha_{\barN\rightarrow \barN'}$ 
obtained from $\alpha$ by replacing each $N_i$ with $N'_i$
where $\barN \df (N_1 , \cdots, N_m)$ and
$\barN' \df (N_1' , \cdots, N'_m)$.

Next, if $\bara \df (a_1,\cdots, a_m)$ is a tuple of values over
$\Num \cup \Const \cup \Null$ and $\barN \df (N_1, \cdots, N_m)$
a tuple of names over $\sN$, 
we denote by $\eta_{\barN}^{\bara}$ the environment 
that maps each name $N_i$ into the value $a_i$.
We say that $\bara$ is \emph{consistent} with  $\barN$ if $\type(N_i)
= \sn$ implies $a_i \in \Num \cup \Null$ and $\type(N_i) = \so$
implies $a_i \in \Const \cup \Null$ for each $i$. 
For two environments $\eta$ and $\eta'$, by $\eta ; \eta'$ we mean $\eta$ overridden by $\eta'$.
That is, $\eta ; \eta'(N) \df \eta(N)$ if $\eta$ is defined on the name $N$ and $\eta'$ is not; otherwise $\eta ; \eta'(N) \df \eta'(N)$.

For a bag $B$, let $\notnull B$ be the same as $B$ but with
occurrences of \NULL\ removed. A tuple is called {\em null-free} if
none of its components is \NULL.

With these, the semantics of expressions is defined in
Fig.~\ref{fig:seme}. Note that we omit the optional parts in the generalized projection and grouping/aggregation as it does not affect the semantics but is reflected only in the names as appear in Figure~\ref{fig:semn}.

Now given an expression $E$ of \sqlra\ and a database $D$, the value of $E$
in $D$ is defined as $\semvl{E}_{D,\emptyset}$  where $\emptyset$ is the
empty mapping (i.e., the top level expression has no free variables). 
\OMIT
{\color{red}
$(a_1,\cdots,a_m)\in_k \semvl{\pi_{t_1,\cdots, t_m}(E)}_{D,\eta}$ if:
\begin{itemize}
    \item 
    $\vec{c}_1 \in_{k_1} \semvl{E}_{D,\eta}, \cdots, \vec{c}_{n} \in_{k_n}  \semvl{E}_{D,\eta}$ 
    and
    for every $\vec{c} \not \in \{ \vec{c}_1,\cdots, \vec{c}_n \} $ it holds that $\vec{c} \not \in\semvl{E}_{D,\eta}$
    \item
    $k = \Sigma_{j\in I} k_j$
    \item
     $I \subseteq \{ 1,\cdots,n \} $ is such that for every $j\in I$ it holds that  
    $a_1 = \semvl{t_1}_{D,\eta;\eta_{\ell(E)}^{c_{j}}}, \cdots, a_m = \semvl{t_m}_{D,\eta;\eta_{\ell(E)}^{c_{j}}}$, and for every $j\not \in I$ there is $1\le m'\le m$ for which $a_{m'} \ne \semvl{t_{m'}}_{D,\eta;\eta_{\ell(E)}^{c_{j}}}$
\end{itemize}
}


\OMIT{
\begin{example}
	\LP{especially fot the complicated grouping with Count* and with some other aggregates}
\end{example}
\begin{example}
	\LP{More than one!}
\end{example}

\LP{a short comment on the well-formed issue? currently the semantics is not defined fully for all of the syntactically possible expressions.}
}

\section{Appendix for Section~\ref{sec:twoval}}

\subsection{Proof of Theorem~\ref{thm:main}}

\repeatresult{theorem}{\ref{thm:main}}{ \thmmain}

In what follows, for each semantics ($\semtwovl{\,}$ and $\semtwovlsyneq{\,}$), we prove the two parts of the theorem both for \sqlra\ expressions and for \sqlrarec\ expressions by presenting both directions of the translations (and ensuring that the translations lie within the appropriate languages).
We note here that in both directions it is straightforward that the size of the resulting expressions is linear in that of the original one.

\begin{lemma}
The $\semtwovl{\,}$ 
semantics of
	\sqlrarec\ expressions, and of \sqlra\ expressions, capture
	their  SQL semantics
$\semsql{\,}$	efficiently.    
\end{lemma}\label{lem:twovl}
\begin{proof}
The proof consists of two directions:
\subsubsection{$\semtwovl{\,}$ to $\semsql{\,}$}
Let $E$ be an  \sqlra expression and let $G \df \trtosql{E}$ be the expression obtained from $E$ as described in Figure~\ref{fig:trtf} in Section~\ref{sec:twoval}.
To show that $\semtwovl{E}_D = \semthreevl{G}_D$ for every database $D$, we prove inductively that
for every database $D$ and every environment $\eta$, the following claims hold:
\begin{itemize}
	\item[$(a)$] $\semtwovl{E}_{D,\eta} = \semthreevl{G}_{D,\eta}$
	\item[$(b)$] $\semtwovl{\theta}_{D,\eta} = \true$ if and only if  $\semthreevl{\tcond}_{D,\eta} = \true$.
	\item[$(c)$]
	$\semtwovl{\theta}_{D,\eta} = \false$ if and only if  $\semthreevl{\fcond}_{D,\eta} = \true$.
\end{itemize}

Before starting, we write explicitly the translations of the following conditions:\begin{itemize}
    \item $\trt{\brt \eq \brt' }= 
    \trt{ \wedge_{i=1}^n (t_i\eq t'_i)} = \wedge_{i=1}^{n} \trt{t_i \eq t'_i} = \wedge_{i=1}^n (t_i\eq t'_i)$
    \item $\trf{\brt \eq \brt' }= 
    \trf{ \wedge_{i=1}^n (t_i\eq t'_i)} = \vee_{i=1}^{n} \trf{t_i \eq t'_i} = \vee_{i=1}^n (\isnul(t_i) \vee \isnul(t'_i) \vee \neg t_i\eq t'_i)$
    \item for $\omega\in\{<,\le,>,\ge\}$, 
    $\trt{\brt \,\omega\, \brt' }=
    \trt{ \bigvee_{i=0}^{n-1} ( (t_1,\cdots, t_i)\eq (t'_1,\cdots, t'_i) \wedge t_{i+1} \,\omega\,t'_{i+1})} = \vee_{i=0}^{n-1} \trt{(t_1,\cdots, t_i)\eq (t'_1,\cdots, t'_i) \wedge t_{i+1} \,\omega\,t'_{i+1}}  =
    \vee_{i=0}^{n-1}( \trt{(t_1,\cdots, t_i)\eq (t'_1,\cdots, t'_i)} \wedge\trt{ t_{i+1} \,\omega\,t'_{i+1}}) = \vee_{i=0}^{n-1}( {(t_1,\cdots, t_i)\eq (t'_1,\cdots, t'_i)} \wedge{ t_{i+1} \,\omega\,t'_{i+1}}
)    $
    \item
    for $\omega\in\{<,\le,>,\ge\}$, 
    $\trf{\brt \,\omega\, \brt' }=
    \trf{ \bigvee_{i=0}^{n-1} ( (t_1,\cdots, t_i)\eq (t'_1,\cdots, t'_i) \wedge t_{i+1} \,\omega\,t'_{i+1})} 
    =
    \bigvee_{i=0}^{n-1} \trf{ (t_1,\cdots, t_i)\eq (t'_1,\cdots, t'_i) \wedge t_{i+1} \,\omega\,t'_{i+1}}
    =
    \bigvee_{i=0}^{n-1} (\trf{ (t_1,\cdots, t_i)\eq (t'_1,\cdots, t'_i)} \vee\trf{ t_{i+1} \,\omega\,t'_{i+1}})
    =
    \bigvee_{i=0}^{n-1} (
    \vee_{j=1}^i (\isnul(t_j) \vee \isnul(t'_j) \vee \neg t_j\eq t'_j)
    \vee
    (\isnul(t_{i+1})\vee\isnul(t'_{i+1})\vee
    \neg t_{i+1} \,\omega\,t'_{i+1}))
    $
\end{itemize}

\paragraph*{Induction base:}
\begin{itemize}
	\item[$(a)$] By definition, 
	$\semtwovl{R}_{D,\eta} = R^D = \semthreevl{R}_{D,\eta}$
	\item[$(b)$]
	\begin{enumerate}
		\item  $\semtwovl{\true}_{D,\eta} \df \true$ and $\semthreevl{\true}_{D,\eta} \df \true$ and thus the claim holds.
		\item  $\semtwovl{\false}_{D,\eta} \df \false$ and $\semthreevl{\false}_{D,\eta} \df \false$ thus the claim vacuously holds.
		\item
		$\semtwovl{ \isnul(t)}_{D,\eta} = \true$ if and only if $\semvl{t}_{\eta} = \NULL$ if and only if $\semthreevl{\isnul(t)}_{D,\eta} = \true$
\item
It holds that 
\[\semtwovl{ t\,\omega\, t'}_{D,\eta} =  \true\] 
if and only if (by definition)
	\[
		\semvl{ t}_{\eta} \ne \NULL, \, 
			\semvl{ t'}_{\eta} \ne \NULL, \text{ and }
	\semvl{ t}_{\eta} \, \omega \, \semvl{ t'}_{\eta}
	\]
	if and only if (by definition)
\[\semvl{ t\,\omega\, t'}_{D,\eta} =  \true\] 
\item
$\semtwovl{\isempty(E)}_{D,\eta} = \true$ if and only if $\semtwovl{E}_{D,\eta} = \emptyset$. By induction hypothesis, this holds if and only if  $\semthreevl{G}_{D,\eta} = \emptyset$. This, in turn, holds if and only if $\semthreevl{ \isempty(G)}_{D,\eta} = \true$.
\item 
It holds that 
\[\semtwovl{\brt \, \omega \, \any(E)}_{D,\eta} = \true\] 
if and only if 
(by definition of the semantics of $\any$)
\[ \bigvee_{\brt'\in \semtwovl{E}_{D,\eta}} \semtwovl{\brt\, \omega \,\brt'}_{D,\eta} = \true\]
if and only if 
(by definition of $\vee$, induction hypothesis $(a)$ and $\trt{t\,\omega\,t'}$)
\[ \bigvee_{\brt'\in \semvl{G}_{D,\eta}} \semvl{\brt\,\omega\, \brt'}_{D,\eta} = \true\]
 if and only if (by definition)
\[\semvl{\brt \, \omega \, \any(G)}_{D,\eta} = \true\] 
 \item
 It holds that
 \[\semtwovl{\brt \, \omega \, \all(E)}_{D,\eta} = \true\] 
 	if and only if 	(by definition of the semantics of $\all$) 
 	\[ \bigwedge_{\brt'\in \semtwovl{E}_{D,\eta}} \semtwovl{\brt\, \omega \,\brt'}_{D,\eta} = \true\]
if and only if 
(by definition of $\wedge$, induction hypothesis $(a)$ and $\trt{t\,\omega\,t'}$)
\[ \bigwedge_{\brt'\in \semvl{G}_{D,\eta}} \semvl{\brt\,\omega\, \brt'}_{D,\eta} = \true\]
 if and only if (by definition)
\[\semvl{\brt \, \omega \, \all(G)}_{D,\eta} = \true\] 
\OMIT{
 	\[\forall t'\in E: \semtwovl{t\eq t'}_{D,\eta} = \true\] 
 	if and only if (by induction hypothesis)
 	\[\forall t'\in G: \semtwovl{t\eq t'}_{D,\eta} = \true\]
 	if and only if (by definition) 	\[\semthreevl{t \, \omega \, \all(G)}_{D,\eta} = \true\].}
	\end{enumerate}
	\item[$(c)$]
\begin{enumerate}
	\item  $\semtwovl{\true}_{D,\eta} \df\true$ and thus the claim is trivial.
	\item  $\semtwovl{\false}_{D,\eta} \df \false$ and $\semthreevl{\false}_{D,\eta} \df\false$ thus the claim holds.
	\item
	$\semtwovl{ \isnul(t)}_{D,\eta} = \false$ if and only if $\semvl{t}_{\eta} \ne  \NULL$ if and only if (by definition) $\semthreevl{\isnul(t)}_{D,\eta} = \false$
	\OMIT{\item It holds that 
	\[\semtwovl{t\eq t' }_{D,\eta} =  \false\] 
	if and only if (by definition)
	\[
	\semvl{t}_{\eta} = \NULL \vee
	\semvl{t'}_{\eta} = \NULL
	\vee
	\semvl{t}_{\eta} \ne \semvl{t'}_{\eta}
	\]
if and only if (by definition)
\[
\semvl{\isnul(t)}_{\eta} = \true \vee 
\semvl{\isnul(t')}_{\eta} = \true \vee
	\semvl{t}_{\eta} \ne \semvl{t'}_{\eta}
\]
if and only if for all $D$
\[
\semthreevl{\isnul(t) \vee \isnul(t')\vee\neg (t\eq t')}_{D,\eta} = \true
\]}
\item\label{a}
It holds that 
\[\semtwovl{ t\,\omega\, t'}_{D,\eta} =  \false\] 
if and only if (by definition)
	\[
		\semvl{ t}_{\eta} = \NULL, \text{ or } 
			\semvl{ t'}_{\eta} = \NULL, \text{ or }
	\semvl{ t}_{\eta} \, \not \omega \, \semvl{ t'}_{\eta}
	\]
	if and only if (by definition)
\[\semvl{ \isnul(t) \vee \isnul(t') \vee \neg t\,\omega\, t'}_{D,\eta} =  \true\] 
\item
$\semtwovl{\isempty(E)}_{D,\eta} = \false$ if and only if $\semtwovl{E}_{D,\eta} \ne \emptyset$. By induction hypothesis $(a)$, this holds if and only if  $\semthreevl{G}_{D,\eta} \ne \emptyset$. This, in turn, holds if and only if $\semthreevl{ \isempty(G)}_{D,\eta} = \false$.
\item 
It holds that 
\[\semtwovl{\brt \,\omega\, \any(E)}_{D,\eta} = \false\] 
if and only if 
(by definition of the semantics of $\any$)
\[ \bigvee_{\brt'\in \semtwovl{E}_{D,\eta}} \semtwovl{\brt\, \omega \,\brt'}_{D,\eta} = \false\]
if and only if 
(by induction hypothesis (a))
\[ \bigvee_{\brt'\in \semvl{G}_{D,\eta}} \semtwovl{\brt\, \omega \,\brt'}_{D,\eta} = \false\]
if and only if (by definition of $vee$)
\[ \text{for all }{\brt'\in \semvl{G}_{D,\eta}}: \semtwovl{\brt\, \omega \,\brt'}_{D,\eta} = \false\]
if and only if (by induction hypothesis (b))
\[ \text{for all }{\brt'\in \semvl{G}_{D,\eta}}:  \semvl{\trf{\brt\, \omega \,\brt'}}_{D,\eta} = \true\]
if and only if (by the definition of the translation and of $\isempty$)
\[
\semvl{\isempty(\sigma_{\neg \theta}(G))}=\true
\]
with $\theta \df \trf{\brt\, \omega \,\ell(G)}$
\item 
It holds that 
\[\semtwovl{\brt \,\omega\, \all(E)}_{D,\eta} = \false\] 
if and only if 
(by definition of the semantics of $\all$)
\[ \bigwedge_{\brt'\in \semtwovl{E}_{D,\eta}} \semtwovl{\brt\, \omega \,\brt'}_{D,\eta} = \false\]
if and only if 
(by induction hypothesis (a))
\[ \bigwedge_{\brt'\in \semvl{G}_{D,\eta}} \semtwovl{\brt\, \omega \,\brt'}_{D,\eta} = \false\]
if and only if (by definition of $\wedge$)
\[ \text{there is }{\brt'\in \semvl{G}_{D,\eta}}: \semtwovl{\brt\, \omega \,\brt'}_{D,\eta} = \false\]
if and only if (by induction hypothesis (b))
\[ \text{there is }{\brt'\in \semvl{G}_{D,\eta}}:  \semvl{\trf{\brt\, \omega \,\brt'}}_{D,\eta} = \true\]
if and only if (by the definition of the translation and of $\isempty$)
\[
\semvl{\neg\isempty(\sigma_{ \theta}(G))}=\true
\]
with $\theta \df \trf{\brt\, \omega \,\ell(G)}$
\OMIT{ \item
 It holds that
 \[\semtwovl{\brt \, \omega \, \all(E)}_{D,\eta} = \false\] 
 	if and only if 	(by definition of the semantics of $\all$) 
 	\[ \bigwedge_{\brt'\in \semtwovl{E}_{D,\eta}} \semtwovl{\brt\, \omega \,\brt'}_{D,\eta} = \false\]
 	if and only if 	(by induction hypothesis (a)) 
 		\[ \bigwedge_{\brt'\in \semvl{G}_{D,\eta}} \semtwovl{\brt\, \omega \,\brt'}_{D,\eta} = \false\]
if and only if 
(by definition of translation of composite conditions and item (c)~(\ref{a}))
\[ \bigvee_{\brt'\in \semvl{G}_{D,\eta}} \semvl{\vee_{i=1}^n(\isnul(t_i)\vee \isnul(t'_i)\vee \neg t_i\omega t'_i)}_{D,\eta} = \true\]
 if and only if (by definition)
\[\semvl{\neg \isempty(\sigma_{\vee_{i=1}^n(\isnul(t_i)\vee \isnul(t'_i)\vee \neg t_i\omega t'_i) }(G))}_{D,\eta} = \true\]} 	
\end{enumerate}
\end{itemize}

\paragraph*{Induction Step:}
\begin{itemize}
	\item[$(a)$] 
	The claim follows from the induction hypothesis $(a)$ directly.
	The equivalence of $\sigma_{\theta}(E)$ and $\sigma_{\tcond}(G)$ is derived from both induction hypotheses $(a)$ and $(b)$.
	\item[$(b)$]
	\begin{enumerate}
		\item It holds that \[\semtwovl{\theta_1 \vee \theta_2}_{D,\eta} = \true \] if and only if (by the definition of $\vee$ in 2VL) 
		\[
		\semtwovl{\theta_1 }_{D,\eta} = \true
		\text{ or }
		\semtwovl{\theta_2}_{D,\eta} = \true \]
		if and only if (by applying induction hypothesis $(b)$)
		\[
		\semthreevl{\trt{\theta_1} }_{D,\eta} = \true
		\text{ or }
		\semthreevl{\trt{\theta_2}}_{D,\eta} = \true \]
		if and only if (by the definition of $\vee$ in Kleene's logic) \[
		\semthreevl{\trt{\theta_1} \vee \trt{\theta_2}}_{D,\eta} = \true
		\]
		if and only if (by the definition of $\trt{\,}$)
		 \[
		\semthreevl{\trt{\theta_1 \vee \theta_2}}_{D,\eta} = \true
		\]
		\item 
		It holds that \[\semtwovl{\theta_1 \wedge \theta_2}_{D,\eta} = \true \] if and only if (by the definition of $\wedge$ in 2VL) 
\[
\semtwovl{\theta_1 }_{D,\eta} = \true
\text{ and }
\semtwovl{\theta_2}_{D,\eta} = \true \]
if and only if (by applying induction hypothesis $(b)$)
\[
\semthreevl{\trt{\theta_1} }_{D,\eta} = \true
\text{ and }
\semthreevl{\trt{\theta_2}}_{D,\eta} = \true \]
if and only if (by the definition of $\wedge$ in Kleene's logic) \[
\semthreevl{\trt{\theta_1} \wedge \trt{\theta_2}}_{D,\eta} = \true
\]
if and only if (by the definition of $\trt{\,}$)
 \[
\semthreevl{\trt{\theta_1 \wedge \theta_2}}_{D,\eta} = \true
\]
		\item
		It holds that
		\[\semtwovl{\neg \theta}_{D,\eta} = \true \] 
		if and only if  
		\[
		\semtwovl{\theta }_{D,\eta} = \false \]
		if and only if (induction hypothesis)
		\[
		\semthreevl{\fcond }_{D,\eta} = \true\]
	\end{enumerate}
	\item[$(c)$]
	\begin{enumerate}
		\item[$1.$] $\semtwovl{\theta_1 \vee \theta_2}_{D,\eta} = \false $ if and only if  
		\[
		\semtwovl{\theta_1 }_{D,\eta} = \false
		\wedge
		\semtwovl{\theta_2}_{D,\eta} = \false \]
		if and only if (induction hypothesis)
		\[
		\semthreevl{
			\fcond_1 }_{D,\eta} = \true
		\wedge
		\semthreevl{\fcond_2}_{D,\eta} = \true \]
		if and only if \[
		\semthreevl{\fcond_1 \wedge \fcond_2}_{D,\eta} = \true
		\]
		\item[$2.$]
		$\semtwovl{\theta_1 \wedge \theta_2}_{D,\eta} = \false $ if and only if  
		\[
		\semtwovl{\theta_1 }_{D,\eta} = \false
		\vee
		\semtwovl{\theta_2}_{D,\eta} = \false \]
		if and only if  (induction hypothesis)
		\[
		\semthreevl{\trf{\theta_1} }_{D,\eta} = \true
		\vee
		\semthreevl{\trf{\theta_2}}_{D,\eta} = \true \]
		if and only if \[
		\semthreevl{\trf{\theta_1} \vee \trf{\theta_2}}_{D,\eta} = \true
		\]
		\item[$3.$]
		$\semtwovl{\neg \theta}_{D,\eta} = \false $ if and only if  
		\[
		\semtwovl{\theta }_{D,\eta} = \true \]
		if and only if (induction hypothesis)
		\[
		\semthreevl{\tcond }_{D,\eta} = \true\]
	\end{enumerate}
\end{itemize}

\subsubsection{$\semsql{\,}$ to $\semtwovl{\,}$}
Let $E$ be a  \sqlra expression and let $F \df \trfromsql{E}$ be the expression obtained from $E$ by inductively replacing each condition $\theta$ with $\trt{\theta}$ as described in Figure~\ref{fig:threetotwo}.
To show that $\semsql{E}_D = \semtwovl{F}_D$ for every database $D$, we prove inductively that
for every database $D$ and every environment $\eta$, the following claims hold:
\newsavebox{\threetotwo}
\sbox{\threetotwo}{%
	\parbox{\textwidth-2\fboxsep-2\fboxrule}{%
		\textbf{Basic conditions:}\\
		\begin{tabular}{ll}
	  	$\trt{\theta}\df$&$ \theta$ for $\theta \df \true\, |\, \false\,|\, \isnul(t) \,|\, t \, \omega \, t'$\\
	     $\trt{\isempty(E)}\df$ &$ \isempty(\trfromsql{E})$\\
	     $\trt{\brt \, \omega \, \any(E)}\df$ &$\brt \, \omega \, \any(\trfromsql{E}) $\\
	     $\trt{\brt \, \omega \, \all(E)}\df$ &$\brt \, \omega \, \all(\trfromsql{E}) $\\
		\end{tabular}
		\begin{tabular}{rl}
	  	$\trf{\theta}\df$&$\neg \theta$ for $\theta \df \true\, |\, \false\,|\, \isnul(t)$\\
	  	$\trf{t \,\omega\, t'} \df$&$\neg\isnul(t) \wedge \neg\isnul(t')\wedge \neg t\,\omega\,t'$\\
	  	\OMIT{$\trf{\brt \eq \brt'} \df$&$ \vee_{i=1}^n 
	(\neg \isnul(t_i) \wedge \neg \isnul(t'_i) \wedge \neg t_i \eq t'_i)$
	  	\\
	 $\trf{\brt \neq \brt'} \df$&$\brt \eq \brt'$ 	
	  	\\
	  	$\trf{\brt \, \omega \, \brt'}\df$&$\vee_{i=0}^{n-1} \big(
(t_1,\cdots, t_i) \eq(t'_1,\cdots, t'_i) \wedge$\\&$ (\neg \isnul(t_{i+1}) \wedge \neg \isnul(t'_{i+1})\wedge \neg(t_{i+1} \,\omega\,t'_{i+1} )
\big)$
\\
where & $\omega \in \{\le,<,>,\ge \}$
\\}
\OMIT{	    $\trf{\brt \in E}\df$ &$\isempty(\sigma_{\neg \theta}(\trfromsql{E}) ) $
	    \\
	    where &
	    $\theta \df \bigvee_{i=1}^n
	    ( \neg \isnul(t_i) \wedge \neg \isnul(N_i) \wedge \neg t_i \eq N_i
	    )$
	    \\
	    and & $\brt \df (t_1,\cdots, t_n), \ell(\trfromsql{E}) \df (N_1,\cdots, N_n) $
	    \\}
	     $\trf{\isempty(E)}\df$ &$ \neg \isempty(\trfromsql{E})$\\
	     $\trf{\brt \, \omega \, \any(E)}\df$ &$\isempty(\sigma_{\neg \theta}(\trfromsql{E}) )$\\
	     $\trf{\brt \, \omega \, \all(E)}\df$ &$\neg \isempty(\sigma_{ \theta}(\trfromsql{E}) ) $\\
	    where&
	 $\theta \df  \trf{ \brt\,\omega\, \ell(E)} 
$\\
\OMIT{	     
	     where &$\theta \df
	     \bigwedge_{i=1}^n \neg \isnul(t_i) \wedge \bigwedge_{i=1}^n \neg \isnul(N_i) \wedge 
	  	\bigvee_{i=1}^{n} (
	  	(t_1,\cdots,t_{i-1})\eq(N_1,\cdots, N_{i-1}) \wedge \neg (t_i \, \omega\, N_i)
	  	)
	     $\\}
	     and & $\brt \df (t_1,\cdots, t_n),\, \ell({E}) \df (N_1,\cdots,N_n)$
		\end{tabular}
		\\\textbf{Composite conditions} exactly as in
	Fig.~\ref{fig:trtf}.}%
	}%

	\newcommand{\threetotwofig}{%
		\begin{figure*}[h]
			\centering
			\fbox{\usebox{\threetotwo}}
			\caption{SQL semantics to $\utof$ semantics: $\trt{\cdot}$ and $\trf{\cdot}$ translations }
			\label{fig:threetotwo}
		\end{figure*}
	}%
	\threetotwofig

\begin{itemize}
	\item[$(a)$] $\semthreevl{E}_{D,\eta} = \semtwovl{\trfromsql{E}}_{D,\eta}$
	\item[$(b)$] 
	$\semthreevl{\theta}_{D,\eta} = \true$ if and only if  $\semtwovl{\tcond}_{D,\eta} = \true$.
	\item[$(c)$]
	$\semthreevl{\theta}_{D,\eta} = \false$ if and only if  $\semtwovl{\fcond}_{D,\eta} = \true$.
\end{itemize}

\paragraph*{Induction base:}
\begin{itemize}
	\item[$(a)$] By definition, $\semthreevl{R}_{D,\eta} = R^D = \semtwovl{R}_{D,\eta}$
	\item[$(b)$] 
	\begin{enumerate}
		\item $\semthreevl{\true}_{D,\eta} = \true$ and $\semtwovl{\true}_{D,\eta} = \true$ and thus the claim holds.
		\item $\semthreevl{\false}_{D,\eta} = \false$ and $\semtwovl{\false}_{D,\eta} = \false$ and thus the claim holds trivially.
		\item
		$\semthreevl{ \isnul(t)}_{D,\eta} = \true$ if and only if $\semvl{t}_{\eta} = \NULL$. This holds if and only if $\semtwovl{ \isnul(t)}_{D,\eta} = \true$.
		\item
		It holds that 
		\[\semvl{ t\,\omega\, t'}_{D,\eta} = \true\]
		if and only if (by definition)
		\[
		\semvl{ t}_{\eta} \ne \NULL, \,
		\semvl{ t'}_{\eta} \ne \NULL,\,
		\semvl{ t}_{\eta} \,\omega\,  \semvl{ t'}_{\eta}
		\]
		if and only if (by definition)
		\[\semtwovl{ t\,\omega\, t'}_{D,\eta} = \true\]
	\item
	It holds that 
	$\semtwovl{\isempty(E)}_{D,\eta} = \true$ if and only if $\semtwovl{E}_{D,\eta} = \emptyset$. This holds if and only if (induction hypothesis) $\semvl{\trfromsql{E}}_{D,\eta} = \emptyset$, which holds if and only if   $\semvl{\isempty(\trfromsql{E})}_{D,\eta} = \true$.
	\item
	It holds that  
\[\semthreevl{\brt \, \omega\,\any (E)}_{D,\eta} = \true\] 
if and only if (by definition)
\[ \bigvee_{\brt'\in \semvl{E}_{D,\eta}} \semvl{\brt \,\omega\, \brt'}_{\eta} = \true
\] 
if and only if 
(by induction hypothesis (a))
\[ \bigvee_{\brt'\in \semtwovl{\trfromsql{E}}_{D,\eta}} \semvl{\brt \,\omega\, \brt'}_{\eta} = \true
\] 
if and only if 
(by induction hypothesis (b))
\[ \bigvee_{\brt'\in \semtwovl{\trfromsql{E}}_{D,\eta}} \semtwovl{\trt{\brt \,\omega\, \brt'}}_{\eta} = \true
\] 
if and only if (by definition of $\trt{\brt \,\omega\, \brt'}$)
\[ \bigvee_{\brt'\in \semtwovl{\trfromsql{E}}_{D,\eta}} \semtwovl{\brt \,\omega\, \brt'}_{\eta} = \true
\] 
if and only if (by definition of $\any$)
\[\semtwovl{\brt \,\omega\, \any( \trfromsql{E})}_{D,\eta} = \true\]
	\item
	It holds that  
\[\semthreevl{\brt \, \omega\,\all (E)}_{D,\eta} = \true\] 
if and only if (by definition)
\[ \bigwedge_{\brt'\in \semvl{E}_{D,\eta}} \semvl{\brt \,\omega\, \brt'}_{\eta} = \true
\] 
if and only if 
(by induction hypothesis (a))
\[ \bigwedge_{\brt'\in \semtwovl{\trfromsql{E}}_{D,\eta}} \semvl{\brt \,\omega\, \brt'}_{\eta} = \true
\] 
if and only if 
(by induction hypothesis (b))
\[ \bigwedge_{\brt'\in \semtwovl{\trfromsql{E}}_{D,\eta}} \semvl{\trt{\brt \,\omega\, \brt'}}_{\eta} = \true
\] 
if and only if (by definition of $\trt{\brt \,\omega\, \brt'}$)
\[ \bigwedge_{\brt'\in \semtwovl{\trfromsql{E}}_{D,\eta}} \semtwovl{{\brt \,\omega\, \brt'}}_{\eta} = \true
\] 
if and only if (by definition of $\all$)
\[\semtwovl{\brt \,\omega\, \all( \trfromsql{E})}_{D,\eta} = \true\]
	\end{enumerate}
	\item[$(c)$]
	\begin{enumerate}
		\item $\semthreevl{\true}_{D,\eta} = \true$ and $\semtwovl{\false}_{D,\eta} = \false$ and thus the claim holds.
		\item $\semthreevl{\false}_{D,\eta} = \false$ and $\semtwovl{\true}_{D,\eta} = \true$ and thus the claim holds.
		\item
		It holds that 
		$\semthreevl{ \isnul(t)}_{D,\eta} = \false$ if and only if $\semvl{t}_{\eta}\ne \NULL$. This holds by definition, if and only if $\semtwovl{\neg \isnul(t)}_{D,\eta} = \true$.
		\item
		It holds that 
		\[\semvl{ t\,\omega\, t'}_{D,\eta} = \false\]
		if and only if (by definition)
		\[
		\semvl{ t}_{\eta} \ne \NULL, \,
		\semvl{ t'}_{\eta} \ne \NULL,\text{ and }
		\semvl{ t}_{\eta}\, \not \omega\,  \semvl{ t'}_{\eta}
		\]
		if and only if (by definition)
		\[\semtwovl{ \neg\isnul(t) \wedge \neg\isnul(t') \wedge \neg t\,\omega\, t'}_{D,\eta} = \true\]
	\item
$\semvl{\isempty(E)}_{D,\eta} = \false$ if and only if $\semvl{E}_{D,\eta} \ne \emptyset$. By induction hypothesis $(a)$, this holds if and only if  $\semtwovl{\trfromsql{E}}_{D,\eta} \ne \emptyset$. This, in turn, holds if and only if $\semtwovl{\neg \isempty(\trfromsql{E})}_{D,\eta} = \true$.
\item
It holds that  
\[\semvl{\brt \, \omega \, \any( E)}_{D,\eta} = \false\] 
if and only if (by definition of $\any$)
\[ \bigvee_{\brt'\in \semvl{E}_{D,\eta}} \semvl{\brt \, \omega \, \brt'}_{\eta} = \false
\] 
if and only if 
(by induction hypothesis (a))
\[ \bigvee_{\brt'\in \semtwovl{\trfromsql{E}}_{D,\eta}} \semvl{\brt \, \omega \, \brt'}_{\eta} = \false
\]
if and only if (by definition of $\vee$)
\[
\text{for all }{\brt'\in \semtwovl{\trfromsql{E}}_{D,\eta}}: \semvl{\brt \, \omega \, \brt'}_{\eta} = \false
\]
if and only if
(by induction hypothesis (b))
\[
\text{for all }{\brt'\in \semtwovl{\trfromsql{E}}_{D,\eta}}: \semtwovl{\trf{\brt \, \omega \, \brt'}}_{D,\eta} = \true
\]
if and only if (by definition of $\isempty$)
\[
\semtwovl{
\isempty(\sigma_{\neg \theta}(\trfromsql{E})}_{D,\eta}
\]
with $\theta\df \trf{\brt \, \omega \, \ell(E)}$.

\item
It holds that  
\[\semvl{\brt \, \omega \, \all( E)}_{D,\eta} = \false\] 
if and only if (by definition of $\all$)
\[ \bigwedge_{\brt'\in \semvl{E}_{D,\eta}} \semvl{\brt \, \omega \, \brt'}_{\eta} = \false
\] 
if and only if 
(by induction hypothesis (a))
\[ \bigwedge_{\brt'\in \semtwovl{\trfromsql{E}}_{D,\eta}} \semvl{\brt \, \omega \, \brt'}_{\eta} = \false
\]
if and only if (by definition of $\wedge$)
\[
\text{there is }{\brt'\in \semtwovl{\trfromsql{E}}_{D,\eta}}: \semvl{\brt \, \omega \, \brt'}_{\eta} = \false
\]
if and only if
(by induction hypothesis (b))
\[
\text{there is }{\brt'\in \semtwovl{\trfromsql{E}}_{D,\eta}}: \semtwovl{\trf{\brt \, \omega \, \brt'}}_{D,\eta} = \true
\]
if and only if (by definition of $\isempty$)
\[\semtwovl{
\neg \isempty(\sigma_{ \theta}(\trfromsql{E})}_{D,\eta} = \true
\]
with $\theta\df \trf{\brt \, \omega \, \ell(E)}$.
	\end{enumerate}
\end{itemize}

\paragraph{Induction Step:}
\begin{enumerate}
	\item[$(a)$] The claim is straightforward from the definitions. The only non-trivial case is $\sigma_{\theta}(E)$ for which we obtain $\sigma_{\theta^{\true}}(F)$ and the claim holds from both induction hypotheses.
	\item[$(b)$]
	\begin{enumerate}
		\item[$1.$] It holds that
	\[\semtwovl{\trt{\theta_1} \vee \trt{\theta_2}}_{D,\eta}=\true\]
		if and only if (truth table of $\vee$ in $\twovl$)
		\[\semtwovl{\trt{\theta_1}}_{D,\eta}=\true \vee \semtwovl{\trt{\theta_2}}_{D,\eta}=\true\] if and only if 
		(induction hypothesis)
		\[\semthreevl{\theta_1}_{D,\eta}=\true \vee \semthreevl{\theta_2}_{D,\eta}=\true\]
		if and only if (truth table of $\vee$ in $\threevl$)
		\[\semthreevl{\theta_1 \vee \theta_2}_{D,\eta}=\true\]
		\item[$2.$] 
		It holds that
		\[\semthreevl{\theta_1 \wedge \theta_2}_{D,\eta} = \true\]
		if and only if 
		\[\semthreevl{\theta_1}_{D,\eta} = \true \wedge \semtwovl{\theta_2}_{D,\eta} = \true
		\] if and only if (induction hypothesis $(b)$)
		\[\semtwovl{\trt{\theta_1}}_{D,\eta} = \true \wedge \semtwovl{\trt{\theta_2}}_{D,\eta} = \true
		\] if and only if 
		\[\semtwovl{\trt{\theta_1} \wedge \trt{\theta_2}}_{D,\eta} = \true
		\] 
		\item[$3.$]
		It holds that 
		\[\semthreevl{\neg \theta}_{D,\eta}  = {\true}\] if and only if
		\[\neg \semthreevl{ \theta}_{D,\eta}  = {\true}\] if and only if
		\[ \semthreevl{ \theta}_{D,\eta}  = {\false}\]
		if and only if (induction hypothesis $(b)$)
		\[\semtwovl{ \trf{\theta}}_{D,\eta}  = {\true}\]
	\end{enumerate}
	\item[$(c)$]
	\begin{enumerate}
		\item[$1.$] 
		It holds that
		\[\semthreevl{\trf{\theta_1} \wedge \trf{\theta_2}}=\true\] 
		if and only if
		(truth table of $\wedge$ in $\threevl$)
		\[\semthreevl{\trf{\theta_1}} = \true \wedge \semthreevl{\trf{\theta_2}} = \true\]
		if and only if (induction hypothesis)
		\[\semtwovl{\theta_1} = \false \wedge \semtwovl{\theta_2} = \false\]
		if and only if 	(truth table of $\vee$ in $\twovl$)
		\[\semtwovl{\theta_1  \vee \theta_2} = \false\]
		\item[$2.$] 
		It holds that 
		\[\semthreevl{\theta_1 \wedge \theta_2}_{D,\eta} = \false\]
		if and only if 
		\[\semthreevl{\theta_1}_{D,\eta} = \false \vee \semvl{\theta_2}_{D,\eta} = \false
		\] if and only if (induction hypothesis $(b)$)
		\[\semtwovl{\trf{\theta_1}}_{D,\eta} = \true \vee \semtwovl{\trf{\theta_2}}_{D,\eta} = \true
		\] if and only if 
		\[\semtwovl{\trf{\theta_1} \vee \trf{\theta_2}}_{D,\eta} = \true
		\] 
		\item[$3.$] 
		It holds that \[\semthreevl{\neg \theta}_{D,\eta}  = {\false}\] 
		if and only if
		\[\neg \semthreevl{ \theta}_{D,\eta}  = {\false}\] if and only if
		\[ \semthreevl{ \theta}_{D,\eta}  = {\true}\]
		if and only if (induction hypothesis $(b)$)
		\[\semtwovl{ \trt{\theta}}_{D,\eta}  = {\true}\]
	\end{enumerate}
\end{enumerate}

\OMIT
{\begin{example}
	\label{sqltouf-ex}
	We now look at translations of queries $Q_1$--$Q_5$ from
	Example~\ref{queries-ex}. This time, we assume they are written under
	the usual SQL semantics, and show how they would look under the
	$\utof$ semantics. As before, in queries $Q_2$, $Q_3$, and $Q_4$
	nothing changes: they are the same. 
	The translation of query $Q_1$ is as follows:
	$$
	\sigma_{\neg \isnul(R.A) \wedge \neg R.A\in S }(R)\,,
	$$
	or in two-valued SQL,
	\begin{sql}
		SELECT R.A FROM R WHERE R.A IS NOT NULL
		AND R.A NOT IN (SELECT S.A FROM S)
	\end{sql}
	In the translation $Q'_5$ of query $Q_5$,
	the  condition $(c\_a > 0) \wedge 
	\neg (c\_c \in \pi_{o\_c}(O))$ in the
	subquery is translated as  
	$$(c\_a >
	0) \wedge \Big(\neg \isnul(c\_c) \wedge \neg \big(c\_c \in \pi_{o\_c}(O)\big)\Big)\,.$$
	Since \lstinline{c_custkey} is the key, $\neg \isnul(c\_c)$ is always
	true, and thus could be removed. In other words, $Q_5$ does not
	change: it could have been written in two-valued SQL without any
	problem.
\end{example}}

\end{proof}

\begin{lemma}\label{lem:syneq}
The $\semtwovlsyneq{\,}$ 
semantics of
	\sqlrarec\ expressions, and of \sqlra\ expressions, capture
	their  SQL semantics
$\semsql{\,}$	efficiently.    
\end{lemma}

We show that this lemma can be further generalized and be proved as a corollary. We now set the ground for presenting this generalized theorem. 

Are $\semtwovl{}$ and $\semtwovlsyneq{}$ the only possible two-valued semantics out there? Of course
not. Consider for example
the predicate $\leq$. In both $\semtwovlsyneq{\,}$ and $\semtwovl{\,}$,
$\NULL\leq\NULL$ is $\vf$, but 
under the syntactic equality
interpretation it is not unreasonable to say that $\NULL\leq\NULL$ is
true, as $\leq$ subsumes equality. This gives a general idea of how
different semantics can be obtained: when some arguments of a
predicate are $\NULL$, we can decide what the truth value based on
other values. Just for the sake of example, we could say that 
$n \leq \NULL$ is 
\vf\ for $n < 0$ and $\vt$ for $n \geq 0$.

To define such 
alternative semantics in a general way, we simply state, for each comparison predicate in $\{\eq,\neq,\le,<,>,\ge \}$, what 
happens when one or two of its arguments are \NULL.
To this end, we associate each comparison predicate with a binary predicate with the same notation.
Our construction is in fact even more general, as it covers also predicates of higher arities.

\OMIT{
We introduce the notion of {\em grounding} 
of 
predicates.
It is a function
$\ground$ that takes an $n$-ary predicate $P$ of type $\tau$ and a
non-empty sequence $\bI=\langle i_1,\cdots ,i_k\rangle$ of indices
$1\le i_1<\cdots < i_k \le n$, and produces a relation
$\ground(P,\bI)$ that contains $\tau$-records $(t_1, \cdots, t_n)$
where $t_i = \NULL$ for every $i\in \bI$, and $t_i \ne \NULL$ for every
$i\not\in \bI$. In the above example, if $P$ is $\leq$, then
$\ground(P,\langle 1\rangle) = \{(\NULL,n)\mid n \geq 0\}$ while
$\ground(P,\langle 2\rangle) = \{(\NULL,n)\mid n < 0\}$ and
$\ground(P,\langle 1,2\rangle) = \{(\NULL,\NULL)\}$.
}

We introduce the notion of {\em grounding}:
It is a function
$\ground$ that takes an $n$-ary predicate $P$ of type $\tau$ and a
non-empty sequence $\bI=\langle i_1,\cdots ,i_k\rangle$ of indices
$1\le i_1<\cdots < i_k \le n$, and produces a relation
$\ground(P,\bI)$ that contains $\tau$-records $(t_1, \cdots, t_n)$
where $t_i = \NULL$ for every $i\in \bI$, and $t_i \ne \NULL$ for every
$i\not\in \bI$. 
In the above example, if $P$ is $\leq_2$, then
$\ground(P,\langle 1\rangle) = \{(\NULL,n)\mid n \geq 0\}$ while
$\ground(P,\langle 2\rangle) = \{(\NULL,n)\mid n < 0\}$ and
$\ground(P,\langle 1,2\rangle) = \{(\NULL,\NULL)\}$.

The semantics $\semtwovlgrnd{\,}$ based on such grounding is given by
redefining the truth value of $\semtwovlgrnd{P(t_1,\ldots,t_n)}$ as
that of $\bar t \in
\ground(P,\bI)$, where $\bI$ is the list on indices $i$ such that
$\semtwovlgrnd{t_i}=\NULL$.

Such semantics
generalize $\semtwovl{\,}$ 
and $\semtwovlsyneq{\,}$. \OMIT{In the former,
by setting $\ground(\eq_{m},\bI) = \emptyset$ for each nonempty $\bI$;
in the latter,
\OMIT{it is the same except that $\ground(=,\langle 1,2\rangle)$ would
contain the tuple $(\NULL,\NULL)$. }
it is the same except that $\ground(=_{2m},\langle i,m+i\rangle)$ for every $1\le i \le m$ would
contain also the tuple $(c_1,\cdots, c_{i-1},\NULL,c_{i+1},\cdots, c_m,c_{1},\cdots, c_{i-1},\NULL,c_{i+1},\cdots, c_m)$.
}
In the former,
by setting $\ground(\eq,\bI) = \emptyset$ for each nonempty $\bI$;
in the latter,
it is the same except that $\ground(=,\langle 1,2\rangle)$ would
contain also the tuple $(\NULL,\NULL)$.

We say that a grounding is {\em
	expressible} if for each predicate $P$
	and each $\bI$ there is a
condition $\theta_{P,\bI}(\bar t)$ such that $\bar t \in \ground(P,\bI)$ if and only if
$\pi_{\bar\bI}(\bar t)$ satisfies $\theta_{P,\bI}(\bar t)$ where $\pi_{\bar\bI}$ denotes the projection on the complement of $\bI$, i.e., the indices that do not occur in $\bI$. Note that the
projection on the complement $\bar\bI$ of $\bI$ would only contain non-null
elements. 
All the semantics used so far in the paper -- and sketched
earlier in this section -- satisfy this condition. 

Notice that for $\semtwovl{\,}$ it holds that $\theta_{\eq, \{1\}}(t,t') = 
\theta_{\eq, \{2\}}(t,t') = 
\theta_{\eq, \{1,2\}}(t,t') = \false$,
and for $\semtwovlsyneq{\,}$ it holds that $\theta_{\omega, \{1\}}(t,t') = 
\theta_{\omega, \{2\}}(t,t') = \false$ for every $\omega$, and  
$\theta_{\omega, \{1,2\}}(t,t') = \true$ whenever $\omega$ subsumes equality, that is $\omega \in\{\le,\ge,\eq \}$.

We discuss now the translation from the special case $\semtwovlsyneq{\,}$ to SQL. Though this translation is captured by the one depicted in  Figure~\ref{fig:twogrtothree}, we write it explicitly in its reduced form.
\begin{itemize}
    \item For $\omega \in \{<,>,\neq\}$
    \begin{itemize}
        \item $\trt{t\,\omega \, t'} = \neg\isnul(t) \wedge
        \neg \isnul(t')\wedge t\,\omega \, t'
    $
    \item
    $\trf{t\,\omega \, t'} = \isnul(t) \vee \isnul(t') \vee \left(\neg\isnul(t) \wedge\neg \isnul(t')\wedge \neg t\,\omega \, t'\right)
    $ 
    \end{itemize} 
    \item
    For $\omega \in \{\le,\ge,\eq\}$
    \begin{itemize}
        \item $\trt{t\,\omega \, t'} = \left(\neg\isnul(t) \wedge
        \neg \isnul(t')\wedge t\,\omega \, t'\right) \vee
        \left( \isnul(t) \wedge
         \isnul(t') \right)
    $
    \item
    $\trf{t\,\omega \, t'} = (\isnul(t)\wedge  \neg \isnul(t'))\vee 
   ( \isnul(t')  \wedge \neg \isnul(t))
    \vee \left(\neg\isnul(t) \wedge\neg \isnul(t')\wedge \neg t\,\omega \, t'\right)
    $ 
    \end{itemize} 
\end{itemize}

\begin{theorem} \label{thm:grgen}
	For every expressible grounding $\ground$, the $\semtwovlgrnd{\,}$ semantics of
	\sqlrarec\ or \sqlra\ expressions captures their SQL semantics.
\end{theorem}

\begin{proof}
\paragraph*{From $\semtwovlgrnd{\,}$ to $\semsql{\,}$}

\newsavebox{\twogrtothree}
\sbox{\twogrtothree}{%
	\parbox{\textwidth}{%
\begin{center}
\textbf{Basic conditions}\\ $\;$ \\   
\begin{tabularx}{\textwidth}{rlrl}
$\trt{P(\brt)}$&$\df {\displaystyle \bigvee_{I\subseteq\{1,\cdots,k \}} \Big(c_I \wedge \theta_{P,\bI}(\brt)\Big)}$
&
$\trf{P(\brt)}$&$ \df \bigvee_{I\subseteq\{1,\cdots,k \}} \Big(c_I \wedge \neg\theta_{P,\bI}(\brt)\Big)$ 
\\
\multicolumn{4}{l}{
where 
$c_I \df \bigwedge_{j\in I} \isnul(t_j) \wedge  \bigwedge_{j\in \bar{I}} \neg \isnul(t_j)$
}
\\
$\trt{t \,\omega\,  t'}$ &$\df {\displaystyle  \bigvee_{I\subseteq\{1,2\} } \big( c_I \wedge \theta_{\omega,I}(t,t') \big)
}$
&$ \trf{t  \,\omega\,  t'}$
&$ \df{\displaystyle  \bigvee_{I\subseteq\{1,2\} } \big( c_I \wedge \neg \theta_{\omega,I}(t,t') \big)
}
    $ \\ 
    \multicolumn{4}{l}{
where 
$c_I \df \bigwedge_{j\in I} \isnul(N_j) \wedge \bigwedge_{j\in \bar{I}} \neg \isnul(N_j)  $ with $N_1 \df t, N_2\df t'$
}\\
\OMIT{
$\tcon{\brt \eq  \brt'}$ &$\df 
{\displaystyle  \bigwedge_{i=1}^n
\tcon{t_i \eq t'_i}
}
$
&
$ \fcon{\brt  \eq  \brt'}$
&$ \df {\displaystyle \bigvee_{i=1}^n} 
\fcon{t_i \eq t'_i}
$ \\ 
$\tcon{\brt \neq  \brt'}$ &$\df {\displaystyle  \bigvee_{i=1}^n
\fcon{t_i \eq t'_i}
}$
&
$ \fcon{\brt  \neq  \brt'}$
&$ \df {\displaystyle 
\bigwedge_{i=1}^n
\tcon{t_i \eq t'_i}
}
$}\\
\OMIT{
$\tcon{\brt \,\omega \, \brt'}$ &$\df
{\displaystyle
\begin{multlined}[t]
\bigvee_{i=1}^{n-1}\big(
\tcon{(t_1,\cdots , t_i) \eq (t'_1,\cdots , t'_i)}
\wedge
\tcon{ t_{i+1} \,\omega \, t'_{i+1}}\big)
\end{multlined}
}$
&
$ \fcon{\brt   \,\omega \, \brt'}$
&$ \df {\displaystyle
\begin{multlined}[t]
\bigwedge_{i=1}^{n-1}\big(
\tcon{(t_1,\cdots , t_i) \eq (t'_1,\cdots , t'_i)}
\vee \neg
\fcon{ t_{i+1} \,\omega\, t'_{i+1}}\big)
\end{multlined}
}$
\\
\multicolumn{4}{l}{$\text{where
}  \omega\in \{ <,\le,\ge,>\},\,\brt\df(t_1,\cdots,t_n),\,\,\brt' \df(t'_1,\cdots,t'_n) $}\\ $\;$ \\
$\tcon{\brt \in E}$&$ \df 
\neg \isempty \big(\sigma_{\theta}(G)\big)
$&$
\fcon{\brt\in E}$&$
\df
\isempty \big(\sigma_{\theta}(G)\big)$
\\ \multicolumn{4}{l}
{$
\text{
where } 
\theta\df  
\tcon{\brt \eq \ell(E)}$}
\\}
$\tcon{t\, \omega\, \any(E)}$ & $ \df 
\neg \isempty(\sigma_{ \theta } (G)) $
&
$\fcon{t\, \omega\, \any(E)}$&
$ \df 
\isempty(\sigma_{   \theta } (G)) 
$		
\\
$			\tcon{t\, \omega\, \all(E)}$&$ \df 
\isempty(\sigma_{ \neg \theta } (G))
$&$
\fcon{t\, \omega\, \all(E)}$&$
\df
\neg \isempty(\sigma_{\neg \theta } (G))$
\\ \multicolumn{4}{l}{$\text{
where }
\theta\df  
{\displaystyle  \tcon{\brt\,\omega\,\ell(E)}
}
$}
\end{tabularx}.
	\end{center}
}%
}%

\newcommand{\twogrtothreefig}{%
	\begin{figure*}[t]
		\centering
		\fbox{\usebox{\twogrtothree}}
		\caption{\mbox{$\semtwovlgrnd{}{\,}$}  to SQL semantics
			(new cases compared to Fig.~\ref{fig:trtf})}
		\label{fig:twogrtothree}
	\end{figure*}
}%
\twogrtothreefig

\begin{itemize}
    \item[(b)]
    \begin{enumerate}
        \item
        It holds that \[\semtwovlgrnd{P(\brt)}_{D,\eta} = \true\]
       if and only if (by definition)
       \[
       \semsql{\bigvee_{I\subseteq\{1,\cdots,k \}} \Big(c_I \wedge \theta_{P,\bI}(\brt)\Big)}_{D,\eta} =\true
       \]
\item
It holds that \[\semtwovlgrnd{t\,\omega\,t'}_{D,\eta} = \true\]
       if and only if (by definition)
       \[
       \semsql{\bigvee_{I\subseteq\{1,2 \}} \Big(c_I \wedge \theta_{\omega,\bI}(\brt)\Big)}_{D,\eta} =\true
       \]
	\item
	It holds that  
\[\semtwovlgrnd{\brt \, \omega\,\any (E)}_{D,\eta} = \true\] 
if and only if (by definition)
\[ \bigvee_{\brt'\in \semtwovlgrnd{E}_{D,\eta}} \semtwovlgrnd{\brt \,\omega\, \brt'}_{D,\eta} = \true
\] 
if and only if 
(by induction hypothesis (a))
\[ \bigvee_{\brt'\in \semvl{G}_{D,\eta}} \semtwovlgrnd{\brt \,\omega\, \brt'}_{D,\eta} = \true
\] 
if and only if 
(by induction hypothesis (b))
\[ \bigvee_{\brt'\in \semvl{G}_{D,\eta}} \semvl{\trt{\brt \,\omega\, \brt'}}_{D,\eta} = \true
\] 
if and only if (by definition of $\trt{\brt \,\omega\, \brt'}$)
\[ \bigvee_{\brt'\in \semvl{G}_{D,\eta}} \semvl{\brt \,\omega\, \brt'}_{D,\eta} = \true
\] 
if and only if (by definition of $\any$)
\[\semvl{\brt \,\omega\, \any( G)}_{D,\eta} = \true\]
	\item
	It holds that  
\[\semtwovlgrnd{\brt \, \omega\,\all (E)}_{D,\eta} = \true\] 
if and only if (by definition)
\[ \bigwedge_{\brt'\in \semtwovlgrnd{E}_{D,\eta}} \semtwovlgrnd{\brt \,\omega\, \brt'}_{D,\eta} = \true
\] 
if and only if 
(by induction hypothesis (a))
\[ \bigwedge_{\brt'\in \semvl{G}_{D,\eta}} \semvl{\brt \,\omega\, \brt'}_{D,\eta} = \true
\] 
if and only if 
(by induction hypothesis (b))
\[ \bigwedge_{\brt'\in \semvl{G}_{D,\eta}} \semvl{\trt{\brt \,\omega\, \brt'}}_{D,\eta} = \true
\] 
if and only if (by definition of $\trt{\brt \,\omega\, \brt'}$)
\[ \bigwedge_{\brt'\in \semvl{G}_{D,\eta}} \semvl{{\brt \,\omega\, \brt'}}_{\eta} = \true
\] 
if and only if (by definition of $\all$)
\[\semvl{\brt \,\omega\, \all( G)}_{D,\eta} = \true\]
\OMIT{
\item
It holds that 
\[
\semtwovlgrnd{\brt \, \omega \, \any(E)}_{D,\eta} = \true
\] if and only if 
\[\semsql{
\neg \isempty (\sigma_{\tcon{\brt \,\omega\, \ell(E)}}(G))
}_{D,\eta}
= \true
\] we use similar arguments as previous item while replacing $\eq$ with $\omega$.
\item
It holds that 
\[
\semtwovlgrnd{\brt \,\omega\, \all(E)}_{D,\eta} = \true
\]
if and only if (by definition of $\all$)
\[
\bigwedge_{\brt'\in \semtwovlgrnd{E}_{D,\eta}
} \semtwovlgrnd{\brt \,\omega\, \brt'}_{D,\eta}
= \true
\]
if and only if (by induction hypothesis (a) and previous items)
\[
\bigwedge_{\brt'\in \semsql{G}_{D,\eta}
} \semsql{\tcon{\brt \,\omega\, \brt'}}_{D,\eta}
= \true
\]
if and only if (by definition of $\wedge$)
\[\semsql{
\isempty (\sigma_{\neg \tcon{\brt \,\omega\, \ell(E)}}(G))
}_{D,\eta}
= \true
\]}
    \end{enumerate}
    
    \item[(c)]
\begin{enumerate}
        \item
        It holds that \[\semtwovlgrnd{P(\brt)}_{D,\eta} = \false\]
       if and only if (by definition)
       \[
       \semsql{\bigvee_{I\subseteq\{1,\cdots,k \}} \Big(c_I \wedge \neg\theta_{P,\bI}(\brt)\Big)}_{D,\eta} =\true
       \]
\item
It holds that \[\semtwovlgrnd{t\,\omega\,t'}_{D,\eta} = \false\]
       if and only if (by definition)
       \[
       \semsql{\bigvee_{I\subseteq\{1,2 \}} \Big(c_I \wedge \neg\theta_{\omega,\bI}(\brt)\Big)}_{D,\eta} =\true
       \]
\item 
It holds that 
\[\semtwovlgrnd{\brt \,\omega\, \any(E)}_{D,\eta} = \false\] 
if and only if 
(by definition of the semantics of $\any$)
\[ \bigvee_{\brt'\in \semtwovlgrnd{E}_{D,\eta}} \semtwovlgrnd{\brt\, \omega \,\brt'}_{D,\eta} = \false\]
if and only if 
(by induction hypothesis (a))
\[ \bigvee_{\brt'\in \semvl{G}_{D,\eta}} \semtwovlgrnd{\brt\, \omega \,\brt'}_{D,\eta} = \false\]
if and only if (by definition of $vee$)
\[ \text{for all }{\brt'\in \semvl{G}_{D,\eta}}: \semtwovlgrnd{\brt\, \omega \,\brt'}_{D,\eta} = \false\]
if and only if (by induction hypothesis (b))
\[ \text{for all }{\brt'\in \semvl{G}_{D,\eta}}:  \semvl{\trf{\brt\, \omega \,\brt'}}_{D,\eta} = \true\]
if and only if (by the definition of the translation and of $\isempty$)
\[
\semvl{\isempty(\sigma_{\neg \theta}(G))}=\true
\]
with $\theta \df \trf{\brt\, \omega \,\ell(G)}$
\item 
It holds that 
\[\semtwovlgrnd{\brt \,\omega\, \all(E)}_{D,\eta} = \false\] 
if and only if 
(by definition of the semantics of $\all$)
\[ \bigwedge_{\brt'\in \semtwovlgrnd{E}_{D,\eta}} \semtwovlgrnd{\brt\, \omega \,\brt'}_{D,\eta} = \false\]
if and only if 
(by induction hypothesis (a))
\[ \bigwedge_{\brt'\in \semvl{G}_{D,\eta}} \semtwovlgrnd{\brt\, \omega \,\brt'}_{D,\eta} = \false\]
if and only if (by definition of $\wedge$)
\[ \text{there is }{\brt'\in \semvl{G}_{D,\eta}}: \semtwovlgrnd{\brt\, \omega \,\brt'}_{D,\eta} = \false\]
if and only if (by induction hypothesis (b))
\[ \text{there is }{\brt'\in \semvl{G}_{D,\eta}}:  \semvl{\trf{\brt\, \omega \,\brt'}}_{D,\eta} = \true\]
if and only if (by the definition of the translation and of $\isempty$)
\[
\semvl{\neg\isempty(\sigma_{ \theta}(G))}=\true
\]
with $\theta \df \trf{\brt\, \omega \,\ell(G)}$
\OMIT{\item
To show that \[
\semtwovlgrnd{\brt \, \omega \, \any(E)}_{D,\eta} = \false
\] if and only if 
\[\semsql{
 \isempty (\sigma_{\tcon{\brt \,\omega\, \ell(E)}}(G))
}_{D,\eta}
= \true
\] we use similar arguments as previous item while replacing $\eq$ with $\omega$.
\item
It holds that 
\[
\semtwovlgrnd{\brt \,\omega\, \all(E)}_{D,\eta} = \false
\]
if and only if (by definition of $\all$)
\[
\bigwedge_{\brt'\in \semtwovlgrnd{E}_{D,\eta}
} \semtwovlgrnd{\brt \,\omega\, \brt'}_{D,\eta}
= \false
\]
if and only if (by induction hypothesis (a) and previous items)
\[
\bigwedge_{\brt'\in \semsql{G}_{D,\eta}
} \semsql{\tcon{\brt \,\omega\, \brt'}}_{D,\eta}
= \false
\]
if and only if (by definition of $\wedge$)
\[\semsql{
\neg \isempty (\sigma_{\neg \tcon{\brt \,\omega\, \ell(E)}}(G))
}_{D,\eta}
= \true
\]}
    \end{enumerate}    
    
\end{itemize}

\paragraph*{From $\semsql{\,}$ to $\semtwovlgrnd{\,}$}

\newsavebox{\sqltogr}
\sbox{\sqltogr}{%
	\parbox{\textwidth}{%
\begin{center}
\textbf{Basic conditions}\\ $\;$ \\   
\begin{tabularx}{\textwidth}{rlrl}
${P(\brt)}^{\true} $&$\df {\displaystyle 	 \bigwedge_{i=1}^n \neg \isnul(t_i) \wedge 
	P(\brt)}$
&
${P(\brt)}^{\false}$&$ \df {
\displaystyle 
\bigwedge_{i=1}^n \neg \isnul(t_i) \wedge \neg
	P(\brt)
}$ 
\\
\multicolumn{4}{l}{
where 
$\brt \df (t_1,\cdots , t_n)$
}
\\
$\tcon{t \,\omega\,  t'}$ &$\df {\displaystyle 	 \neg \isnul(t) \wedge  \neg \isnul(t') \wedge 
	\theta_{\omega,I}(t,t')}$
&$ \fcon{t  \,\omega\,  t'}$
&$ \df{\displaystyle 	  \neg \isnul(t) \wedge  \neg \isnul(t') \wedge  
\neg	\theta_{\omega,I}(t,t')}
    $
\end{tabularx}.
	\end{center}
}%
}%

\newcommand{\sqltogrfig}{%
	\begin{figure*}[t]
		\centering
		\fbox{\usebox{\sqltogr}}
		\caption{SQL semantics to \mbox{$\semtwovlgrnd{}{\,}$} 
			(new cases compared to Fig.~\ref{fig:threetotwo})}
		\label{fig:sqltogr}
	\end{figure*}
}%
\sqltogrfig
We use all the translations that appear in Figure~\ref{fig:twogrtothree}
.
\OMIT{
\begin{itemize}
	\item 
	${P(\brt)}^{\true} \df \bigwedge_{i=1}^k \neg \isnul(t_i) \wedge 
	P(\brt)$
	\item
	${P(\brt)}^{\false} \df \bigwedge_{i=1}^k \neg \isnul(t_i) \wedge \neg
	P(\brt)$
\end{itemize}
}

The proofs of this direction are shown similarly to the proofs of the translation from $\semvl{}$ to $\semtwovl{}$.
\end{proof}
Notice that Lemmas~\ref{lem:syneq} and~\ref{lem:twovl} are both direct consequence of Theorem~\ref{thm:grgen}.

\section{Appendix for Section~\ref{sec:restore}}

\repeatresult{proposition}{\ref{prop:restequiv}}{ \proprestequiv}
\begin{proof}
We start by showing the first part of the proposition:
\begin{enumerate}
	\item Let $\bara \in_k \semtwovlgen{\sigma_{\theta}(E)}_{D,\eta}$. 
	By definition, this holds if and only if 
	both  $\bara \in_k \semtwovlgen{E}_{D,\eta}$ and $\semtwovlgen{\theta}_{D,\eta;\eta_{\ell(E)}^{\bara}} = \true$. 
	Since $\star$ is two-valued logic and due to the truth tables of the connector $\neg$, the above happens if and only if both  $\bara \in_k \semtwovl{E}_{D,\eta}$ and $\semtwovl{\neg \theta}_{D,\eta;\eta_{\ell(E)}^{\bara}} = \false$.
	By the definition of the difference operator the above happens if and only if 
	 $\bara \in_k \semtwovl{E\setminus \sigma_{\neg \theta}(E)}_{D,\eta}$.
\item
$\semtwovlgen{\brt \in E}_{D,\eta} = \false$ if and only if (by definition of $\in$) the 2VL disjunction  $\bigvee_{\brt'\in \semtwovlgen{E}_{D,\eta}} \semtwovlgen{\brt \eq \brt'} = \false$. This holds if and only if (by definition of the connector $\vee$ in 2VL)
for every $\brt'\in \semtwovlgen{E}_{D,\eta}$ it holds that $\semtwovlgen{\brt \eq \brt'} = \false$. By definition of $\eta_{\ell(E)}^{\brt'}$, this holds if and only if for every $\brt'\in \semtwovlgen{E}_{D,\eta}$ it holds that
$\semtwovlgen{\brt \eq \ell(E)}_{\eta;\eta_{\ell(E)}^{\brt'}} = \false$.  
This holds if and only if (by definition of $\sigma_{\theta}$) 
 $\semtwovlgen{\sigma_{\brt \eq \ell(E) }(E))}_{D,\eta} = \emptyset$. 
\item
$\semtwovlgen{\brt \, \omega \, \any(E)}_{D,\eta} = \false$ if and only if (by definition of $\any$) the disjunction $\bigvee_{t'\in \semtwovl{E}_{D,\eta}} \semtwovlgen{t \, \omega \, t'}_{\eta} = \false$. This holds if and only if (by definition of the connector $\vee$ in 2VL)
for every $\brt'\in \semtwovlgen{E}_{D,\eta}$ it holds that $\semtwovlgen{\brt \, \omega \, \brt'}_{\eta} = \true$. 
By definition of $\eta_{\ell(E)}^{\brt'}$, this holds if and only if for every $\brt'\in \semtwovlgen{E}_{D,\eta}$ it holds that
$\semtwovlgen{\brt \, \omega \, \ell(E)}_{\eta;\eta_{\ell(E)}^{\brt'}} = \false$.  
This holds if and only if (by definition of $\sigma_{\theta}$) 
 $\semtwovlgen{\sigma_{\brt \, \omega\,\ell(E) }(E))}_{D,\eta} = \emptyset$. 
\item
$\semtwovlgen{\brt \, \omega \, \all(E)}_{D,\eta} = \true$ if and only if (by definition of $\all$) the conjunction  $\bigwedge_{\brt'\in \semtwovlgen{E}_{D,\eta}} \semtwovlgen{\brt \, \omega \, \brt'} = \true$. 
This holds if and only if (by definition of the connector $\wedge$ in 2VL)
for every $\brt'\in \semtwovlgen{E}_{D,\eta}$ it hold that $\semtwovlgen{\brt \, \omega \, \brt'} = \true$. 
By definition of $\eta_{\ell(E)}^{\brt'}$, this holds if and only if for every $\brt'\in \semtwovlgen{E}_{D,\eta}$ it holds that
$\semtwovlgen{\brt \, \omega \, \ell(E)}_{\eta;\eta_{\ell(E)}^{\brt'}} = \true$. By definition of $\neg$ in 2VL, this holds  
if and only if for every $\brt'\in \semtwovlgen{E}_{D,\eta}$ it holds that
$\semtwovlgen{\neg(\brt \, \omega \, \ell(E))}_{\eta;\eta_{\ell(E)}^{\brt'}} = \false$.
This holds 
if and only if (by definition of $\sigma_{\theta}$) 
$\semtwovlgen{\sigma_{\neg(\brt \, \omega\,\ell(E)) }(E)}_{D,\eta} = \emptyset$. 
\end{enumerate}
For the second part, we present counter-examples:
\begin{enumerate}
    \item Let $D$ consist of a single relation $R$ with $\ell(R)\df N$ that consists of the single tuple $\NULL$.
    Let us set $E\df R$ and $\theta \df N \eq 1$. By definition of the SQL semantics, it holds that $\semsql{\sigma_{N \eq 1}(R)}_{D,\emptyset} = \emptyset$. However, by definition we have
    $\semsql{R \setminus \sigma_{N \eq 1}(R)}_{D,\emptyset} =  \semsql{R}_{D,\emptyset} \setminus  \semsql{\sigma_{N \eq 1}(R)}_{D,\emptyset} $
and $\semsql{\sigma_{N \eq 1}(R)}_{D,\emptyset} =\emptyset $. Thus,  $\semsql{R \setminus \sigma_{N \eq 1}(R)}_{D,\emptyset} = R^D$ which concludes the proof.
\item
Let us consider the same settings as above, and let us set $\brt \df (1)$.
It holds that $\semsql{1 \in R}_{D,\emptyset} = \unknwn$ but $\semsql{\sigma_{1 \eq N }(R)}_{D,\emptyset} = \emptyset$.
\item
With the same settings as in the previous item, and with $\eq$ in the place of $\omega$, we get a similar contradiction. 
\item
With the same settings as previously, we get $\semsql{\sigma_{\neg(1\eq N) }(R)}_{D,\eta} = \emptyset$ whereas $\semsql{1 \eq \all(R)}_{D,\eta} = \false$.
\end{enumerate}
\end{proof}

\section{Appendix for Section~\ref{sec:coincide}
}
\subsection{Proof of Theorem~\ref{thm:trivtrans}}

\repeatresult{theorem}
{\ref{thm:trivtrans}}
{\thmtrivtrans}

We first present the following lemma that connects nullable attributes and their semantics.
\begin{lemma}\label{lem:nf}
For every \sqlrarec\ expression $E$, database $D$, environment $\eta$, and name $N\in \ell(E)$ and $N\not\in \nulable{E}$ the following holds: 
For every $\bara \in \semsql{E}_{D,\eta}$ it holds that $\semsql{N}_{D,\eta;\eta_{\ell(E)}^{\bara}} \ne \NULL$. 
\end{lemma}
\OMIT{
\begin{corollary}
For every \sqlrarec\ expression $E$, database $D$, environment $\eta$, and null-free subexpression $\sigma_{\theta}(F)$ of $E$  for all $\bara\in \semtwovl{F}$ it holds that
$\semtwovl{\theta}_{D,\eta;\eta_{\ell(F)}^{\bara}} = \semvl{\theta}_{D,\eta;\eta_{\ell(F)}^{\bara}}$
\end{corollary}
}\OMIT{\begin{proof}
We prove this by induction on $E$:
\begin{itemize}
    \item If $E\df R$ then the claim follows from the definition.
    \item
    If $E\df \nulable{\epsilon(E')}$ then $\ell(E) = \ell(\epsilon(E'))$ and for every $\bara\in \epsilon(E')$ it also holds that  $\bara\in E'$. Thus, the claim follows from the induction hypothesis. For $E\df \sigma_{\theta}(E')$ the proof is similar.
    \item If $E\df E_1\times E_2$ then it holds that $\ell(E)\df \ell(E_1)\cdot \ell(E_2)$. 
    Let $N \in \ell(E_1\times E_2) \setminus \nulable{E_1\times E_2}$.
    Since every element in $\semvl{E_1\times E_2}_{D,\eta}$ is of the form $(\bara ,\bar{b})$ with $\bara \in \semvl{E_1}_{D,\eta}$ and $\bar b \in \semvl{E_2}_{D,\eta}$, it holds that 
    $\semvl{N}_{D,\eta;\eta_{\ell(E)}^{(\bara, \bar b)}} \ne \NULL$.
    \item
    \item
    \item
\end{itemize}
\end{proof}}
The proof of this lemma is based on a straightforward induction on $E$ and thus is omitted. 

\OMIT{We can conclude that 
\begin{corollary}
\label{lem:nf}
Let $E$ be an \sqlrarec\ expression, $D$ a database, and $\eta$ environment. If every subexpression of $E$ of the form $\sigma_{\theta}(F)$ is null-free then for every subcondition $\theta$ of $E$ it holds that 
\[
\semvl{\theta}_{D,\eta} = \semtwovl{\theta}
\]
\end{corollary}
}

It suffices to show that Theorem~\ref{thm:trivtrans} holds for \sqlrarec\ expressions $E$ such that for all of their subcoditions of the form $\neg \theta $ it holds that $\theta$ is atomic. 
This is due to the fact that De-Morgan laws hold for Kleene's logic. Thus, whenever nagation appears on a condition that is not atomic we can reformulate the condition so that the condition holds (i.e., negation would appear only before atomic conditions).  

Let $E$ be an \sqlrarec\ expression such that all of its subcondtioins of the form $\neg \theta$ are such that $\theta$ is atomic, and all of its subexpressions of the form $\sigma_{\theta}(F)$ are null-free. 
We prove mutually the claims by induction on $E$ and $\theta$:
\begin{itemize}
    \item[$(a)$]
    For every subexpression $F$ of $E$ it holds that $\semtwovl{F}_{D,\eta} = \semsql{F}_{D,\eta}$.
     \item[$(b)$]
     For every null-free subexpression $\sigma_{\theta}(F)$ of $E$ and for every subcondition $\theta'$ of $\theta$ it holds that
    $\semtwovl{\theta'}_{D,\eta;\eta_{\ell(F)}^{\bara}} =\true$ if and only if $ \semvl{\theta'}_{D,\eta;\eta_{\ell(F)}^{\bara}}= \true$ for all $\bara\in \semtwovl{F}_{D,\eta}$.
\end{itemize}

\paragraph{Induction basis}
\begin{itemize}
    \item[(a)] By definition $\semtwovl{R}_{D,\eta} = R^D =  \semsql{R}_{D,\eta}$
    \item[(b)]
Assume that $\theta$ and $F$ are atomic.
By induction basis (a) it holds that $\semtwovl{F}_{D,\eta} = \semvl{F}_{D,\eta}$. 
Let $\bara\in \semvl{F}_{D,\eta}$. We want to show that $\semtwovl{\theta}_{D,\eta;\eta_{\ell(F)}^{\bara}} = \true$ if and only if $\semvl{\theta}_{D,\eta;\eta_{\ell(F)}^{\bara}} = \true$.
This direction can be proved similarly to the proof of correctness of the translation $\trt$ from $\semvl{}$ to $\semtwovl{}$ or from $\semtwovl{}$ to $\semvl{}$.
\end{itemize}
\paragraph{Induction step}
\begin{enumerate}
    \item[$(a)$] 
    \begin{enumerate}
        \item[1.]
        If $E\df \pi_{t_1,\cdots, t_m} E' \,|\, E_1\op E_2
        \,|\, \epsilon(E) \,|\, \Group_{\bar M,\langle F_1(N_1),\cdots, F_m(N_m)\rangle} (E)
        $ (possible renaming are omitted to simplify the presentation) then the claim follows directly from induction hypothesis (a).
        \item[2.]
         If $E\df \sigma_{\theta}(E')$ then
    by applying induction hypothesis (a) we get that for every $\semvl{E'}_{D,\eta}= 
    \semtwovl{E'}_{D,\eta}$. 
    For all $\bara \in \semtwovl{E'}_{D,\eta}$:
    \begin{itemize}
        \item If $\theta \df \theta_1 \op \theta_2$ then by induction hypothesis (b) we have that 
        $\semtwovl{\theta_1}_{D,\eta; \eta_{\ell(E')}^{\bara}} =\true$  if and only if  $\semvl{\theta_1}_{D,\eta; \eta_{\ell(E')}^{\bara}}=\true$ and similarly for $\theta_2$. Thus, by the truth table of connectors we can conclude that  $\semtwovl{\theta}_{D,\eta; \eta_{\ell(E')}^{\bara}} =\true$  if and only if  $\semvl{\theta}_{D,\eta; \eta_{\ell(E')}^{\bara}}=\true$ which completes the proof of this case. 
        \item If $\theta\df \neg \theta'$ then we do case analysis and show that for all $\bara \in \semtwovl{E'}_{D,\eta}$: 
        \begin{itemize}
         \item If $\theta' \df \true \,|\, \false $ then
 $\semtwovl{\theta'}_{D,\eta'} = \theta'$ for every $\eta'$ and thus the claim holds. 
       \item If $\theta' \df \isnul(t)$ then
     since $\sigma_{\theta}(E')$ is null-free, it holds that  $\semtwovl{\theta'}_{D,\eta;\eta_{\ell(E)}^{\bara}} = \false$ and thus  $\semvl{\theta}_{D,\eta;\eta_{\ell(F)}^{\bara}} = \false$.
    \item If $\theta' \df \brt \,\omega\,\brt'$ then 
    due to Lemma~\ref{lem:nf} and to the fact $\sigma_{\theta}(E')$ is null-free, we get $\semtwovl{\theta'}_{D,\eta;\eta_{\ell(E)}^{\bara}} = \semvl{\theta}_{D,\eta;\eta_{\ell(F)}^{\bara}}$.
    \item If $\theta' \df \brt \in F'$ then by induction hypothesis (a) we get
    $\semtwovl{F'}_{D,\eta; \eta_{\ell(F)}^{\bara}} = \semvl{F'}_{D,\eta; \eta_{\ell(F)}^{\bara}}$. 
    Since $\sigma_{\theta}(E')$ is null-free, and due to Lemma~\ref{lem:nf}, it holds that $\semtwovl{t_i\eq N_i}_{D,\eta; \eta_{\ell(F)}^{\bara}} =  \semvl{t_i\eq N_i}_{D,\eta; \eta_{\ell(F)}^{\bara}}$ where $\ell(F)\df (N_1,\cdots, N_n)$,
    and therefore the claim holds.
    \item If $\theta' \df \brt \,\omega\, \any(F')\,|\,\brt \,\omega\, \all(F') $ then we can show the claim similarly to the previous item.
    \item
    If $\theta \df \isempty(F')$ then the claim is implied from induction hypothesis (a).
        \end{itemize}
    \end{itemize}
    \end{enumerate}
    \item[$(b)$]
    \begin{enumerate}
        \item[1.] By definition,
        $\semsql{\theta_1 \vee \theta_2}_{D,\eta}  \semsql{\theta_1 }_{D,\eta} \vee \semsql{\theta_2}_{D,\eta}$. By induction hypothesis, $\semsql{\theta_1 }_{D,\eta} =\true \text{ if and only if } \semtwovl{\theta_1 }_{D,\eta}=\true$ and $\semsql{\theta_2 }_{D,\eta} =\true \text{ if and only if }  \semtwovl{\theta_2 }_{D,\eta} =\true$. Thus, $\semsql{\theta_1 \vee \theta_2}_{D,\eta} =\true \text{ if and only if }  \semtwovl{\theta_1 \vee \theta_2}_{D,\eta}=\true$.
        \item[2.] 
        By definition,
        $\semsql{\theta_1 \wedge \theta_2}_{D,\eta} = \semsql{\theta_1 }_{D,\eta} \wedge \semsql{\theta_2}_{D,\eta}$. By induction hypothesis, $\semsql{\theta_1 }_{D,\eta}=\true \text{ if and only if } \semtwovl{\theta_1 }_{D,\eta}=\true$ and $\semsql{\theta_2 }_{D,\eta} =\true \text{ if and only if } \semtwovl{\theta_2 }_{D,\eta}= \true$. Thus, $\semsql{\theta_1 \wedge \theta_2}_{D,\eta} = \true \text{ if and only if }  \semtwovl{\theta_1 \wedge \theta_2}_{D,\eta}= \true$.
    \end{enumerate}
\end{enumerate}

\section{Appendix for Section~\ref{sec:mvl}}

\newsavebox{\sqlq}
{\sbox{\sqlq}{%
	\parbox{\columnwidth}{%
		\begin{align*}
		Q_{\theta}&\df
		\sqlkw{SELECT} 	\,\,\alpha	\,\,
		\sqlkw{FROM}	\,\, R	\,\,
		\sqlkw{WHERE}	\,\, \theta \\
		Q_{\mathrm{not\_in}}&\df 
		\begin{multlined}[t]
		\sqlkw{SELECT} 
		\,\,	\mbox{\lstinline{*}}\,\,
		\sqlkw{FROM}\, \, R \,\,
		\sqlkw{WHERE} \,\,\ell(R)\,\, 
		\sqlkw{NOT IN}\\
		(	\sqlkw{SELECT}\,\, \mbox{\lstinline{*}}\,\,
		\sqlkw{FROM}\,\, S)
		\end{multlined}\\
		Q_{\mathrm{empty},\omega}&\df
		\begin{multlined}[t]
		\sqlkw{SELECT} \,\,		\mbox{\lstinline{*}}\,\,
		\sqlkw{FROM}\,\, R\,\,
		\sqlkw{WHERE} \,\,
		\sqlkw{NOT EXISTS}\\
		(	\sqlkw{SELECT}\,\, 		\mbox{\lstinline{*}}
		\,\,\sqlkw{FROM} \,\,S\,\,
		\sqlkw{WHERE}\,\, \ell(R) \,\omega \,\ell(S))
		\end{multlined}\\
		Q_{\mathrm{not}\_\any, \omega}& \df
		\begin{multlined}[t]
		\sqlkw{SELECT} \,\,
		\mbox{\lstinline{*}} \,\,
		\sqlkw{FROM} \,\,R\,\,
		\sqlkw{WHERE NOT}\,\, 
		\ell(R)\, \omega \\
		\sqlkw{ANY}\,\,
		(	\sqlkw{SELECT}\,\, \mbox{\lstinline{*}}\,\,
		\sqlkw{FROM}\,\, S)
		\end{multlined}\\
		Q_{\all,\omega} &\df
		\begin{multlined}[t]
		\sqlkw{SELECT} \,\,
		\mbox{\lstinline{*}}\,\,
		\sqlkw{FROM} \,\,R
		\,\,
		\sqlkw{WHERE} \,\,\ell(R)\, \omega\\ 
		\sqlkw{ALL}
		(	\sqlkw{SELECT}\,\, \mbox{\lstinline{*}}\,\,
		\sqlkw{FROM} \,\,S)
		\end{multlined}\\
		Q_{\mathrm{empty},\mathrm{not}\, \omega}
		&\df\!
		\begin{multlined}[t]
		\sqlkw{SELECT} 	\,\,	\mbox{\lstinline{*}}\,\,
		\sqlkw{FROM} \,\,R\,\,
		\sqlkw{WHERE} \,\,
		\sqlkw{NOT EXISTS}\\
		(	\sqlkw{SELECT} 	\,\,	\mbox{\lstinline{*}}
		\,\,\sqlkw{FROM} \,\,S
		\,\,		\sqlkw{WHERE}
		\\
		\sqlkw{NOT} \,\, \ell(R) \,\omega\, \ell(S) )
		\end{multlined}
		\end{align*}
	}
}
}
\sbox{\sqlq}{%
		\parbox{0.95\columnwidth}{%
\begin{align*}
\!\!\!\!\begin{array}{lr}
\sqlkw{SELECT}\!\!\! &\alpha\\
\sqlkw{FROM}   &R	\\
\sqlkw{WHERE}  &\theta \\
\end{array}
 =
\left(\!\!\begin{array}{lr}
\sqlkw{SELECT} \!\!\!	&\alpha	\\
\sqlkw{FROM}	& R	\\
\end{array}\!\!\right)
\hspace{-1pt}
\begin{array}{c}
\sqlkw{EXCEPT}\\
\sqlkw{ALL}
\end{array}
\hspace{-1pt}
\left(\!\!\begin{array}{lr}
\sqlkw{SELECT}\!\!\! &\alpha \\
\sqlkw{FROM}   &R \\
\sqlkw{WHERE}  &\neg \theta \\
\end{array}\!\!\right)
\end{align*}\!\!\!\!
\begin{align*}
	\begin{aligned}
	&\sqlkw{SELECT} \ \alpha\
	\sqlkw{FROM}\, \, R\\
	&\sqlkw{WHERE}\,\,\beta\,\, 
	\sqlkw{NOT IN}\\
	& \hspace{17pt}
	(\sqlkw{SELECT}\ \gamma\ 
	\sqlkw{FROM}\,\, S)
	\end{aligned}
&\,\,=\,\,
	\begin{aligned}
&\sqlkw{SELECT}\ \alpha\ 
\sqlkw{FROM}\,\, R
\\&\sqlkw{WHERE }
\sqlkw{NOT }
\sqlkw{EXISTS}\\
&\hspace{17pt}(\sqlkw{SELECT}\,\, 		\mbox{\lstinline{*}}
\,\,\sqlkw{FROM} \,\,S\\
&\hspace{17pt}\sqlkw{WHERE}\ \beta=\gamma)
\end{aligned}\\
\\
\begin{aligned}
&\sqlkw{SELECT} \ \alpha\ 
\sqlkw{FROM} \,\,R\\
&\sqlkw{WHERE }
\sqlkw{NOT}\,\, 
\beta\ \omega\
\sqlkw{ANY}\\
&\hspace{17pt} (\sqlkw{SELECT}\ \gamma\ 
\sqlkw{FROM}\,\, S)
\end{aligned}
&\,\,\,\,=\,\,\,\,
\begin{aligned}
&\sqlkw{SELECT}\ \alpha\
\sqlkw{FROM}\,\, R\\
&\sqlkw{WHERE} \,\,
\sqlkw{NOT EXISTS}\\
&\hspace{17pt}(\sqlkw{SELECT}\,\, 		\mbox{\lstinline{*}}
\,\,\sqlkw{FROM} \,\,S\\
&\hspace{17pt}\sqlkw{WHERE}\,\, \beta\ \omega \ \gamma)
\end{aligned}\\
\\
\begin{aligned}
&\sqlkw{SELECT} \ \alpha\ 
\sqlkw{FROM} \,\,R
\\&
\sqlkw{WHERE} \ \beta\ \omega\,
\sqlkw{ALL}\\
&\hspace{17pt}(\sqlkw{SELECT}\ \gamma\
\sqlkw{FROM} \,\,S)
\end{aligned}
&\,\,\,\,=\,\,\,\,
\begin{aligned}
	&\sqlkw{SELECT}\ \alpha\
	\sqlkw{FROM} \,\,R\\
&\sqlkw{WHERE} \,\,
	\sqlkw{NOT EXISTS}\\
	&\hspace{17pt} (\sqlkw{SELECT} 	\,\,	\mbox{\lstinline{*}}
	\,\,\sqlkw{FROM} \,\,S
	\\&	\hspace{17pt}	\sqlkw{WHERE }
	\sqlkw{NOT} \ \beta\ \omega\ \gamma)
\end{aligned}
\end{align*}
		}
	}

\newcommand{\figsqlq}
{\centering
	\begin{figure}[t]
	\fbox{{\usebox{\sqlq}}}
	\caption{	\label{fig:sqlq} SQL queries equivalent under
	the two-valued semantics but not SQL's three-valued
	semantics ($\alpha,\beta,\gamma$ are attribute lists and
	$\omega$ is a comparison)}
	\end{figure}
}
{\figsqlq}

\subsection{Proof of Theorem~\ref{thm:mvl}}

\repeatresult{theorem}{\ref{thm:mvl}}{ \thmmvl}

Before presenting the translations, we make some observations.
Let us denote the truth values of any MVL by $\tau_1, \cdots , \tau_M$. 

\begin{lemma}\label{lem:per}
For any truth value $\tau_i$ there exists integers $l^{\vee}_i,p^{\vee}_i, l^{\wedge}_i, p^{\wedge}_i$ such that 
the following holds for any $j$:
\begin{itemize}
	\item $\underbrace{\tau_i \vee \cdots \vee \tau_i}_{j\text{ times}} = 
	\underbrace{\tau_i \vee \cdots \vee \tau_i}_{l_i+\Big( (j-l^{\vee}_i) \mod (p^{\vee}_i  - l^{\vee}_i)  \Big) 
		\text{ times}}
	$
	\item
	$\underbrace{\tau_i \wedge \cdots \wedge \tau_i}_{j\text{ times}} = 
	\underbrace{\tau_i \wedge \cdots \wedge \tau_i}_{l_i+\Big( (j-l^{\wedge}_i) \mod (p^{\wedge}_i  - l^{\wedge}_i)  \Big) 
		\text{ times}}
	$
\end{itemize}
\end{lemma}
\begin{proof}
	We start by proving the existence of $l_i^{\vee},p_i^{\vee}$.
	For any $i,j$ we denote by $\tau_i^j$ the disjunction 
	\[
	\underbrace{\tau_i \vee \cdots \vee \tau_i}_{j\text{ times}}.
	\]
Since we have finite number of truth values,
 the pigeonhole principle ensures that for any truth value $\tau_i$ there are two natural numbers $l < p $ such that the following holds:
\begin{itemize}
	\item $\tau^{j'} \ne \tau^j$ for any $1\le j' < j < p$, and
	\item
	 $\tau^{l}= \tau^{p}$.
\end{itemize}

Due to associativity of $\vee$, it follows that the claim holds with 
$l_{i}^{\vee} \df l$   and  $p_{i}^{\vee}\df p$.

Similar arguments show the existence of $l_i^{\wedge}$ and $p_i^{\wedge}$.
\end{proof}

\rm 
Since the operator $\vee$ is both associative and commutative,  any disjunction can be written as follows:
\[
 \underbrace{\tau_1 \vee \cdots \vee \tau_1}_{i_1\text{ times}}
 \vee 
 \cdots
 \vee 
  \underbrace{\tau_M \vee \cdots \vee \tau_M}_{i_M\text{ times}}
\]
Alternatively, we can denote disjunctions
 by $f_{\vee}(i_1, \cdots , i_M)$ where each $i_j$ is the number of times the truth value $\tau_j$ appears in the disjunction.
For similar reasons, we can denote any conjunction as $f_{\wedge}(i_1, \cdots , i_M)$ where $i_j$ is the number of times the truth value $\tau_j$ appears in the conjunction.
The previous lemma leads us to the following straightforward corollary.
\begin{corollary}\label{cor:mod}
	for any $n_1,\cdots, n_M$ it holds that:
	\begin{itemize}
		\item 
		$f_{\vee}(n_1,\cdots,n_M) = 
		f_{\vee}(n'_1,\cdots,n'_M)$
		where $n'_{\iota} = 
		l_{\iota}+\Big( 
		(n_{\iota} - l^{\vee}_{\iota}) \mod 
		(p_{\iota}^{\vee}-l^{\vee}_{\iota})
		\Big)
		$, and
		\item
		$f_{\wedge}(n_1,\cdots,n_M)
		=
		f_{\wedge}(n''_1,\cdots,n''_M)$
		where $n''_{\iota} = 
		l_{\iota}+\Big( 
		(n_{\iota} - l^{\wedge}_{\iota}) \mod 
		(p_{\iota}^{\wedge}-l^{\wedge}_{\iota})
		\Big)
		$
	\end{itemize}
	for all $1 \le \iota \le M$.
\end{corollary}

%
%
%
%
%
%
%
%
%
%
%
\paragraph*{From $\semmvl{\,}$ to $\semsql{\,}$}
The translation is presented in Figure~\ref{fig:mvltothree}.

 \newsavebox{\mvltothree}
 \sbox{\mvltothree}{%
 	\parbox{0.5\textwidth}{%
 		\begin{center}
 			\textbf{Basic conditions}\\ $\;$ \\   
 \begin{tabularx}{\textwidth}{rl}
 	$(t\,\omega\, t')^{\tau} \df$& $\theta_{\omega,\tau}(t, t')$
 	\text{where} $\omega\in\{\eq, < ,\le , >,\ge \}$
 	\\
 \OMIT{	$(\brt \eq \brt')^{\tau} \df$& $\bigvee_{\tau_1,\cdots,\tau_n 
 		:\, \tau_1 \wedge\cdots\wedge\tau_n = \tau} 
 	\Big(
 	\wedge_{i=1}^n
 	\theta_{\eq,\tau_i}(t_i, t'_i)\Big)$
 	\\}
 \OMIT{	\\
 \text{ where}	&$\brt \df (t_1,\cdots, t_n), \,\,\brt' \df (t'_1,\cdots, t'_n) $\\
 	$\left(P(\brt)\right)^{\tau} \df$& $\theta_{P,\tau}(\brt)$
 	}\\
 	$(\brt \in E)^{\tau} \df$&
 	$ (f_{\vee}(n'_1,\cdots, n'_k) \eq \tau) $ \text{ where }
 	\\
 	$n'_{\iota} \df $ & $l_{\iota}+\Big( 
 	(n_{\iota} - l^{\vee}_{\iota}) \mod 
 	(p_{\iota}^{\vee}-l^{\vee}_{\iota})
 	\Big) $\text{ where } 
 	\\
 	$n_i \df $&
 	$\Group_{\emptyset}\langle \Count(\textasteriskcentered)\rangle \Big(\sigma_{\theta_i = \true}(E_i)\Big)$ \text{ where } 
 	\\
 	$E_i \df $&
 	$\pi_{\theta_i(\bar N) \to\theta_i} (F)$, $\ell(F) \df \bar N$, $\theta_i(\bar N) \df (\bar N \eq \brt )^{\tau_i}$
 	\\
 	$(t \,\omega \, \any(E))^{\tau} \df$&
 	$ (f_{\vee}(n'_1,\cdots, n'_k) \eq \tau) $
 	\\
 	$(t \,\omega \, \all(E))^{\tau} \df$&
 	$ (f_{\wedge}(n''_1,\cdots, n''_k) \eq \tau) $ \text{ where } 
 	\\
 	$n'_{\iota} \df $ & $l_{\iota}+\Big( 
 	(n_{\iota} - l^{\vee}_{\iota}) \mod 
 	(p_{\iota}^{\vee}-l^{\vee}_{\iota})
 	\Big) $, \\
 	$n''_{\iota} \df $ & $l_{\iota}+\Big( 
 	(n_{\iota} - l^{\wedge}_{\iota}) \mod 
 	(p_{\iota}^{\wedge}-l^{\wedge}_{\iota})
 	\Big) $,
 	\\
 	$n_i \df $&
 	$\Group_{\emptyset}\langle \Count(\textasteriskcentered)\rangle \Big(\sigma_{ (t\,\omega\,\ell(F))\eq \tau_i}(F)\Big)$,
 \end{tabularx}
 			$\;$ \\ 
 			\textbf{Composite conditions}\\ $\;$ \\   
 			\begin{tabularx}{\textwidth}{rl}
 				$(\theta_1\vee \theta_2)^{\tau} \df$
 				&
 				$\bigvee_{\tau_1,\tau_2:\tau_1\vee\tau_2 = \tau}\Big((\theta_1)^{\tau_1}\vee (\theta_2)^{\tau_2} \Big)
 				$
 				\\
 				$(\theta_1\wedge \theta_2)^{\tau} \df$
 				&
 				$\bigvee_{\tau_1,\tau_2:\tau_1\wedge\tau_2 = \tau}\Big((\theta_1)^{\tau_1}\wedge (\theta_2)^{\tau_2} \Big)
 				$\\
 				$(\neg \theta)^{\tau} \df$
 				&
 				$\bigvee_{\tau':\neg\tau' = \tau}
 				(\theta)^{\tau'}
 				$
 			\end{tabularx}
 		\end{center}
 	}%
 }%

 \newcommand{\mvltothreefig}{%
 	\begin{figure}[t]
 		\centering
 		\fbox{\usebox{\mvltothree}}
 		\caption{\mbox{$\semmvl{\,}$}  to SQL semantics
 		}
 		\label{fig:mvltothree}
 	\end{figure}
 }%
 \mvltothreefig

We note that the existence of the conditions $\theta_{\omega, \tau}$ and $\theta_{P,\tau}$ is a consequence of expressibility.

The proof structure is similar to other translations, and we use similar notations. Note that the following is straightforward from definition:
\begin{itemize}
\item 
$\semmvl{t \, \omega \, t'}_{D,\eta} = \tau$ if and only if 
$\semsql{\theta_{\omega,\tau}(t,t')}_{D,\eta} = \true$.
\item
$\semmvl{P(\brt)}_{D,\eta} = \tau$ if and only if
$\semsql{\theta_{P,\tau}(\brt)}_{D,\eta} = \true$.
\end{itemize}

Before presenting the translations for $\brt \,\omega\, \any(E)$ and $\brt \,\omega\, \all(E)$, we explain it intuitively. We note that even though the semantics of these
conditions are defined as a disjunction, we cannot use the translation of composite conditions since the disjunction’s size (and elements) depend on the actual semantics of $E$. 
We thus compute the disjunction by using the ability of \sqlra\ to count, and the fact the order of elements in the disjunction is not significant.
We use the following observation that is straightforward from the definition of the semantics of \sqlra.

\begin{proposition}\label{prop:mvl}
For every database $D$ and environment $\eta$ it holds that the integer $n_i \df \semmvl{\Group_{\emptyset}\langle \Count(\textasteriskcentered)\rangle \Big(\sigma_{ (t\,\omega\,\ell(F))\eq \tau_i}(F)\Big)}_{D,\eta}$ is the number of tuples $\brt'\in\semmvl{F}_{D,\eta}$ for which $\semmvl{\brt \,\omega\, \brt' }_{D,\eta} = \tau_i$.
\end{proposition}

With this, we can move to the proof:
\begin{enumerate}
    \item 
    It holds that 
    \[
    \semmvl{\brt\, \omega \, \any (E)}_{D,\eta} = \tau 
    \] if and only if (by definition)
    \[
    \bigvee_{\brt'\in \semmvl{E}_{D,\eta}} \semmvl{\brt \,\omega\, \brt'}_{D,\eta} = \tau
    \]
    if and only if (by induction hypothesis)
        \[
    \bigvee_{\brt'\in \semvl{F}_{D,\eta}} \semmvl{\brt \,\omega\, \brt'}_{D,\eta} = \tau
    \]
    if and only if (due to associativity and transitivity) 
    \[
    f_{\vee}(m_1,\cdots, m_M) = \tau
    \]
    where $m_i$ is the number of tuples $\brt'\in \semvl{F}_{D,\eta}$ for which $\semmvl{\brt \,\omega\, \brt'}_{D,\eta} = \tau_i$, if and only if (by Corollary~\ref{cor:mod} and Proposition~\ref{prop:mvl})
    \[
    \semvl{
    f_{\vee}(n'_1,\cdots, n'_M) \eq \tau} = \true
    \]
    if and only if (by definition)
 \[
    \semvl{{\brt\, \omega \, \any (E)}^{\tau}}_{D,\eta} = \true 
    \]
    \item
    The proof for $(\brt\,\omega\,\all(E))^{\tau}$ can be shown similarly. 
\end{enumerate}

 For composite conditions the proof follows directly. 

\paragraph*{From $\semsql{\,}$ to $\semmvl{\,}$} 
This direction uses the same translation as from $\semsql{\,}$ to $\semtwovl{\,}$. The correctness follows directly from the corresponding proof and from the fact that the many-valued connectors are extensions of the corresponding two-valued ones.

\rm 

\section{Appendix: Details of the User Survey}
In this section we present the survey design. Other than asking respondents about classifying themselves as practitioners vs academics, the survey consists of 10 questions. 

\begin{itemize}
    \item Two questions look at queries where 3VL and 2VL outputs coincide. One of them is a query that finds the intersection of relations $R$ and $S$ by means of a \sqlkw{IN} subquery, and the other one looks for tuples in a relation where the value of two attributes is the same, while one attribute is \NULL. Here users are presented with the SQL output and an alternative and they overwhelmingly choose SQL.
    \item We then change the \sqlkw{IN} condition to \sqlkw{NOT IN} and give users the query \ref{sql:in} from the Introduction, as well as a database with a nonempty relation $R$ and $S=\{\NULL\}$. Output options are $R$ under 2VL and $\emptyset$ under 3VL; users strong preference is for 2VL. Specifically, 82\% opt for 2VL, 13\% for 3VL, and 5\% do not have a preference. 
    \item In the next question equality between attributes is changed for inequality, i.e., we are asking for tuples $(a,b)\in T$ with $\neg(a=b)$. If one component is \NULL, under 2VL such tuples will be selected but under 3VL they will not. Then 62\% of users prefer 2VL, while 22\% prefer 3VL and 15\% do not know. 
    \item Our next question concerns the behavior of \sqlkw{ANY} in the presence of nulls. We give the following query
    \begin{sql}
    SELECT R.A FROM R 
    WHERE NOT (R.A > ANY (SELECT S.A FROM S))
    \end{sql}
    where again $S=\{\NULL\}$ and offer as answers $R$ itself (returned  under 2VL) and $\emptyset$ returned by 3VL. The preferences are: 58\% for 2VL, 33\% for 3VL, and 9\% do not know.  
    \item We then switch to equivalences of queries that encode those mentioned in Proposition \ref{prop:restequiv}. These are valid under 2VL but not valid 3VL. The first pair of queries are queries (\ref{sql:in}) and (\ref{sql:exists}). By the 66\% to 44\% vote the users prefer 2VL.
    \item The next pair of queries is 
    \begin{sql}
    SELECT * FROM R
    \end{sql}
    and 
    \begin{sql}
    SELECT * FROM R WHERE R.A=1 OR R.A <> 1
    \end{sql}
    The respondents by the vote of 58\% to 42\% prefer them to be equivalent, i.e., follow 2VL. 
    \item The final pair of queries is 
    \begin{sql}
    SELECT R.A FROM R 
    WHERE NOT (R.A > ANY (SELECT S.A FROM S))
    \end{sql}
    and 
    \begin{sql}
    SELECT R.A FROM R 
    WHERE NOT EXISTS
       (SELECT * FROM S WHERE R.A > S.A)
      \end{sql}
      These queries are equivalent under 2VL, and this is the behavior favored by 58\% of respondents (with no don't-knows).
     \item As a final question, we have relations $R(A)$ and $S(A)$ with $S=\{1\}$ and $R=\{1,\NULL\}$ where $R.A$ is a foreign key referencing $R$. Under 2VL, the foreign key constraint will not be valid, but in SQL's 3VL it is. This is the most confusing example, with 31\% preferring 2VL, 38\% preferring 3VL, and 31\% not knowing whether the constraint should hold.  
\end{itemize}

\end{document}